\documentclass[11pt,a4paper]{article}

\usepackage{amsthm,amsmath,multirow,graphicx,subfigure,booktabs,url,mathtools,wrapfig,mathrsfs,bm,relsize,caption, enumitem, fullpage, xcolor,algorithm}

\usepackage[noend]{algpseudocode}

\RequirePackage[authoryear]{natbib}
\RequirePackage[psamsfonts]{amssymb}
\usepackage[colorlinks=true,
linkcolor=blue,
urlcolor=blue,
citecolor=blue]{hyperref}

\numberwithin{equation}{section}
\newtheorem{conj}{Conjecture}
\newtheorem{thm}[conj]{Theorem}
\newtheorem{cor}[conj]{Corollary}
\newtheorem{prop}[conj]{Proposition}
\newtheorem{lemma}[conj]{Lemma}
\newtheorem{ass}{Assumption}

\newtheorem{definition}{Definition}

\newtheorem{innercustomgeneric}{\customgenericname}
\providecommand{\customgenericname}{}
\newcommand{\newcustomtheorem}[2]{%
	\newenvironment{#1}[1] 
	{%
		\renewcommand\customgenericname{#2}%
		\renewcommand\theinnercustomgeneric{##1}%
		\innercustomgeneric
	}
	{\endinnercustomgeneric}
}

\newcustomtheorem{customAss}{Assumption}

\theoremstyle{definition}
\newtheorem{remark}{Remark}
\newtheorem{example}{Example}

\def\e{\mathbf{e}}

\def\PP{\mathbb{P}}
\def\EE{\mathbb{E}}
\def\RR{\mathbb{R}}

\def\I{\mathcal{I}}
\def\E{\mathcal{E}}
\def\H{\mathcal{H}}
\def\S{\mathcal{S}}
\def\J{\mathcal{J}}
\def\P{\mathcal{P}}

\def\wh{\widehat}
\def\wt{\widetilde}

\def\0{\mathbf{0}}
\def\eps{\varepsilon}
\def\bI{\bm{I}}
\def\C{\Sigma_Z}
\def\whC{\wh \Sigma_Z}

\def\T{\top}
\def\ul{\underline}
\def\ol{\overline}

\def\diag{\textrm{diag}}   
\def\rank{\textrm{rank}}
\def\Cov{\textrm{Cov}}   
\def\sgn{\textrm{sgn}}
\def\op{{\rm op}}

\def\X{\bm X} 
\def\Z{\bm Z} 
\def\bE{\bm E} 

\def\i{\infty}
\def\nij{\setminus\{i,j\}}
\let\emptyset\varnothing

\def\group{\overset{{\rm I}}{\sim}}
\def\groupH{\overset{{\rm H}}{\sim}}
\def\groupS{\overset{{\rm S}}{\sim}}

\newcommand{\sbt}{\,\begin{picture}(-1,1)(-0.5,-2)\circle*{2.3}\end{picture}\ }

\renewcommand{\parallel}{\mathrel{/\mkern-5mu/}}
\makeatletter
\newcommand{\notparallel}{%
	\mathrel{\mathpalette\not@parallel\relax}%
}
\newcommand{\not@parallel}[2]{%
	\ooalign{\reflectbox{$\m@th#1\smallsetminus$}\cr\hfil$\m@th#1\parallel$\cr}%
}
\makeatother

\def\mike{\color{cyan}}

\begin{document}
    \title{Detecting approximate replicate components of a high-dimensional  random vector with latent structure}
    \author{Xin Bing\thanks{Department of Statistics and Data Science, Cornell University, Ithaca, NY. E-mail: \texttt{xb43@cornell.edu}.}~~~~~Florentina Bunea\thanks{Department of Statistics and Data Science, Cornell University, Ithaca, NY. E-mail: \texttt{fb238@cornell.edu}.}~~~~~Marten Wegkamp\thanks{Departments of Mathematics, and   of Statistics and Data Science, Cornell University, Ithaca, NY. E-mail: \texttt{mhw73@cornell.edu}.} }
	\date{}
	\maketitle

	\begin{abstract}
	
     High-dimensional feature vectors are likely to contain sets of measurements that are approximate replicates of one another. In complex applications, or automated data collection, these feature sets   are not known a priori, and need to be determined. 
     
     
This work proposes  a class of latent factor models  on the observed,  high-dimensional,  random vector $X \in \RR^p$, for defining, identifying and estimating  the index set of its approximately replicate components. The model class is parametrized  by a $p \times K$ loading matrix $A$ that contains a hidden sub-matrix whose 
rows can be partitioned into  groups of parallel vectors. Under this model class, a set of approximate replicate components of $X$ corresponds to a set of parallel rows in $A$: these entries of $X$ are, up to scale and additive error, the same linear combination of the $K$ latent factors; the value of $K$ is itself unknown.  

The problem of finding approximate replicates in $X$ reduces to identifying, and estimating, the location of the hidden sub-matrix within $A$, and of the  partition $\mathcal{H}$ of its row index set $H$. Both $H$ and $\mathcal{H}$ can be fully characterized in terms of a new family of criteria based on the correlation matrix of $X$, and their identifiability, as well as that of the unknown latent dimension $K$, are  obtained as consequences.   The constructive nature of the identifiability arguments enables computationally efficient procedures, with consistency guarantees. 

Furthermore, when the loading matrix $A$ has a particular sparse structure, provided by the errors-in-variable parametrization, the difficulty of the problem is elevated.  The task becomes that of separating out groups of parallel rows that are proportional to canonical basis vectors from other, possibly dense, parallel rows in $A$. This is  met under a scale assumption, via a principled way of selecting the target row indices,   guided by the succesive maximization of Schur complements of appropriate covariance matrices. The resulting procedure is an enhanced version of that developed for recovering general parallel rows in $A$, is also computationally efficient, and consistent.

	\end{abstract}


	\section{Introduction}

    Latent factor models are simple, ubiquitous,  tools for describing data generating mechanisms that yield random vectors $X \in \RR^p$ with possibly very correlated entries, and subsequently approximately low rank covariance matrix. 	The history of factor analysis can be traced back to the 1940s \citep{L40,L41,L43, J67,J69,J70,J77}, with foundational work established by \cite{anderson1956}, and a wealth of recent works motivated by applications to economy and finance, educational testing and psychology, forecasting, biology, to give a limited number of examples.

    A factor model assumes the existence of  an  integer $ 1\le  K \le  p$  and of a random vector $Z \in \RR^K$ of unobservable, latent, factors,  such that the observed $X \in \RR^p$ 
    has the representation 	\begin{equation}\label{model}
    	X = AZ + E, 
    	\end{equation} for some real-valued loading matrix $A \in \RR^{p \times K}$ and random noise $E \in \RR^p$ uncorrelated with $Z$. The corresponding covariance matrix of $X$  has the expression $\Sigma = A \C A^\top  + \Sigma_E$, 
    where $\C$ is the covariance matrix of $Z$ and $\Sigma_E$ that of $E$. A large amount of literature has been and continues to be devoted to  the estimation of approximately low rank covariance matrices corresponding to factor models, for instance, \cite{Chandrasekaran,Hsu2011,rpca,agarwal2012,FAN2008186,fan2011,fan2018,FAN20195,donoho2018}, to name a few. 
    
A related, but different,  line of research is devoted to the estimation of the loading matrix $A$ itself, in identifiable factor models that place structure on $A$ \citep{anderson1956,ICA,NIPS2003_2463,bai2012,fan2011,anandkumar12,NIPS2012_4637,arora2013practical,anandkumar14b,fan2018,jaffe2018learning,bing2019inference,bing2017adaptive,Top,sparseTop} 

We treat a problem of an intermediate nature in this work: we also model the dependency between the components of a high dimensional random vector $X \in \RR^p$ via latent factors, as in the covariance estimation literature, and place structure on $A$, as in the literature devoted to loading matrix estimation, but our focus is different. We study factor models relative to matrices $A$ that are allowed to have groups  of parallel rows, a structure that in general  is not sufficient  for identifying $A$.  The indices of rows that are respectively parallel form a partition $\mathcal{H}$ of the collection $H$ of all indices of parallel rows in $A$.  The components 
$X_j$ of $X$, with $j$ in a group of $\mathcal{H}$ are, up to scale and additive error, the same linear combination of the background latent factors. This parametrization of $A$ thus  provides a way of modeling those components of $X$ that are very highly dependent, in that they are  ``approximate replicates'' of each other, while allowing for general dependency between the other entries in $X$, albeit modeled via a factor model. 

Very high dimensional random vectors  do typically have entries  that are approximately redundant, and we give only a few  motivating examples from biology. For instance, in human systems immunology, in addition to measuring serological features (cytokines, chemokines, antibody titers etc) that are highly correlated with each other, one often measures the same  feature under slightly different technical conditions (e.g. same titer measured at different serum dilutions, same cytokines measured using different technical platforms). In genetic perturbation experiments, one could quantify the effect of the exact same genetic perturbation using different reporter assays. Although some of these redundancies may be obvious based on the experimental design (technical redundancies), others are known to be induced by latent, underlying biological mechanisms, but it is unknown which of the collected measurements reflect them. We address the latter problem in this work. 

We study factor models on high dimensional vectors $X \in \RR^p$ that contain approximate replicate components, in {\it unknown} positions. The focus of our work is in determining their locations, a problem that reduces to that of  identifying and estimating the location of a hidden sub-matrix of $A$, with unknown row index set $H$, and of  unknown  partition $\mathcal{H}$. Since the entirety of a  matrix $A$ with such structure is typically  not  identifiable, we also study an instance of it, provided by an  added sparsity constraint, under which both a hidden sparse sub-matrix of $A$, and $A$ itself are identifiable.   

 The following section provides a detailed summary of our approach and results.

	\subsection{Our contributions}\label{contribution}  

	To state our results, we will assume that $X$ follows a factor model (\ref{model}) with 
	$A\in \RR^{p\times K}$ and $\rank(A)=K$,
	\begin{equation*}
		    \EE[E]=\0 \qquad \text{ and }\qquad  \Sigma_E= \EE [ E E^\top] = \text{diag}( \sigma_1^2,\ldots, \sigma_p^2)>0
		\end{equation*}
 and  $Z$ is standardized such that 
	\[ \EE[Z]=\0 \quad\text{and}\quad
	\diag(\C) = {\bm 1}\quad 
	\text{with\quad $\C:= \EE[ZZ^\top]$}  \quad \text{and}  \quad \rank(\C) = K.
	\]
	We make the following  assumption, which we will show later identifies $K$. It is based on the intuition that for $p\gg K$, the matrix  $A$ is expected to contain many parallel rows. 
	 \begin{ass}\label{ass_parallel}
      The index  set of parallel rows
      \[ H:= \left\{i\in \{1,\ldots,p\}:\ 
      A_{i\sbt} \parallel A_{j\sbt} \ 
      \text{ for some } j\in \{1,\ldots,p\}\setminus \{i\}\right\}
      \] of $A$ is non-empty and
       $\rank( A_{H\sbt})=K$.  
    \end{ass}
    
The  partition of $H$ consists of   disjoint sets $H_1, \ldots, H_G$, and all indices $i,j\in H_k$ correspond to  parallel rows $A_{i\sbt}$ and $A_{j\sbt}$, for each $k\in [G] := \{1,2,\ldots, G\}$. The assumption $H\ne \emptyset$ and $\rank( A_{H\sbt})=K$ implies that $G\ge K$, and  $|H_k|\ge 2$ for each $k\in [G]$. While we are not aware of a study of factor models under Assumption \ref{ass_parallel}, we mention that it is a generalization of the {\it errors-in-variables parametrization}, see for instance \cite{yalcin2001}, which we discuss in more detail below. 

 With the  parametrization of $A$ provided by Assumption \ref{ass_parallel},  we define {\it approximate replicate measurements}  relative to the parallel rows in $A$  as the groups of variables $X_j$, with $j \in H_k$ and $k \in [G]$: they are, up to scale, the same linear combination of the $K$ background factors, up to the  additive measurement error term  $E_j$. The problem of detecting approximate replicate features $X_j$ reduces to that of recovering the parallel rows of $A$, and their partition.

    In the context of this problem and modelling assumptions, we give below  the organization of the paper, section by section, and summarize our contributions.

   \paragraph{Section \ref{sec_parallel}: Identification and recovery of approximately replicate features.}
    The results of our Section \ref{sec_parallel} can be summarized as follows.\\ 
    
    \noindent{\it  1. A new score function for an  if and only if characterization of the partition of parallel rows of a loading matrix.} We show in Section \ref{sec_ident} that Assumption \ref{ass_parallel} is sufficient for the unique, and constructive,  determination of $H$, its partition $\H = \{H_1, \ldots, H_G\}$,  and $K$ from the correlation matrix $R=[ {\rm Corr}(X_i,X_j)]_{i,j\in [p]}$. In Proposition \ref{prop_I}, we prove the non-trivial fact that the parallelism between rows of $A$ is preserved by appropriately modified rows in $R$.  To be  precise,  we show  that the vectors $R_{i \sbt \setminus \{i,j\}} $ and $R_{j \sbt\setminus \{i,j\}}$, each defined by leaving out, respectively,  the $i$th and $j$th entries from the rows $R_{i \sbt}$ and $R_{j \sbt}$,  are parallel in $\RR^{p-2}$ if and only if $A_{i\sbt}$ and $A_{j\sbt}$ are parallel. 
    This realization, combined with the fact that two non-zero vectors  $v,w$ are parallel if and only if 
    	\[\min_{\|(\alpha,\beta)\|_r = 1}\|\alpha  v + \beta w\|_q = 0 \]
    	for any $0<r\le \infty$ and $0\le q\le \infty$, 
    naturally	leads to defining the class of criteria 
	\begin{equation}\label{score0}
	S_{q,r}(i,j) = 
	(p-2)^{-1/q} 
	\min_{\|(\alpha,\beta)\|_r=1} \left\|\alpha R_{i\sbt \setminus\{i,j\}} + \beta R_{j \sbt\setminus\{i,j\}}\right\|_q, \ \text{for any} \ i,j \in [p],
	\end{equation}
for identifying  parallel vectors in $A$ via the model-independent correlation matrix $R$. 

In Proposition \ref{prop_score}, we establish the following  characterization of both $H$\ and its partition $\H$,
	\begin{equation}\label{iff_I}
	    S_{q,r}(i,j) = 0 \quad  \iff \quad  i,j\in H_a\quad  \text{for some } a\in [G]
	\end{equation}
	for any $0<r\le \i$ and $0\le q\le \i$. The new  criteria (\ref{score0}) are indexed by two parameters $(q,r)$. Proposition \ref{prop_score}  shows that $r=\i$ is optimal. In practice, we prefer $(q,r)=(2,\i)$ since $S_{2,\i}$ can be written in closed form,  as proved in Proposition \ref{prop_score_2}. In Theorem \ref{thm_partial_ident}, we prove that Assumption \ref{ass_parallel} identifies not only $H$ and $\H$, but  the dimension $K$ as well.

	 To the best of our knowledge, the only criterion similar in spirit to   our proposed (\ref{score0}) is that in \cite{bunea2020},  introduced in the much more restricted setting of a factor model in which the matrix $A$ has only $0/1$ entries, each row is a canonical basis vector $\e_k \in \RR^K$, and the $p$ rows can be partitioned in $K$ groups of replicates of $\e_k$, for each $k \in [K]$. Thus, in our notation, $H = [p]$ and $G = K$. In this mathematically simpler setting, it suffices to
     compute  the supremum-norm
    \[\|
            \Sigma_{i\sbt\nij} - \Sigma_{j\sbt\nij}
        \|_{\i}\]
    of the differences
    between rows  $\Sigma_{i\sbt\nij} $ and $ \Sigma_{j\sbt\nij}$ of the covariance matrix $\Sigma=\Cov(X)$, for all pairs $i,j\in [p]$. When $A$ has general real-valued entries, this criterion no longer discrimates between general parallel rows, which we prove is made possible by an additional minimization over $(\alpha,\beta)$ on the set $\{( \alpha, \beta)\in\RR^2:\|(\alpha, \beta)\|_r=1\}$. Furthermore, our proposed class of criteria is based on the scale invariant correlation matrix $R$. \\

    \noindent{\it 2. A new method for estimating the approximate replicate index set, its partition, and the latent dimension.} As the identifiability proofs of $H$ and $K$ are constructive, they lead to a practical estimation procedure, stated in Section \ref{sec_est_H}, that is easy to implement, even for large $p$. 
    Section \ref{sec_est_H} first introduces  the empirical counterpart $\wh S_{q,r}$ of the criterion $S_{q,r}$ in (\ref{score0}), by simply replacing $R$ by the empirical correlation matrix $\wh R$.
From the tight, in probability, bound  
\[\max_{i,j}| \wh S_{q,r}(i,j) -  S_{q,r}(i,j)|\le 2\delta_n\] in Theorem \ref{thm_score_ell_q} of Section \ref{sec_theory_H_K}, with
$\delta_n= O(\sqrt{\log(p\vee n)/ n})$, in conjunction with the characterization of $H$ in (\ref{iff_I}), 
we estimate $H$ by the set $\wh H$ of all pairs $(i,j)$ with $\wh S_{q,r}(i,j)\le 2\delta_n$. In this paper, we derive the order of magnitude of $\delta_n$ under the assumption that $X$ is sub-Gaussian. 

In Corollary \ref{cor_H} of Section \ref{sec_theory_H_K} we show that $\wh H$ and its partition $\wh \H$ consistently estimates $H$ and $\H$, respectively. After estimating $H$ and $\H$, Section \ref{sec_est_H} devises the estimation of $K$, exploiting the fact that $\Sigma -\Sigma_E = A\C A^\top$ has rank $K$. 
Proposition \ref{prop_K_parallel} in Section \ref{sec_theory_H_K} shows that this procedure is consistent under mild regularity conditions.


\paragraph{Section \ref{sec_overlap}: Identification and recovery under a canonical parametrization.}
While in Section \ref{sec_parallel} we studied the recovery of generic  approximate replicates, in this section we shift focus to replicates generated in a particular way, and motivate our interest in this problem below. \\

\noindent {\it 1. Pure variables with arbitrary loadings.} A particular instance of Assumption \ref{ass_parallel} is 
  \begin{ass}\label{ass_parallel_prime}
      There exists a subset $I\subseteq H$ such that $A_{I\sbt}$ (up to row permutation) contains at least two $K\times K$ diagonal matrices with non-zero diagonal entries,
    \end{ass}	

\noindent to  which we refer in the sequel as a {\it canonical parametrization}.
This  assumption can be re-stated, equivalently as 
   \begin{customAss}{2$^\prime$}\label{ass_I}
		For any $k\in [K]$, there exist  at least two $i,j\in H$ with $i\ne j$ such that $A_{ik}\ne 0$, $A_{jk}\ne 0$ and $A_{ik'} = A_{jk'} = 0$ for all $k'\in [K]\setminus \{ k\}$. 
  \end{customAss}	
\noindent The collection of the set of indices with existence postulated by  Assumption \ref{ass_I} is the set $I$, defined in Assumption \ref{ass_parallel_prime}. We let $\I = \{I_1,\ldots, I_K\}$  denote its partition.

 One arrives at this sparse parametrization  of $A$ from  both mathematical and applied  statistics perspectives. It has been long understood, see for instance  \cite{anderson1956,bai2012}, that  a particular version of Assumption \ref{ass_parallel_prime} determines $A$ uniquely.  When $I$ and $\Sigma_E$ are known, it is sufficient to assume the existence of only one diagonal sub-matrix in $A$, whereas the existence of a duplicate is a sufficient identifiability condition, when neither $I$ nor $\Sigma_E$ are  known \citep{bing2017adaptive}. 

From an applied perspective, the interest in this parametrization can be most easily seen from its equivalent formulation, Assumption \ref{ass_I}. In \cite{yalcin2001}, it is called the  {\em errors-in-variables parametrization}.  It is popular in the social sciences literature \citep{koopmans1950,mcdonaldbook}
   and 
  a further, large,  array of applications   are given in the review paper \cite{yalcin2001}.
  		 Versions of this factor model class are  routinely used in educational and psychological testing, where the latent variables are viewed as aptitudes or psychological states \citep{thurstone,anderson1956,bollen1989}. The components of  $X$ are test results, with some tests  
	specifically designed to measure {\it only a single}  aptitude $Z_k$, for each aptitude, whereas others test mixtures of aptitudes.  
	
	This brings into focus a particular type of replicate measurements, sparse replicates, which satisfy $X_j = A_{jk}Z_k + E_j$, for all $j \in I_k$, and each $1 \leq k \leq K$. Following the  terminology in network analysis and previous work on factor models we refer to $X_j$ with $j \in I$ as {\it pure variables.} By experimental design, the index set $I$ and the dimension $K$ {\it are  known} in  the  classical literature mentioned above. 
	
	In modern applications, when $p$ is very large, as in genetic applications, or when  many features are automatically collected, neither $I$ nor $K$ are known in advance, and are not identifiable without further modeling assumptions. We show in Section \ref{sec_overlap} that  Assumption \ref{ass_parallel_prime}, or equivalently, Assumption \ref{ass_I},  is  sufficient for this task. 
	

While we will refer to this assumption as the pure-variable parametrization,  we note that, importantly, pure variables $X_i$ and $X_j$  connected to the same latent factor $Z_k$, with indices  collected in the set $I_k$, are allowed different factor loadings, in that $A_{ik}$ and $A_{jk}$  can be different. This is in line with assumptions made in topic models \citep{arora2013practical,Top}, but has not yet been extended, to the best of our knowledge,  to latent factor models on $X \in \RR^p$,  corresponding to $Z \in \RR^K$ and a matrix $A$  with arbitrary real values.

More generally, latent factor models for arbitrary random vectors $X \in \RR^p$,  under the parametrization on $A$ given by  Assumption \ref{ass_I}, with  $I$ {\it unknown},  have not been studied, to the best of our knowledge. The work of \cite{bing2017adaptive} uses  a restricted version 
 of  Assumption \ref{ass_I}, motivated by biological applications, in which $|A_{ik}| = |A_{jk}| = 1$.  Their entire estimation procedure of $I$ and $A$  is crucially tailored to a model in which all pure variables, in all groups, have the same loading, by convention taken to be equal to one,  and it cannot be generalized to the model under the parametrization considered here. The procedure we propose, and therefore its analysis,  are new, and entirely  different from existing work.
 However, the newly proposed  procedure of estimating $I$ can be employed for determining the centers of latent overlapping clusters, using the rationale in   \cite{bing2017adaptive}, but adding the flexibility provided by Assumption \ref{ass_I}.\\

    \noindent{\it 2. Identifying the pure variable index set when the loading matrix  has additional, non-pure, parallel, rows.}  Allowing for different loadings of the pure variables brings challenges in establishing the identifiability of $I$ and therefore of $A$, especially when we allow for the existence of other parallel, but with  arbitrary entries, rows in $A$. 
	In Section \ref{sec_ident_full}, we formally establish the identifiability of $I$ and of its partition $\I = \{I_1,\ldots, I_K\}$ under Assumption \ref{ass_I}, and an additional assumption that we will discuss shortly.
	
	Based on the observation that $I\subseteq H$, the first step towards identifying $I$ finds the set of parallel rows  $H$ and its partition $\H$. 
	
	When there are no parallel rows of $A$ corresponding to non-pure variables, $H$ and $\H$ reduce to $I$ and $\I$, respectively and the results of Section \ref{sec_parallel} apply directly.  Corollary \ref{thm_ident_I} of Section \ref{sec_ident_full} summarizes the identifiability of $I$ and $\I$ in this scenario. 
	
	However, if there exist non-pure variables  corresponding to rows in $A$ that are parallel, it turns out that the set $I$ is not identifiable:  the index set $J_1$ corresponding to these non-pure variables is also included in  $H$, and $H = I \cup J_1$.
	
	Separating $I$ from $J_1$ reduces to selecting $K$ distinct indices from $H$, and proving that they correspond to pure variables.  One has liberty in selecting these indices. We opted for selecting  those that correspond to variables that  contain as much information as possible. 
	 A systematic way of selecting  $K$ representive variables  is given in Lemma \ref{lem_find_L} of Section \ref{sec_ident_full} based on successively maximizing  certain Schur complements of the low rank matrix   $\Cov(AZ) = A \C A^\T$. These quantities are equivalent with conditional variances,  when $Z$ follows a Gaussian distribution. Whereas this selection process may be of interest in its own right, we  guarantee that its output  is indeed a set of pure variables under the additional  Assumption \ref{ass_J_prime}, stated in Section \ref{sec_ident_full}. It is a scaling assumption,  that compares the loading $|A_{jk}|$, for each $j\in J_1$ and $k\in [K]$, with $\max_{i\in I_k} \| A_{i\sbt}\|_1$, the largest loading of pure variables in $I_k$.

    We formally prove in Theorem \ref{thm_ident_I_post} of Section \ref{sec_ident_full} that the pure variable index set $I$, its partition,  as well as the assignment matrix $A$,  are indeed identifiable under Assumptions \ref{ass_I} \& \ref{ass_J_prime}.
    We give a constructive proof that determines these quantities uniquely from the correlation matrix $R$.\\
    
 
    \noindent{\it 3. Estimating the pure variable set and its partition, with guarantees.}
    The estimation of $I$ follows the steps of the identifiability proofs. As the first step estimates $H$ and $K$ by the procedure in Section \ref{sec_est_H}, we revisit theoretical guarantees of $\wh H$ and $\wh K$ under Assumption \ref{ass_I}. We show, in Corollary \ref{cor_I_prime} of Section \ref{sec_theory_I},  that $\wh H$ estimates $I$, while possibly including  a few indices $i,j$ belonging to non-pure variables that are 
{\em near-parallel} in the sense that $S_q (i,j) < 4\delta_n$.
    

Proposition \ref{prop_K} in Section \ref{sec_theory_I} further shows the consistency of $\wh K$ under mild assumptions, in particular  on the size of the indices corresponding to  
the non-pure variables  that are near-parallel. 

After consistently estimating $K$, Section \ref{sec_est_pure} gives a principled way for  sifting the pure variables from other variables with indices  in the set $\wh H$.
This pruning step is a sample adaptation of the constructive method given, at the population level, in   Lemma \ref{lem_find_L} of Section \ref{sec_ident_full}.  Theorem \ref{thm_I_post} of Section \ref{sec_theory_I} shows that this procedure consistently yields the pure variable index set,   under certain regularity conditions.
Its technical proof 
involves comparing estimated Schur complements of appropriate sub-matrices of $A \C A^\T$. These  quantities are not easy to handle, especially under the  additional layer of complexity   induced by the existence of near parallel variables, and   a delicate uniform control between these matrices and their empirical counterparts is required.  Nevertheless, under a simple set of conditions, we prove that our proposed procedure consistently finds the partition of the pure variable index set. \\

\noindent 
{\it 4.  Estimation of $A$ under the canonical parametrization.}    
    Estimation of the non-zero loadings in the submatrix $A_{I\sbt}$ of $A$ is discussed in Section 
    \ref{sec_est_AI}. Under fairly mild regularity conditions, we establish the consistency of 
the proposed estimator in Theorem \ref{thm_BI}. In Section \ref{sec_est_AJ}, we carefully explain that estimation of the remaining part of $A$ follows  from straightforward adaptation of the methods introduced, developed  in  \cite{bing2017adaptive}. The theoretical guarantees, including the minimax optimal properties of the resulting estimate of $A$   carry over as well. For this reason, we focus in this paper on the recovery of the pure variables and their corresponding loadings $A_{I\sbt}$. In a sense, finding the set $I$ is the first, yet most difficult step in the recovery of the matrix $A$.

\paragraph{Section \ref{sec_sims}: Practical considerations and simulations.}
In Section \ref{sec_sims} we discuss the  practical considerations associated with the implementation of our procedure,  including the selection of tuning parameters and the pre-screening of features with weak signal.  We verify our theoretical claims via extensive simulations.

 \paragraph{Appendix: Proofs and supplementary simulation results.}
Finally, all proofs are collected in the Appendix. Perhaps of 
 independent interest, 	as a byproduct of our proof, we establish in Appendix \ref{app_concentration_R}  deviation inequalities in operator norm for the empirical sample correlation matrix based on  $n$ independent sub-Gaussian random vectors in $\RR^p$, with $p$ allowed to exceed $n$. While 
    similar deviation inequalities for the sample covariance matrix have been well understood   \cite{vershynin_2012,lounici2014,bunea2015,koltchinskii2017},  the operator norm concentration inequalities of the sample correlation matrix is relatively less explored. \cite{elkaroui2009} studied the asymptotic behaviour of the limiting spectral distribution of the sample correlation matrix when $p/n\to (0,\i)$. \cite{wegkamp2016, han2017} prove a similar result for  Kendall's tau sample correlation matrix.

	
	\subsection{Notation} 
	For any positive integer $d$, we write $[d] := \{1,\ldots, d\}$. For two real numbers $a$ and $b$, we write $a\vee b = \max\{a,b\}$ and $a\wedge b = \min\{a, b\}$. 
	
	For any vector $v\in \RR^d$, we write its $\ell_q$-norm as $\|v\|_q$ for $0\le q\le \i$.
	We denote the $k$-th canonical unit vector in $\RR^d$ by $\e_k$, that has zero coordinates except for the $k$-th coordinate, which equals 1.
	Two vectors $u,v\in\RR^d$ are parallel, and we write
	$
	v\parallel u$, if and only if $|\sin(\angle(u, v))| = 0.
	$
	For a subset $S\subset [d]$, we define $v_S$ as the subvector of $v$ with corresponding indices in $S$.
	
	Let $M\in \RR^{d_1\times d_2}$ be any matrix. For any set $S_1\subseteq [d_1]$ and $S_2 \subseteq [d_2]$, we use $M_{S_1,S_2}$ to denote the submatrix of $M$ with corresponding rows $S_1$ and columns $S_2$. We also write $M_{S_1S_2}$ by removing the comma when there is no confusion. In particular, $M_{S_1\sbt}$ ($M_{\sbt S_2}$) stands for the whole rows (columns) of $M$ in $S_1$ ($S_2$). We use $\|M\|_{\rm op}$, $\|M\|_{\rm F}$ and $\|M\|_\i$ to denote the operator norm, the Frobenius norm and elementwise sup-norm, respectively. For any positive semi-definite matrix $M$, we denote by $\lambda_k(M)$ its $k$-th eigenvalue with non-increasing order.
	
	Let $S=S_1\cup\cdots\cup S_K$ 
	be the collection of $K$ sets of indices with $S \subseteq [p]$. 
	For any $i, j\in [p]$,  the notation $i \groupS j$ means that there exists $k\in [K]$ such that $i,j\in S_k$. Its complement $i\not\groupS j$ means $i$ and $j$ do not simultaneously belong to any $S_k$. Let $\Sigma := \Cov(X)$ be the covariance matrix of the random vector $X\in \RR^p$ with diagonal matrix $D_{\Sigma} = \diag(\Sigma_{11}, \ldots, \Sigma_{pp})$
	and we denote the  correlation matrix of $X$  by
	$
	R = D_{\Sigma}^{-1/2}\Sigma D_{\Sigma}^{-1/2}.
	$

	\section{Identification and recovery of approximately replicate features}\label{sec_parallel}
	In this section we show that under Assumption \ref{ass_parallel} stated in the Introduction, both $H$ and its partition $\H$ introduced in Section \ref{contribution}, as well as the latent dimension $K$, can be uniquely  determined from the scale invariant correlation matrix $R$. 
	Our identifiability proofs are constructive, and are  the basis of the estimation procedures described in Section \ref{sec_est_H}, and further analyzed in Section \ref{sec_theory_H_K}.

	\subsection{Identification of the parallel row index set  and latent dimension of $A$}\label{sec_ident}

		We begin by noting that  model (\ref{model}) implies the   decomposition  
	\[
	\Sigma ~ =~  A\C A^\T + \Sigma_E,
	\]
	and consequently
	\begin{equation}\label{eq_R}
	R ~=~  D_{\Sigma}^{-1/2}\Sigma D_{\Sigma}^{-1/2}    
	~ := ~
	B\C B^\T + \Gamma, 
	\end{equation}
	where \[ B := D_{\Sigma}^{-1/2}A \ \  \text{and} \ \ \Gamma := D_{\Sigma}^{-1/2} \Sigma_E D_{\Sigma}^{-1/2}.\]
	Since the matrix  $B$ has the same support as $A$, both matrices $A$ and $B$  share the same index set $H$ of parallel rows, and its partition $\H$.

	The following proposition provides an if and only if characterization of both $H$ and $\H$.
	Its proof is deferred to Appendix \ref{app_proof_prop_I}. We assume \[\min_{i\in [p]} \|A_{i\sbt}\|_2>0.\] Otherwise, 
    we remove all zero rows of $A$ in  the  pre-screening  step described in Section \ref{sec_pre_screen}. The notation  $i \groupH j$ means that both $i,j\in H_k$ for some $k\in [G]$. 
    For any $i,j\in [p]$, let $R_{i,\nij} \in \RR^{p-2}$ denote the $i$th row of $R$, with the entries in the $i$th and $j$th columns removed. 
	
	\begin{prop}\label{prop_I}
	Under model (\ref{model}) and Assumption  \ref{ass_parallel},
		\begin{equation}\label{eq_iff_I}
		i\groupH j\quad \iff \quad A_{i\sbt} \parallel A_{j\sbt} \quad \iff \quad R_{i,\setminus\{i,j\}} \parallel 	R_{j,\setminus\{i,j\}}.
		\end{equation}
    Furthermore, $H$ and $\mathcal{H}$ are uniquely defined.\footnote{The partition $\H$ is unique up to a group permutation.} 
	\end{prop}
	The first equivalence  in (\ref{eq_iff_I}) trivially follows from the definition of $H$. The second equivalence in (\ref{eq_iff_I}) is  key, especially with a view towards estimation, as it shows that  the connection between $H$ and $A$ can be transferred over to the correlation matrix $R$, a model-free estimable quantity. 
	
	To make use of Proposition \ref{prop_I}, first recall that for any two non-zero vectors $v$ and $w$, we have 
	\begin{align}\nonumber
	v\parallel w \quad
	& 
	\iff \quad \min_{\|(a,b)\|_r = 1}\|a v  + b w \|_q = 0
	\end{align}
   for any integers $0\le q\le \infty$ and $0<r\le \infty$, with $  (a, b) \in \RR^2$. This observation, in conjunction with Proposition \ref{prop_I},  suggests the usage of the  following general score function for determining, {\it constructively},  the index set $H$: 
	\begin{equation}\label{score}
	S_{q,r}(i,j) = 
	(p-2)^{-1/q} 
	\min_{\|(a,b)\|_r=1} \left\|a R_{i,\setminus\{i,j\}} + b R_{j,\setminus\{i,j\}}\right\|_q
	\end{equation}
	for any $i\ne j$,   $0<r\le \infty$ and $0\le q\le \infty$.
	The factor 
	$(p-2)^{-1/q}$ serves as a normalizing  constant. 	The following proposition justifies its usage, and offers guidelines on the practical choices of $r$ and $q$. Its proof can be found in Appendix \ref{app_proof_prop_score}.
	
    \begin{prop}[A general score function for finding parallel rows]\label{prop_score}
    	Under model (\ref{model}) and Assumption \ref{ass_parallel}, we have: 
    	\begin{enumerate}
    	    \item[(1)] 		$ 	i\groupH j ~\iff ~ S_{q,r}(i,j) = 0$
    		for any $0<r\le \infty$ and $0\le q\le \infty$. 
    		\item[(2)] 	For any fixed $i,j\in [p]$,  $S_{q,r}(i,j)$ defined in (\ref{score}) satisfies 
    		\begin{enumerate}
    			\item[(i)] $S_{q, t}(i,j) \ge S_{q,r}(i,j)$ for all $0<r\le t\le \i$;
    			\item[(ii)] $S_{t,r}(i,j) \ge S_{q, r}(i,j)$ for all $t \ge q$. \\
    		\end{enumerate}
        \end{enumerate}
    \end{prop}

	In view of part {\it (1)}  of Proposition \ref{prop_score}, one should select $(q,r)$ such that $S_{q,r}(i,j)$ is as large as possible whenever $i\not\groupH j$.	Part ${\it (2)} (i)$ of Proposition  \ref{prop_score} immediately suggests  taking $r = \infty$ and considering the class of score functions 
	\begin{equation}\label{score_ell_q}
	S_q(i,j):=S_{q,\i}(i,j) = (p-2)^{-1/q} \min_{\|(a,b)\|_\i = 1}\left\|
	a R_{i,\nij} + b R_{i,\nij} 
	\right\|_q.
	\end{equation}
	In conjunction with part ${\it (2)} (ii)$ of  Proposition  \ref{prop_score}, the most ideal score function is 
	\begin{equation*}
	S_{\i}(i,j) := S_{\infty, \infty}(i,j) = \min_{\|(a,b)\|_\i=1} \left\|a R_{i,\nij} + b R_{j,\nij}\right\|_\i\quad \forall i,j\in [p],
	\end{equation*}
	corresponding to  $q=\i$.
	While this score function could in principle be computed via linear programming, it is expensive for very large $p$ since   $S_{\i,\i}(i,j)$ needs to be computed for each pair $(i,j)$. A compromise is to choose $q = 2$, especially  because the score function 
	\begin{equation}\label{score_ell_2}
	S_{2}(i,j) := S_{2, \infty}(i,j) = {1\over \sqrt{p-2}} \min_{\|(a,b)\|_\i=1} \left\|a R_{i,\nij} + b R_{j,\nij}\right\|_2\quad \forall i,j\in [p]
	\end{equation}
	has a closed form expression, given in Proposition \ref{prop_score_2} below and proved in Appendix \ref{app_proof_prop_score_2}.
%
	To simplify our notation, and recalling that $R_{i,\nij}, R_{j,\nij} \in \RR^{p-2}$, we define 
	\begin{equation*}
	V_{ii}^{(ij)} :=[ R_{i,\nij}] ^\T R_{i,\nij},\quad V_{ij}^{(ij)} := [ R_{i,\nij} ]^\T R_{j,\nij}\qquad \forall i,j\in [p].
	\end{equation*}
	
	\begin{prop}\label{prop_score_2}
		The score function $S_2(i,j)$ defined in (\ref{score_ell_2}) satisfies
		\[
		\big[S_2(i,j)\big]^2 = {V_{ii}^{(ij)} \wedge V_{jj}^{(ij)} \over p-2}\left(
		1 - {[V_{ij}^{(ij)}]^2 \over V_{ii}^{(ij)}V^{(ij)}_{jj}}\right)\qquad \forall\ i,j\in [p].
		\]
	In particular, 	under model (\ref{model}) and Assumption \ref{ass_parallel}, we have: 
	\begin{equation}\label{iff_H}
	i\groupH j\qquad \iff \qquad S_2(i,j) = 0.
	\end{equation}	
	\end{prop}

We showed in Proposition \ref{prop_I} that $H$ and its partition $\H$ are identifiable under Assumption \ref{ass_parallel}, via an if and only if characterization of $H$, and further provided, in Proposition \ref{prop_score_2}, a constructive if and only if characterization of $H$. The latter is used in the proof of Theorem \ref{thm_partial_ident} below to show that the latent dimension $K = \rank(M_{HH})$, with $M = A\C A^\T$, can be uniquely determined, by showing that the matrix $M_{HH}$ itself is identifiable. 
Theorem \ref{thm_partial_ident}  summarizes these identifiability results and its proof is given in Appendix \ref{sec_proof_thm_partial_ident}.

	\begin{thm}[Partial identifiability]\label{thm_partial_ident}
		Under model (\ref{model}) and Assumption \ref{ass_parallel}, the set $H$, its partition $\H$ and $G= |\H|$ are unique. Moreover, $K$ is uniquely determined. 
	\end{thm}

	Assumption \ref{ass_parallel} is sufficient for identifying $K$, but not the entire matrix $A$.  To see this, note that for any invertible matrix $Q \in \RR^{K \times K}$, there exists some diagonal (scaling) matrix $D \in \RR^{K \times K}$  such that $AZ = \wt A \wt Z$,  where $\wt A = AQD$ has the same index set $H$ as $A,$ and the covaraince matrix $\Cov(\wt Z)$ of  $\wt Z := D^{-1}Q^{-1} Z$ is positive definite and satisfies $\diag[\Cov(\wt Z)]={\bm 1}$, as $\Cov(Z)$. We therefore refer as partial identifiability to the results of  Theorem \ref{thm_partial_ident}. We revisit them in Section \ref{sec_overlap}, where we introduce assumptions under which not only $K$, but also $A$  of model (\ref{model}) can be identified.


	
  	

	\subsection{Estimation of the parallel row index set $H$, its  partition $\mathcal{H}$  and latent dimension $K$}\label{sec_est_H}

	Suppose we have access to $n$ i.i.d. copies of $X\in \RR^p$, collected in 
		a $n\times p$ data matrix $\X$.
		We write the sample covariance matrix as 
	\[
	\wh\Sigma  ={1\over n}  \X ^\top \X
	\]
	and  denote the sample correlation matrix by $\wh R$,
	with entries 
	\[\wh R_{ij} = \wh\Sigma_{ij} / \sqrt{\wh\Sigma_{ii}\wh\Sigma_{jj}}\]
	for any $i,j\in [p]$. Our estimation procedure is the sample level analogue of Theorem \ref{thm_partial_ident} of  Section \ref{sec_ident} above: we first estimate  the parallel row index set $H$ and its partition $\mathcal{H}$, then we estimate  $K$. The statistical guarantees of these estimates are provided in Section \ref{sec_theory_H_K}.

	Recall from part ${\it (2)}(i)$ of Proposition \ref{prop_score} that we use $S_q(i,j):=S_{q,\i}(i,j)$ as a generic score for finding $H$, with any $q\ge 1$. We propose to estimate $S_q(i,j)$ by solving the   optimization problem 
	\begin{align}\label{est_score_ell_q}
	&\wh S_{q}(i,j)=  (p-2)^{-1/q}\min_{\|(a,b)\|_\i=1} \left\|a\wh R_{i,\nij} + b \wh R_{j,\nij}\right\|_q
	\end{align}
	for each $i,j\in [p]$ and $i\ne j$.
	In particular, Proposition \ref{prop_score_2} implies that 
	 $\wh S_2(i,j)$
	 has the   closed form
	\begin{align}\label{est_score_ell_2}
	\wh S_2(i,j) &= 
	\left[{\wh V_{ii}^{(ij)} \wedge \wh V_{jj}^{(ij)} \over p-2}\left(
	1 - {[\wh V_{ij}^{(ij)}]^2 \over \wh V_{ii}^{(ij)}\wh V^{(ij)}_{jj}}\right)\right]^{1/2}
	\end{align}
	with 
	$
	\wh V_{ij}^{(ij)} = [\wh R_{i,\nij}]^\T \wh R_{j,\nij} 
	$
	for all $i,j\in [p]$.
	
	\begin{remark}
	{\rm 
		For any pair $(i,j)$ and any $q\ge 1$, the criterion $\wh S_{q}(i,j)$  in (\ref{est_score_ell_q}) can be computed by solving two convex optimization problems because
		\[
		\wh S_{q}(i,j)  =  (p-2)^{-1/q}\min\left\{\min_{|a|\le 1} \left\|a\wh R_{i,\nij} + \wh R_{j,\nij}\right\|_q,~ 
		\min_{|b|\le 1} \left\|\wh R_{i,\nij} + b \wh R_{j,\nij}\right\|_q
		\right\}.
		\]
		In particular,
		for $q = \i$, computation of  $\wh S_\i(i,j)$ requires solving two linear programs, while for $q = 2$, we have the closed form expression in (\ref{est_score_ell_2}).  
		}
	\end{remark}
	

	
	Algorithm \ref{alg_I} gives  the procedure of estimating parallel row index set, which reduces to finding  all pairs $(i,j)$ with $\wh S_q(i,j)$ below the threshold level $2\delta$. It returns not only the estimated index set $\wh H$, but  also its partition $\wh \H := \{\wh H_1,\ldots,\wh H_{\wh G}\}$. Algorithm \ref{alg_I} requires  one single tuning parameter 
	$\delta$, with an explicit rate stated in Section \ref{sec_theory_H_K}. A fully data-driven criterion of selecting $\delta$ is stated in Section \ref{sec_cv}, relying on the following lemma that shows that the number of  estimated parallel rows of Algorithm \ref{alg_I}    increases in $\delta$.
	\begin{lemma}\label{lem_I_hat}
		Let $\wh H(\delta)$ be the estimated set of parallel rows from Algorithm \ref{alg_I}. Then 
		\[
		|\wh H(\delta)| ~\le~ |\wh H(\delta')|\qquad \forall \delta \le \delta'.
		\]
	\end{lemma}
	\begin{proof}
		For a given $\delta>0$, suppose $i\in \wh H(\delta)$. Then, from Algorithm \ref{alg_I}, there exists some $j\ne i$ such that $\wh S_q(i,j) \le 2\delta$.  This implies $\wh S_q(i,j) \le 2\delta'$ for any $\delta' \ge \delta$. Hence, $i\in \wh H(\delta')$, as desired. 
	\end{proof}


	{\begin{algorithm}[ht]
			\caption{Estimate the parallel row index set $H$ by $\wh H$ and its  partition  $\H$ by $\wh \H$}\label{alg_I}
			\begin{algorithmic}[1]
				\Require Matrix $\wh R\in \RR^{p\times p}$, a positive integer $q\ge 1$, a tuning parameter $\delta > 0$.
				\Procedure  {Parallel}{$\wh R$, $\delta$}
				\State $\wh \H \gets \emptyset$.
				\For {$i=1,\ldots, p-1$} 
				\For {$j=i+1,\ldots, p$}
				\State Compute $\wh S_q(i,j)$ by solving (\ref{est_score_ell_q})
				\If {$\wh S_q(i,j) \le 2\delta$} 
				\State  $\wh\H \gets$ \textsc{Merge($\{i,j\},\ \wh \H$)}
				\EndIf	
				\EndFor
				\EndFor
				\State \Return $\wh \H$, $\wh G=|\wh \H|$ and $\wh H = \cup_{k} \wh \H_k$
				\EndProcedure
				\Statex
				\Function {Merge}{$S$, $\wh \H$}
				\State Add = True
				\ForAll {$g \in \wh \H$}
				\Comment $\wh \H$ is a collection of sets
				\If {$g \cap S\ne \emptyset$} 
				\State  $g\gets g\cup S$
				\Comment Replace $g\in \wh \H$ by $g\cup S$
				\State Add = False
				\EndIf
				\EndFor
				\If {Add}
				\State $\wh \H = \wh \H \cup \{S\}$\Comment add $S$  in $\wh \H$
				\EndIf
				\State \Return $\wh \H$
				\EndFunction
			\end{algorithmic}
	\end{algorithm}}

	\bigskip

	Since $\wh G$ estimates $G$ and $G$ is typically larger than $K$, unless there are exactly $K$ sets of parallel rows in $A$,  we propose the following procedure for estimating  $K$ by using the output $\wh \H = \{\wh H_1, \ldots, \wh H_{\wh G}\}$ of Algorithm \ref{alg_I}. 
	It relies on the observation that $[B\C B^\T]_{LL}$ has rank $K$ where $L=\{\ell_1,\ldots, \ell_G\}$ with $\ell_k\in H_k$ for each $1\le k\le G$. 
	\begin{itemize}
		\item For each $k\in [\wh G]$,  we select one representative variable index from $\wh H_k$ as 
		\begin{equation}\label{def_l_hat}
		\ell_k := \arg\max_{i\in \wh H_k} \bigl\|\wh R_{i,\setminus\{i\}}\bigr\|_q,
		\end{equation}
	    and
		create the set of representative indices 
		\begin{equation}\label{def_L_hat}
		\wh L := \bigl\{\wh \ell_1, \ldots, \wh \ell_{\wh G}\bigr\}.
		\end{equation}
		\item 	Next, motivated by (\ref{eq_MI_offdiag}) and (\ref{ident_MI}) in the proof of Theorem \ref{thm_partial_ident}, we propose to estimate  the submatrix $M_{\wh L\wh L}$ of $M := B \C B^\T$ by 
		\begin{equation}\label{est_M_LL_off_diag}
		\wh M_{ij} = \wh R_{ij},\quad \forall i,j\in \wh H,~ i\ne j,
		\end{equation}
		and
		\begin{eqnarray}\label{est_M_LL_diag}
        	\wh M_{ii}  &=& |\wh R_{i\wh j}|~  {\|\wh R_{i,\setminus\{i,\wh j\}}\|_q  \over\|\wh R_{\wh j,\setminus\{i,\wh j\}}\|_q },\qquad \forall i\in \wh H_k, ~ k\in[\wh G],\\
        	\wh j &=& \arg\min_{\ell \in \wh I_k\setminus \{i\}} \wh S_q(i,\ell).
        	\nonumber
    	\end{eqnarray}
	Instead of choosing $\wh j$ for each $i\in \wh H_k$ as above, we could alternatively estimate $\wh M_{ii}$ via (\ref{est_M_LL_diag}) by averaging over $j \in \wh I_k \setminus \{i\}$. Our numerical experiments indicate that these two procedures have  similar performance.
	
		\item	Finally, we determine 
		the approximate rank $\wh K$ of the matrix $\wh M_{\wh L \wh L}$ from (\ref{est_M_LL_off_diag}) -- (\ref{est_M_LL_diag}) 
		by  
		\begin{equation}\label{est_K}
		\wh K := \max \left\{ k \in [ \wh G]~:~
		\lambda_k(\wh M_{\wh L\wh L}) \ge \mu 
		\right\}
		\end{equation}
		for some tuning parameter $\mu > 0$.\\
		
	\end{itemize}

	\subsection{Statistical guarantees for   $\wh H$, $\wh \H$ and $\wh K$}\label{sec_theory_H_K}

	We will assume that the feature $X$ is a sub-Gaussian random vector. 
	Recall that
	a centered random vector $X\in \RR^d$ is $\gamma$--sub-Gaussian if $\EE[\exp(u^\T X)] \le \exp(\|u\|_2^2\gamma^2/2)$ for any fixed $u\in \RR^d$.
	The quantity $\gamma$ is called the sub-Gaussian constant.
	
	\begin{ass}\label{ass_distr}
	    There exists a constant $\gamma$ such that $\Sigma^{-1/2} X$ is $\gamma$--sub-Gaussian.\footnote{Under model (\ref{model}), if there exist  constants $\gamma_Z,\gamma_E > 0$, such that
		$\C^{-1/2}Z$ and
		$\Sigma_E^{-1/2}E$
		are sub-Gaussian random vectors with sub-Gaussian constants 
		$\gamma_Z$ and $\gamma_E$, respectively, then $\Sigma^{-1/2} X$ is $\gamma$--sub-Gaussian with $\gamma = \max\{\gamma_Z, \gamma_E$\}.}
	\end{ass}

	The only tuning parameter in Algorithm \ref{alg_I} is $\delta$ with theoretical order given by 
	\begin{equation}\label{delta_n}
	\delta_n :=  c(\gamma)\sqrt{\log (p\vee n) / n}
	\end{equation}
	for some constant $c(\gamma)$ depending on $\gamma$ only.  $\delta_n$ is a 
	key quantity that controls   the   deviation  $\wh R$ from  $R$. 
	Indeed, 	under Assumption \ref{ass_distr} and $\log p \le n$, Lemma \ref{lem_dev_R} in Appendix \ref{app_aux_lem} shows that, with probability $1-4/(p\vee n)$, the   event  
	\begin{equation}\label{event:E}
	\E := \left\{\max_{1\le i,j\le p}|\wh R_{ij}-R_{ij}| \le \delta_n\right\}
	\end{equation}
	holds.  
	Throughout the rest of the paper, we make the blanket assumption that $\log p \le n$. 
	
	
	The following theorem provides the uniform deviation bounds for $\wh S_q(i,j) - S_q(i,j)$ over all $(i,j)\in [p]$ for any $1\le q\le \i$. Its proof  is deferred to Appendix \ref{app_proof_thm_score_q}. At this point, it is useful to discuss  the value of $q$ in the criterion $\wh S_q$ used in Algorithm \ref{alg_I}. We found in our simulations that $\wh S_q$ with $q=2$ performs well in terms of statistical accuracy and computational speed, the latter due to its closed form. While we present our conditions and statements in terms of a general, fixed $q\ge 1$, our preferred choice is $q=2$.
	
	\begin{thm}\label{thm_score_ell_q}
		On the event $\E$, one has, for any $1\le q\le \i$,
		\begin{alignat*}{2}
		\max_{1\le i,j\le p}\left|\wh S_q (i,j) - S_q(i,j)\right| ~\le~ 2\delta_n.
		\end{alignat*}
		If, in addition,  model (\ref{model}) and Assumption \ref{ass_parallel} hold , one has
		\begin{alignat*}{2}
		&\wh S_q(i,j) ~\le~ 2\delta_n, \qquad &&\textrm{for all} \quad  i\groupH j;\\
		&\wh S_q(i,j) ~\ge~ \max\left\{0,~ S_q(i,j) - 2\delta_{n}\right\},\qquad &&\textrm{for all}  \quad  i,j \in [p]
		\end{alignat*}
		for any $1\le q\le \i$.
	\end{thm}


	Algorithm \ref{alg_I} with $\delta = \delta_n$ selects those indices $i,j\in \wh H$ for which $\wh S_q(i,j)\le 2 \delta_n$. On the event $\E$, Theorem \ref{thm_score_ell_q} tells us that
	\begin{equation}\label{eq_H_twosides}
	 H	~\subseteq ~\wh H~ \subseteq  ~\left\{ i \in [p]:~ S_q(i,j) \le 4\delta_n~ \text{for some $j\ne i$} \right \}.
	\end{equation}
	Hence, with probability at least $1-4/(p\vee n)$, $\wh H$ includes all parallel rows in $H$ and 
	may mistakenly include  near-parallel rows  corresponding to 
	 $S_q(i,j)< 4\delta_n$.  Note that this holds without imposing any signal condition. 
	
	On the other hand, the partition  $\wh \H$, however, may not include all the groups that make up the 
	 partition $\H$. For instance, there may be two distinct groups of parallel rows with weak signals that get merged, while  some subgroups of parallel rows may be enlarged by a few near-parallel rows. Nevertheless, Theorem \ref{thm_score_ell_q} immediately implies that
	  \begin{equation}\label{cond_signal_H}
	  	S_q(i,j) > 4\delta_n, \quad \text{for all } i\not\groupH j
	  \end{equation} 
	 is a sufficient condition for consistent estimation of both $H$ and $\H$, as summarized in the following corollary. 
	 
	 \begin{cor}\label{cor_H}
		Under model (\ref{model}) and Assumption \ref{ass_parallel}, assume that (\ref{cond_signal_H}) holds.
	Then, on the event $\E$,  the output $\wh G$ and $\wh \H = \{\wh H_1, \ldots, \wh H_{\wh G}\}$ from Algorithm \ref{alg_I} with $\delta = \delta_n$ satisfy: 
	$\wh G = G$, $\wh H = H$ and  $\wh H_k = H_{\pi(k)}$ for all $1\le k\le G$, for some permutation $\pi: [G] \to [G]$.
	
	\end{cor}


 The following proposition states the explicit rate of the tuning parameter $\mu$ under which we provide theoretical guarantees for $\wh K$ as an estimator of $K$, for any $q\ge 1$.  For any $L = \{\ell_1,\ldots,\ell_G\}$ with $\ell_k\in H_k$ for all $k\in [G]$, define
	\begin{equation}\label{def_cL}
	 \ul c(L) :=\lambda_K(B_{L\sbt}\C B_{L\sbt}^\T), \qquad \ol c(L) :=\lambda_1(B_{L\sbt}\C B_{L\sbt}^\T).
	\end{equation}
	
	\begin{prop}\label{prop_K_parallel}
		Under model (\ref{model}),  Assumptions \ref{ass_parallel} \& \ref{ass_distr} and condition (\ref{cond_signal_H}), suppose there exist some constant $c>0$ such that 
		\begin{equation}\label{cond_regularity_parallel}
		  \|R_{i,\nij}\|_q \ge c (p-2)^{1/q},\quad \forall ~ i\groupH j.
		\end{equation}
		For  
		$    \mu = 
		C(\sqrt{\wh G\delta_n^2}  + \wh G\delta_n^2) 
		$
		in (\ref{est_K}), and  for some  $C>0$ depending on $\ol c(\wh L)$,
		we have 
		\begin{equation*}
		     \PP\{\wh K \le  K\} \ge 1 - c'(n\vee p)^{-c''}.
		\end{equation*}
		If, additionally, $\max_L[\ol c(L) / \ul c(L)] G\delta_n^2\le c_0$ for some sufficiently small constant $c_0>0$,  we further have
		\begin{equation*}
		\PP\{\wh K = K\}\ge 1 - c'(n\vee p)^{-c''}. 
		\end{equation*}
	\end{prop}	
	
	The proof of Proposition \ref{prop_K_parallel} is deferred to Appendix \ref{app_proof_prop_K_parallel}. 
	We first prove that 
	$\wh K \le K$  always holds on the event $\E$ with the specified choice of $\mu$. Proving consistency  $\wh K = K$ requires $\mu$ to be sufficiently small so that $\mu<\ul c(\wh L)$,   which is guaranteed by $\max_L[\ol c(L) / \ul c(L)] G\delta_n^2=o(1)$. When both $\ul c(L)$ and $\ol c(L)$ are bounded away from $0$ and $\i$, consistency only requires $G\delta_n^2 = o(1)$, that is, we allow $G=|\H|$ to grow but no faster than $O(n/\log (n\vee p))$. 
	The data-driven choice of the leading constant of $\mu$ is discussed in Section \ref{sec_cv}. Condition (\ref{cond_regularity_parallel}) is a mild regularity condition. For instance, it requires $\min_{i \in H}R_{ii} \ge c$ for  $q = \i$. Sufficient conditions of (\ref{cond_regularity_parallel}) for $q = 2$ are provided in Proposition \ref{prop_K} of Section \ref{sec_theory_I}.

	\section{Identifiability and recovery under canonical parametrization}
	\label{sec_overlap}

    As argued in Section \ref{sec_ident}, although the very general Assumption \ref{ass_parallel} is sufficient for identifying $H$, its partition $\mathcal{H}$ and the dimension of the model, it only determines $A$ up to rotations. Considering  the parametrization provided by Assumption \ref{ass_I} is a first step towards the unique determination of $A$. If $I$ is known, one can use the results of Section \ref{sec_ident_full} below to identify $A$. However, as $I$ is also unknown, Assumption \ref{ass_I}  is not sufficient for identifying $I$, hence neither for $A$. In Section \ref{sec_ident_full} we provide another assumption that, combined with  Assumption \ref{ass_I}, is provably sufficient for determining $I$, and therefore $A$, uniquely.  As the identifiability proofs are constructive, they naturally lead to the procedure of estimating the pure variable  index set, as stated in Section \ref{sec_est_pure}. Section \ref{sec_theory_I} provides the statistical guarantees for its output. Finally, the estimation of the loading matrix $A$ and its statistical guarantees are provided in Section \ref{sec_est_A}. 

    We begin by introducing the notation employed in the following sections. 
    For model (\ref{model}) satisfying Assumption \ref{ass_I}, we let  $J := [p]\setminus I$ be the index set corresponding to non-pure variables. We  introduce the index set corresponding to  parallel, but  non-pure,  rows of $A$ as
	\begin{equation}\label{def_J1}
	J_1 :=  \left\{
	j\in J:  A_{j\sbt} \parallel A_{\ell \sbt} \text{ for some }\ell \in J\setminus\{j\}
	\right\},
	\end{equation}
	and its partition $\J_1 := \{J_1^1, \ldots, J_1^N\}$. 
	With this notation,  the index  set $H$ of {\it all} parallel rows of $A$ and its partition, decompose as 
	\[
	H = I\cup J_1\quad\textrm{ and }\quad   \H = \{H_1,\ldots, H_G\} = \{I_1, \ldots, I_K, J_1^1, \ldots, J_1^N\}.
	\]
	The total number of groups in $\H$ is $G= K+N$.


		\subsection{Model identifiability under a canonical parametrization}\label{sec_ident_full}

    Recall that $H$ and $\H$, as well as $G = |\H|$ and  $K$, are  identifiable under Assumption \ref{ass_parallel}, hence under Assumption \ref{ass_I}. 
    To begin discussing when  $I$ is also identifiable, we distinguish between two cases $J_1 = \emptyset$ and $J_1 \neq \emptyset$, by simply comparing $G$ with $K$.

    If the only parallel rows in $A$ correspond to pure variables ($J_1 = \emptyset$ or, equivalently, $G = K$), then $I$ is identifiable as an immediate consequence of Theorem  \ref{thm_partial_ident}. Moreover, $A$ can be shown to be identifiable. We summarize this in the following theorem, proved in Appendix \ref{app_proof_thm_ident_A}. Although we state Theorem \ref{thm_ident_I} in terms of the $S_2$ score function, which we will later use for estimation purposes, it holds for any of the $S_q$ scores defined in Section \ref{sec_ident}. Recall that $i \group j$ means $i,j\in I_k$ for some $k\in [K]$.
	\begin{thm}\label{thm_ident_I}
		Under model (\ref{model}) and Assumption \ref{ass_I}, assume $J_1 = \emptyset$. Then 
		\[
				i\group j\qquad \iff \qquad S_2(i,j) = 0.
		\]
		Therefore, $I$ is identifiable and its partition $\I$ is identifiable up to a group permutation. Furthermore, the matrix  $A$ is identifiable up to a $K\times K$ signed permutation matrix.
	\end{thm}

	If $J_1 \ne \emptyset$ (or equivalently, $G>K$), the pure variable set $I$ cannot be distinguished from $J_1$, unless further structure is imposed on the model.  To strengthen this observation,  consider the following toy example.
	\begin{example}\label{example}
		Let $ K = 2$ and $ p = 4$. Consider 
		\[
		A = \begin{bmatrix}
		a_1 & 0 \\
		a_2 & 0 \\
		0 & b_1 \\
		0 & b_2 \\
		c_1 & c_2 \\
		d c_1 & d c_2
		\end{bmatrix},\qquad \C = \begin{bmatrix}
		1 & \rho  \\ 
		\rho & 1\\
		\end{bmatrix},
		\]
		where $a_1, a_2, b_1, b_2, c_1, c_2$ and $d$ are all non-zero constants and $\rho \in [0,1)$.
		Clearly, $A$ satisfies Assumption \ref{ass_I} with $\I := \I(A) = \{\{1,2\}, \{3,4\}\}$ and $J_1 = \{5,6\}$ and, in the notation of Theorem \ref{thm_partial_ident}, $G = 3$ and $H = \{1, \ldots, 6\}$. Consider
		the following
		\[
		\wt A =  A Q  = A\begin{bmatrix}
		a_1 & 0 \\
		c_1 & c_2\\
		\end{bmatrix}^{-1} 
		= 
		\begin{bmatrix}
		1 & 0 \\
		{a_2 \over  a_1} & 0 \\
		-{c_1b_1 \over a_1c_2} & {b_1\over c_2} \\
		-{c_1b_2 \over a_1c_2} & {b_2\over c_2} \\
		0 & 1 \\
		0 & d\\
		\end{bmatrix} ,\qquad \wt\Sigma_Z = Q^{-1}\C(Q^{-1})^\T.
		\]
		Since we can always find a $2\times 2$ diagonal matrix $D$ (the diagonal elements are non-zero) such that $\diag(D \wt \Sigma_Z D) = 1$, we conclude 
		\[
		A\C A^\T = \left(\wt A D^{-1}\right)\left(D\wt \Sigma_Z D\right) \left(\wt AD^{-1}\right)^\T.
		\]
		Notice that $\wt A D^{-1}$ satisfies Assumption \ref{ass_I} with $\wt \I = \wt \I(\wt AD^{-1}) = \{\{1,2\}, \{5, 6\}\}$, $|\wt \I| = K = 2$ and $\wt J_1 = \{3,4\}$.
		Therefore, there exist two distinct matrices $A$, each with a respectively different   pure variable index set, but with the same index set $H$ of parallel rows.\\ 
	\end{example}

    Our general identifiability result, stated in Theorem \ref{thm_ident_I_post} below, allows for $J_1 \neq \emptyset$. The rationale behind its proof is the following:

{\bf Step 1.}  Show that $H$, $\H$ and $K$ can be uniquely determined, as in Theorem \ref{thm_partial_ident}. With $|\H| = G$, if $G = K$, appeal to Theorem \ref{thm_ident_I} above to complete the proof. If $G > K$,  provide a statistically meaningful criterion for selecting $K$ representative indices from $H$. There is freedom in the  choice of such a criterion, and we build one such that: (i) each representative  contains as much information as possible; (ii) representatives of  different groups are as uncorrelated as possible.  

{\bf Step 2.} Provide further conditions on $A$ under which the thus selected indices correspond to distinct pure variable indices. Then, use Proposition \ref{prop_score} to reconstruct the entire set $I$, and its partition, by the aid of the score function $S_2$ (or any $S_q$). Complete the proof by showing that $A$ is identifiable.\\

For the first step, we propose 
to select indices of variables that maximize, successively, Schur complements of appropriately defined matrices, with general form given by (\ref{def_i_k_population}). For the second step, we make 
   Assumption \ref{ass_J_prime}, which is  sufficient  for proving the following fact: if  the index $i_k$  is given by  (\ref{def_i_k_population}),  then    $X_{i_{k}}$ is indeed  a pure variable, and can be taken as the representative of its group. This is formally stated  in Lemma \ref{lem_find_L} and explained in its subsequent remark.  With a view towards estimation, we note that this selection criterion will be constructively used in Section \ref{sec_est_pure}.


    \begin{ass}\label{ass_J_prime}
    	Let $\xi_k:= \max_{i\in I_k} \|A_{i\sbt}\|_1$ be the largest loading in group $I_k$, in absolute value.
    	We have
    	\[
    	\sum_{k=1}^K {|A_{jk}| \over \xi_k} \le 1,\qquad \forall j\in J_1.
    	\]
    \end{ass}

    We note, trivially, that the assumption is automatically met  when $J_1 = \emptyset$. 
   Assumption \ref{ass_J_prime} imposes  a scaling constraint between $I$ and $J_1$, without placing any  restriction  on the rows of $A$ corresponding to $J\setminus J_1$. 
    A sufficient condition for Assumption \ref{ass_J_prime} is 
    $\|A_{i\sbt}\|_1 = \xi$ for all $i\in I$ and $\|A_{j\sbt}\|_1 \le \xi$ for all $j\in J_1$, which reduces it  to the condition employed in \cite{bing2017adaptive} with $J_1 = J$ and $\xi = 1$, and thus generalizes that work.\\

    Under Assumption \ref{ass_J_prime}, the lemma below states a systematic way of finding $K$ group representatives  by successively maximizing certain Schur complements of $\Theta := A\C A^\T$. Its proof can be found in Appendix \ref{app_proof_lem_find_L}.

    \begin{lemma}\label{lem_find_L}
    	Assume model (\ref{model}) and Assumptions \ref{ass_I} \& \ref{ass_J_prime}.  For any $1\le k\le K$, let $S_k = \{i_1, \ldots, i_{k-1}\}$ with $S_1 = \emptyset$ satisfying $i_a \in I_{\pi(a)}$ for all $1\le a\le k-1$ and some permutation $\pi: [K] \to [K]$. Then one has 
    	\begin{equation}\label{def_i_k_population}
    	i_k :=\arg\max_{j \in H} \Theta_{jj| S_k} \in I_{\pi(k)},\qquad \text{for all }1\le k\le K,
   	 \end{equation}
    	where $\Theta_{jj| S_k} = \Theta_{jj} - \Theta_{jS_k}^\T \Theta_{S_kS_k}^{-1}\Theta_{S_k j}$ and $H = I\cup J_1$.
    \end{lemma}
   
   \begin{remark}
    The procedure in Lemma \ref{lem_find_L} is based on the following rationale which achieves the previous two goals (i) and (ii).
    Let $W_H := X_H-E_H= A_{H\sbt}Z $ for $H\subseteq [p]$ and observe that  $\Theta$ is the degenerate (rank $K$) covariance matrix of $W \in \RR^p$.
    To add intuition, if $Z$ has a multivariate normal distribution,  then  $\Theta_{H^cH^c|H} = {\rm Cov}(W_{H^c} | W_H)$. In this case, display (\ref{def_i_k_population}) in Lemma \ref{lem_find_L} becomes
    \[
    i_k = \arg\max_{j\in H} {\rm Var}(W_j | W_{S_k}),\qquad \text{for all }1\le k\le K,
    \]
    and the procedure returns  the $K$ largest {\it conditional variances} $  \text{Var}\left(W_j  | W_{i_1}, \ldots, W_{i_{k-1}}\right)$.
    Suppose we have already  selected $W_{i_1}$ and we are considering  the selection of a new index,  $i_2$. Since
    \[
    {\rm Var}(W_j | W_{i_1}) = \Theta_{jj}\left[
    1 - {\rm Corr}(W_j, W_{i_1})
    \right],
    \]
    we see that maximizing the above conditional variance retains more information (for goal (i)), while avoiding linear dependence by reducing ${\rm Corr}(W_{i_1}, W_{i_2})$ (for goal (ii)). While one can always select $K$ variables, from a given collection,  in this manner, it is Assumption  \ref{ass_J_prime} that ensures that their indices do indeed correspond to pure variables in this parametrization of the model. The de-noising step implicit in Lemma \ref{lem_find_L} is crucial for this procedure. It is made possible by the determination of the superset $H = I\cup J_1$, which in turn enables the identifiability of $\Theta$ and of its various functionals employed above, as shown in the proof of Theorem \ref{thm_partial_ident}. In Section \ref{sec_est_pure}, these arguments will be used constructively for estimation purposes.
   \end{remark}  
   
    We use the following toy example to aid our explanation. 
    \begin{example}\label{example_cond_var}  In a factor model under Assumption \ref{ass_I}, 
    	let  $K = 2$, $H = \{1, 2, 3, 4, 5, 6\}$, $\I = \{\{1, 2\}, \{3, 4\}\}$ and assume that the $W$'s, which are denoised versions of the $X$'s,  are given by   $W_1 = 1.1 Z_1$, $W_2 = 0.8Z_1$,  $W_3 = Z_2$, $W_4 = 0.5Z_2$, $W_5 = 0.2Z_1 + 0.4Z_2$, $W_6 = -0.3Z_1 -0.6 Z_2$. The coefficients are chosen such that Assumption \ref{ass_J_prime} is met.  Assume that $Z_1$ and $Z_2$ are two independent, zero mean, Gaussian random variables with $\text{Var}(Z_1) = \text{Var}(Z_2) = 1$.
    	It is easy to see that $\text{Var}(W_1) = \max_{j \in H} \text{Var}(W_j)$ and $\text{Var}(W_3|W_1) \geq \text{Var}(W_j|W_1)$ for $j\in \{5,6\}$.
    	Thus, in this example, the conditional variance of a  denoised pure variable is larger than that of a denoised non-pure variable, when conditioned on a previously determined pure variable $(W_1)$.  
    \end{example}

    Based on the procedure in Lemma \ref{lem_find_L}, both $I$ and $\I$ are identifiable, and so is the entire loading matrix $A$. We summarize these results in the theorem below. Its proof is constructive and is deferred to Appendix \ref{app_proof_thm_ident_I_post}.

    \begin{thm}\label{thm_ident_I_post}
    	Under model (\ref{model}) and Assumptions \ref{ass_I} \& \ref{ass_J_prime}, $I$ is identifiable and its partition $\I$ is identifiable up to a group permutation. Furthermore, the matrix  $A$ is identifiable up to a $K\times K$ signed permutation matrix.
    \end{thm}
    
      
    \begin{remark}[Discussion of Assumptions \ref{ass_I} and \ref{ass_J_prime}]\label{rem_one_pure}\mbox{} 
    \begin{enumerate}
        \item {\it Discussion of Assumption \ref{ass_I}: quasi-pure variables}. 
          As discussed in detail in the Introduction, one way to parametrize a factor model (\ref{model}) in which  $\rank(A) = K$ is via the errors-in-variables parametrization, which refers to loading matrices $A$ that contain a $K\times K$ diagonal sub-matrix. In the terminology of this paper, this requires the existence of {\it one} pure variable per latent factor, and it fixes $A$ uniquely, when $I$ and $\Sigma_E$ are known. From this perspective, when $I$ and $\Sigma_E$ are not known,  Assumption \ref{ass_I} requires the existence of  {\it only one additional} pure variable  per latent factor. 
        
        To preserve full generality, it  is perhaps more realistic to assume the existence of {\it one extra quasi-pure variable}, provided that it also leads to the identifiability results proved above. We argue below that this is possible, although it may lead to unnecessarily heavy techniqualities that would obscure the main message of this work. We therefore content ourselves to explaining how such an 
        assumption can be used, without pursuing it fully throughout the paper. 
        	
        	Suppose that there exists only one pure variable $i_k$ for some group $k\in [K]$, that is, $|I_k| = 1$.  Assume also that  there exists some {\em quasi-pure} variable $j$ for this group $k$, in the sense that for that index $j$ we have: 
    	\begin{equation}\label{quasi-scale}
    	   {\sum_{a\ne k}|A_{ja}| \over \|A_{j\sbt}\|_1} = 1 - {|A_{jk}| \over \|A_{j\sbt}\|_1} \le \varepsilon^2.
    	\end{equation}
    	As we can see, when $\varepsilon$ is small, the $j$th variable is close to the pure variable $i_k$,  as the majority of  the weights in  its corresponding row of $A$  are placed on the $k$th factor. We show in Appendix \ref{app_proof_rem_one_pure} that, for this $i_k$ and $j$, the score in (\ref{score_ell_2}) satisfies 
    	\[  
    	    S_2(i_k, j) \le 2\varepsilon.
    	\]
    	By slight abuse of notation, we write $I_k = \{i_k, j\}$ and provided that $S_2(\ell, \ell') > 2\epsilon$ for all $\ell \not\group \ell'$ with $S_2(\ell,\ell')$ defined in (\ref{score_ell_2}), both $I$ and its partition can be recovered uniquely by applying 
    	Algorithm \ref{alg_I} to the population correlation matrix $R$,  with $q = 2$ and $\delta \ge 2\varepsilon$.
    	
         \item {\it Discussion of Assumption \ref{ass_J_prime}}. This assumption is only active  when $J_1 \neq \emptyset$. 
          
          If $J_1 \neq \emptyset$ and Assumption  \ref{ass_I} holds,  but Assumption  \ref{ass_J_prime} does not hold,  then the representative selection of Lemma \ref{lem_find_L}, as well as group reconstruction, can still be performed in an identical manner. However, one cannot guarantee that all groups consist of only pure variables. Nevertheless, their representatives will continue to have  properties (i) and (ii), by construction, and therefore still be statistically meaningful. 
        \end{enumerate}
    \end{remark}

     \subsection{Estimation of the pure variables index set}
 	\label{sec_est_pure}
 	
 	The estimation procedure follows, broadly, the {\bf Step 1} and {\bf Step 2} employed in the proof of Theorem \ref{thm_ident_I_post} of the previous section. 
 	We first use Algorithm \ref{alg_I} and the procedure of estimating $K$ described in Section \ref{sec_est_H} to obtain estimates $\wh \H= \{\wh H_1, \ldots, \wh H_{\wh G}\}$ and $\wh K$. If $\wh K = \wh G$, no further action is taken; if  $\wh K< \wh G$, we add the pruning step stated below, based on  
 	  the sample analogue of  Lemma \ref{lem_find_L}. 
 	
 	

  For any input  $r < \wh G$ (for instance, $r = \wh K$ when $\wh K < \wh G$), the pruning step consists of the following  steps to estimate $\wh \I$, collected   in Algorithm \ref{alg_I_post}. 

	\begin{itemize}
        \item	  Estimate
        	$\Gamma_{\wh H\wh H}$ by $\wh \Gamma_{ij} = 0$ for all $i\ne j$ and  
    	\begin{equation}\label{est_Gamma_II}
        	\wh \Gamma_{ii} = 1 - \wh M_{ii},\qquad \forall i\in \wh H
    	\end{equation}
    	where $\wh M_{ii}$ is obtained from (\ref{est_M_LL_diag}) with $q = 2$.
	\item
	Set  $\Theta := A\C A^\T$ and, in view of (\ref{eq_R}), estimate $\Theta_{\wh H\wh H}$ by 
	\begin{equation}\label{est_Theta_LL}
	    \wh \Theta_{\wh H\wh H} = \left[\wh \Sigma- \wh \Sigma_E\right]_{\wh H\wh H} = \left[\wh \Sigma - \wh D^{1/2}_{\wh \Sigma}\wh \Gamma \wh D^{1/2}_{\wh \Sigma}\right]_{\wh H\wh H}
	\end{equation}
	with $\wh D_{\wh \Sigma} = \diag(\wh\Sigma_{11}, \ldots, \wh \Sigma_{pp})$. 
	\item For any set $S\subseteq
\wh H$ and $S^c = \wh H\setminus S$, write the Schur complement of $\wh\Theta_{SS}$ of $\wh\Theta_{\wh H\wh H}$ as 
	\begin{equation}\label{est_Theta_schur}
		\wh\Theta_{S^cS^c|S} = 	\wh\Theta_{S^cS^c} - 	\wh\Theta_{S^cS}\wh\Theta_{SS}^{-}	\wh\Theta_{SS^c}
	\end{equation}
	with $M^-$ denoting the Moore-Penrose pseudo-inverse of $M$.  Set $S_1 = \emptyset$, and  let $S_k = \{i_1, \ldots, i_{k-1}\}$  for each $2\le k\le r$,
	and define 
	\begin{equation}\label{def_i_k}
		 i_k = \arg\max_{j\in  \wh H} \wh \Theta_{jj | S_k}.
	\end{equation}
	When there are ties,  arbitrarily pick one of the maximizers.
	
	\item  The final estimate of $\I$ is defined as 	\begin{equation}\label{est_wt_I}
	    \wh \I = \left\{\bigl\{\wh H_k \bigr\}_{1\le k \le r}: \textrm{there exists $a\in $  $[r]$ such that $i_a\in \wh H_k$}\right\}.
    	\end{equation}
	\end{itemize}

	\begin{algorithm}[ht]
	\caption{Prune the parallel row index set obtained from Algorithm \ref{alg_I}}\label{alg_I_post}
	\begin{algorithmic}[1]
			\Require $\wh \Sigma, \wh R\in \RR^{p\times p}$, the partition $\wh \H$ with $\wh G = |\wh \H|$, the integer $1\le r <\wh G$.
			\Procedure {Pruning}{$\wh \Sigma$, $\wh R$, $\wh \H$, $r$}
			\State Compute $\wh \Theta_{\wh H\wh H}$ from (\ref{est_Theta_LL})
			\State Set $S = \emptyset$
			\For{$k = 1,\ldots, r$} 
			    \State Compute $i_k$ from (\ref{def_i_k}) and add $i_k\in S$
			\EndFor 
			\State \Return $\wh \I$ obtained from (\ref{est_wt_I})
			\EndProcedure
		\end{algorithmic}
	\end{algorithm}
	
	We note that  Algorithm \ref{alg_I_post} can take any $1\le r\le K$  as input, including a random $r$. This adds flexibility to the procedure, should one want to use a value of $r$  different from  the value of $\wh K$ defined in (\ref{est_K}), for instance if one uses a different estimator of $K$, or one is interested in a pre-specified value of $r$. We analyze $\wh \I$ in the next section.

	For the reader's convenience, we summarize our whole procedure for  estimating the pure variable index set and its partition, Pure Variable Selection (PVS), in Algorithm \ref{alg_whole}.\\	
	
   \begin{algorithm}[h!]
    	\caption{Pure Variable Selection (PVS)}\label{alg_whole}
	\begin{algorithmic}[1]
			\Require $\wh \Sigma, \wh R\in \RR^{p\times p}$, a positive integer $q\ge 1$, two positive tuning parameters $\delta$ and $\mu$.
			\State Obtain $\wh \H$ and $\wh G = |\wh \H|$ from {\sc Parallel}($\wh R$, $\delta$) in Algorithm \ref{alg_I}
			\State Compute $\wh L$ from (\ref{def_l_hat}) -- (\ref{def_L_hat}) and compute $\wh M_{\wh L \wh L}$ from (\ref{est_M_LL_off_diag}) -- (\ref{est_M_LL_diag})
			\State Estimate $\wh K$ from (\ref{est_K})
			\If{$\wh K < \wh G$}
			\State Obtain $\wh \I$ from {\sc Pruning}($\wh \Sigma$, $\wh R$, $\wh \H$, $\wh K$)
			\Else 
			\State Set $\wh \I = \wh \H$
			\EndIf 
			\State \Return $\wh \I$
		\end{algorithmic}
	\end{algorithm}


    \subsection{Statistical guarantees for the estimated pure variable index set}\label{sec_theory_I}	
	  
	  We provide statistical guarantees for the estimated pure variable index set  obtained via  Algorithm \ref{alg_whole}. As Algorithm \ref{alg_I} is used to estimate $H$, $\H$ and $K$ first, we start by revisiting the statistical guarantees for $\wh H$, $\wh \H$ and $\wh K$ under the canonical parametrization of $A$.  The theoretical properties of $\wh \I$ are given in Theorem \ref{thm_I_post}.

    
    As shown in (\ref{eq_H_twosides}) of Section \ref{sec_theory_H_K}, on the event $\E$, $\wh H$ from Algorithm \ref{alg_I} 
	includes $H = I \cup J_1$ with $J_1$ defined in (\ref{def_J1}), but 
	 may mistakenly include  
	 other  variable indices, for instance pairs  satisfying  
	 $S_q(i,j)< 4\delta_n$.  Under Assumption \ref{ass_parallel}, 
	 if the signal condition (\ref{cond_signal_H}) holds, we further showed that  both $\wh H$ and its partition  $\wh \H$ are consistent. 
	 
	 Under the canonical parametrization provided by Assumption \ref{ass_I}, the set of indices involved in condition (\ref{cond_signal_H}) can be reduced, and the following weaker condition suffices for consistent recovery:
	  \begin{equation}\label{cond_signal_ell_q_relaxed}
	  	S_q(i,j) > 4\delta_n \quad \text{for all } i\not\group j,\ i\in I.
	  \end{equation} 
	 Note that 
	 we still allow
	 \begin{equation*}
	 S_q(i,j) \le 4\delta_n \quad \text{ for some } i\ne j, \ i,j\in J,
	 \end{equation*}
	 by recalling that  $J = [p]\setminus I$. We show, in Lemma \ref{lem_lb_score} of Appendix \ref{app_dicuss}, that variables with indices satisfying the above display are non-pure, but have {\em near-parallel} corresponding rows in $A$. 
	 By collecting these near-parallel rows with indices in $J$ in the new set 
	 \begin{equation}\label{def_J1_bar}
	 \bar J_1 ~:= ~\left\{
	 j \in J: \  S_q(j,\ell) \le 4\delta_n \ \text{ for some } \ell \in J\setminus \{j\} 
	 \right\},
	 \end{equation}
	 we have $J_1 \subseteq \bar J_1$. In fact, as the quantity $4\delta_n$ in (\ref{def_J1_bar}) originates from the estimation error of the score function $S_q$, the set $\bar J_1$ can be viewed as the sample analogue of $J_1$. 
	 
	 We begin by presenting the analogue of the results of Section \ref{sec_theory_H_K}, by giving recovery guarantees for $\wh H$ and 
	 $\wh{\mathcal{H}}$, in Corollaries \ref{cor_I_prime}, and for $\wh K$ in  Proposition \ref{prop_K}, under model (\ref{model}) satisfying the canonical parametrization given by  Assumption \ref{ass_I}. Write $\lfloor x\rfloor$ to denote the largest integer that is no greater than $x$. 
	 
	  \begin{cor}\label{cor_I_prime}
	  	Under model (\ref{model}), Assumption \ref{ass_I} and (\ref{cond_signal_ell_q_relaxed}),  on  the event $\E$, 
	    the outputs $\wh G$ and $\wh \H = \{\wh H_1, \ldots, \wh H_{\wh G}\}$ of  Algorithm \ref{alg_I}, applied  with $\delta = \delta_n$, satisfy
	  	\begin{enumerate}
	  		\item[(1)] $K\le \wh G \le K + \lfloor |\bar J_1|/2 \rfloor$;
	  		\item[(2)] $(I\cup J_1)\subseteq \wh H \subseteq (I\cup \bar J_1)$;
	  		\item[(3)] $I_k = \wh H_{\wh \pi(k)}$ for all $1\le k\le K$ for some permutation  $\wh \pi: [\wh G] \to [\wh G]$.
	  	\end{enumerate}
	  \end{cor}
	  \begin{proof}
	  	The result is a direct consequence of Theorem \ref{thm_score_ell_q} and  (\ref{cond_signal_ell_q_relaxed}) in conjunction with the fact that $\wh \H$ contains at most $\lfloor |\bar J_1|/2 \rfloor$ groups that only consist of variables in $\bar J_1$.
	  \end{proof}

	  In the simple case when $\bar J_1 = \emptyset$ (no parallel or near parallel rows beyond those corresponding to pure variables), or equivalently, 
	  \begin{equation}\label{cond_signal_ell_q_J}
	  S_q(i,j) > 4\delta_n, \quad \forall i,j\in J \text{ with }i\ne j,
	  \end{equation}
	  Corollary \ref{cor_I} below shows that Algorithm \ref{alg_I} is already sufficient for consistent estimation of pure variables. Condition (\ref{cond_signal_ell_q_J}), or equivalently $\bar J_1 = \emptyset$, is stronger than $J_1 = \emptyset$, the identifiability condition in Corollary \ref{thm_ident_I}.
	  
	\begin{cor}\label{cor_I}
		Under model (\ref{model}) and Assumptions \ref{ass_I}, assume that (\ref{cond_signal_ell_q_relaxed}) and (\ref{cond_signal_ell_q_J}) hold.
	Then, on the event $\E$,  the output $\wh G$ and $\wh \H = \{\wh H_1, \ldots, \wh H_{\wh G}\}$ from Algorithm \ref{alg_I} with $\delta = \delta_n$ satisfy: 
	$\wh G = K$, $\wh H = I$ and  $\wh H_k = I_{\pi(k)}$ for all $1\le k\le K$, for some permutation $\pi: [K] \to [K]$.
	
	\end{cor}
	
     Appendix \ref{app_dicuss} provides  insight into   conditions (\ref{cond_signal_ell_q_relaxed}) and (\ref{cond_signal_ell_q_J}) in terms of their induced restrictions on the model parameters. Summarizing, we show, under mild regularity conditions on $B\C B^\T$, that (\ref{cond_signal_ell_q_relaxed}) holds under the following mild separation condition between pure and non-pure variables,
     \[
          \min_{i\in I, j\notin I}\sin(\angle (A_{i\sbt}, A_{j\sbt})\gtrsim \delta_n. 
     \]
    
    To gain insight into  the stronger condition (\ref{cond_signal_ell_q_J}), consider the situation in which the entries of $A_{J\sbt}$ are realizations of independent sub-Gaussian random vectors. If $K\gtrsim \log|J|$, then (\ref{cond_signal_ell_q_J}) holds,  with high probability. This suggests that it may not hold  for a large $|J|$, small $K$ configuration. For this reason, 
    we present the rest of our theoretical guarantees  only under the weaker condition (\ref{cond_signal_ell_q_relaxed}).

		  


    \bigskip

	Since we do not assume (\ref{cond_signal_ell_q_J}),  we need to revisit the estimation of $K$, as $\wh G$ is no longer the best candidate. Instead, we propose to use and analyze $\wh K$,  constructed in  Section \ref{sec_est_H},  with $q = 2$.   Proposition \ref{prop_K} below states the explicit rate of the tuning parameter $\mu$ required for its  estimation. The proposition  offers  the same guarantees as  Proposition \ref{prop_K_parallel}, but  they are established under slightly different conditions, that reflect the usage of the canonical parametrization given by Assumption \ref{ass_I},  and of the  weaker condition (\ref{cond_signal_ell_q_relaxed}) enabled by this parametrization. 

	Recall that 
	$\xi_k$ is defined in Assumption \ref{ass_J_prime}. Define further
	\begin{equation}\label{def_cbI}
	c_{b,I} :=\min_{j \in I\cup \bar J_1} \|B_{j\sbt}\|_2^2,\quad \bar c_z := \lambda_1(\C), \quad c_z:= \lambda_K(\C),
	\end{equation}
	and  
	\begin{equation}\label{def_cr}
	   c_r:= \min_{i\ne j}{1\over p-2}\lambda_K\left(B_{\nij\sbt}\C B_{\nij\sbt}^\T\right).
	\end{equation}
	
	\begin{prop}\label{prop_K}
		Under model (\ref{model}) and  Assumptions \ref{ass_I} \& \ref{ass_distr}, 
		assume (\ref{cond_signal_ell_q_relaxed}) and $\log p = o(n)$.
		In addition,  suppose there exist
		absolute constants  $0<c\le C<\infty$   such that
		\begin{equation}\label{cond_regularity}
		\min(c_{b,I},c_z,c_r)>c, \qquad \bar c_z \le C
		\end{equation}
		For  
		$    \mu = 
		C'(\sqrt{K\delta_n^2}  + K\delta_n^2 +|\bar J_1|\delta_n)
		$ 
		in (\ref{est_K}), and  for a large enough constant $C'>0$,
		we have 
		\begin{equation*}
		    \lim_{n\to\infty} \PP\{\wh K \le  K\}=1.
		\end{equation*}
		If, in addition, $\max\{K\delta_n^2,  |\bar J_1|\delta_n\}=o(1)$  as $n\to\infty$,  we further have
		\begin{equation}\label{equal}\lim_{n\to\infty} \PP\{\wh K = K\}=1. 
		\end{equation}
	\end{prop}


	The proof of Proposition \ref{prop_K} is deferred to Appendix \ref{app_proof_prop_K}. 
	It relies on a careful analysis of $\|\wh M_{\wh L\wh L}-M_{\wh L\wh L}\|_{\op}$, performed when  $\bar J_1 \neq \emptyset$.

	We first prove that 
	$\wh K \le K$  always holds on the event $\E$. In conjunction with part $(1)$ of Corollary \ref{cor_I_prime}, this ensures that,  with high probability, the event $\{\wh G=\wh K\}$ implies the event $\{\wh G=K\}$. In this case, 
	our procedure stops and the output of Algorithm \ref{alg_whole} is $\wh \H$ from Algorithm \ref{alg_I}. 
	On the other hand, if $\wh K< \wh G$, as explained in Section \ref{sec_est_pure}, Algorithm \ref{alg_whole} uses the pruning step from  Algorithm \ref{alg_I_post} to estimate the pure variable index set. 
	
	Proving consistency  $\wh K = K$ requires $\mu $ to be sufficiently small so that $\mu<c_{b,I} c_z$,   which is guaranteed by $\max\{K\delta_n^2,   |\bar J_1|\delta_n\}=o(1)$. See Remark \ref{rem_cond_I_post} below for a discussion of this condition and of (\ref{cond_regularity}).\\ 

    The following theorem  gives theoretical guarantees for the output of Algorithm \ref{alg_I_post} with any $1\le r\le K$, and, in particular, the output of Algorithm \ref{alg_whole} if $r$ is set to $\wh K$. 

	\begin{thm}\label{thm_I_post}
	    Under model (\ref{model}) and Assumptions  \ref{ass_I} \& \ref{ass_distr}, assume (\ref{cond_signal_ell_q_relaxed}) and (\ref{cond_regularity}). Suppose that 
	    \begin{align}\label{cond_Kdelta}
	    &   \lim_{n\to\i}  K\delta_n^2 = 0
	    \end{align}
	    and, if $\bar J_1 \ne \emptyset$, assume that there exist absolute constants $0< C_0<\infty$ and   $0<\eps<1$ such that  
	   \begin{align}
	         &		{\max_{1\le k\le K}\xi_k 
			} \le C_0	\min_{1\le k\le K}\xi_k, \label{cond_pure_ratio}\\
			&	    \max_{j\in \bar J_1}\left(	\sum_{k=1}^K {|A_{jk}| \over  \xi_k} \right)^2  <   {1-\eps }.\label{ass_J1_prime}	 
	    \end{align}
	   Then,   there exists some  permutation $\pi: [K] \to [K]$ such that the output $\wh \I$ of Algorithm \ref{alg_I_post} by using any $1\le r\le K$ satisfies
	   \[ \lim_{n\to\infty} \PP\left\{ \wh I_a = I_{\pi(a)} \, \text{ for all } 1\le a\le r\right\}=1.\]
	   %
        In particular, the claim is valid for $r=\wh K$ defined in (\ref{est_K}). In this case, if additionally $|\bar J_1|\delta_n = o(1)$, then with probability tending to one as $n\to \i$, $\wh K= K$ and the output $\wh \I$ from Algorithm \ref{alg_whole} satisfies $\wh I_a = I_{\pi(a)}$, for all $1\le a\le K$.
	\end{thm}
	
	As already mentioned, the statement of Theorem \ref{thm_I_post} is formulated in terms of  $r$ groups, for any  $1\le r\le K$. In this way, in case we use some estimator that under-estimates $K$, Theorem \ref{thm_I_post} still ensures that a subset of the groups consisting in  pure variable indices (corresponding to the largest conditional variances) can be consistently estimated. When $K$ is consistently estimated, consistent estimation of the entire  pure variable index set, and of  its partition,  is guaranteed.\\

			
	Allowing for a general  
	statement in  Theorem \ref{thm_I_post}, which is valid for  a general $1\le r\le K$, brings technical difficulty to its proof, which is stated in Appendix \ref{app_proof_thm_I_post}. One of the main challenges is to control the difference between the estimated Schur complement $\wh \Theta_{jj|S_k}$ and its population-level counterpart, uniformly for all $j\in \wh H$, all $S_k$ selected from (\ref{def_i_k}) and all $1\le k\le r$. The difficulty is further elevated by the fact that  $\bar J_1 \neq \emptyset$. We provide this uniform control in Lemma \ref{lem_noise} of Appendix \ref{app_proof_thm_I_post} and its proof further relies on the uniform controls of the sup-norm of $\wh M_{\wh H\wh H} - M_{\wh H\wh H}$, the operator norm of $\wh M_{S_kS_k} - M_{S_kS_k}$ and the quadratic form $\wh M_{iS_k}^\T \wh M_{S_kS_k}^{-1}\wh M_{S_ki}$,
	collected in Lemmas \ref{lem_Gamma_M_Ihat}, \ref{lem_op_M} and \ref{lem_M_schur} of Appendix \ref{app_lemmas_M}.

	\begin{remark}[On the conditions of Proposition \ref{prop_K} and  Theorem \ref{thm_I_post}]\label{rem_cond_I_post}
    	The proofs for both results are non-asymptotic,
    	in the sense that their statements continue to hold when   $K\delta_n^2 \le c$,  for some sufficiently small constant $c$, 
    	instead of the assumed $K\delta_n^2 = o(1)$ in (\ref{cond_Kdelta}). Then,  they would   hold on an event with probability $1-c'(p\vee n)^{-c''}$ for some constants $c',c''>0$. To avoid the delicate interplay between different constants, 
    	we opted for the current asymptotic formulation. 
    	
    	Condition (\ref{cond_Kdelta}) requires  $K$  not to grow too fast relative to $n$, $K\log(p\vee n) = o(n)$.  Condition $|\bar J_1|\delta_n = o(1)$ in Proposition \ref{prop_K} allows the size of the  index set corresponding   to  near-parallel rows in $A$ to grow,  but slower than $\sqrt{n/\log(n\vee p)}$. 
    	
    	Conditions (\ref{cond_pure_ratio}) and (\ref{ass_J1_prime}) are only needed when $\bar J_1 \neq \emptyset$.  
	   Condition (\ref{ass_J1_prime}) is the analogue of Assumption \ref{ass_J_prime}  
	    and 
	    allows us to distinguish, at the sample level, pure variables from non-pure variables corresponding to rows in $A$ that  are close to parallel (in the sense $S_q(i,j)<4\delta_n$).
	    Condition (\ref{cond_pure_ratio}) allows us to eventually separate two pure groups, $I_a$ and $I_b$ with $a\ne b$, from each other. It prevents the pure variable variances  from being very different. 
	    For instance,  (\ref{cond_pure_ratio}) holds if $\max_{k}\max_{i\in I_k}\Sigma_{ii} \le C\min_{k}\max_{i\in I_k}\Sigma_{ii}$ coupled with (\ref{cond_regularity}).
        Finally, condition   (\ref{cond_regularity}) is a mild regularity condition on the matrix $B$ and $\C$, which is needed in our rather technically involved proofs.  More discussion of this condition can be found in Appendix \ref{app_dicuss}.
	\end{remark}

\begin{remark}	
	As a byproduct of our proof, we establish the deviation inequalities in operator norm for the empirical sample correlation matrix of $n$ independent sub-Gaussian random vectors in $\RR^p$ with $p$ allowed to exceed $n$. We refer to Appendix \ref{app_concentration_R} for  its formal statement, and only state a simplified result below. Suppose that  $\Sigma^{-1/2}\X_{i\sbt}$ are independent sub-Gaussian random vectors and $\log p \lesssim n$. We prove  that, with probability at least $1 - 2 \exp(-ct) - 4(p\vee n)^{-1}$,  
        \[
            \|\wh R - R\|_{\rm op} \le C \|R\|_{\rm op}\left(
                \sqrt{p\over n} + {p\over n} + \sqrt{t\over n} + {t\over n} + \sqrt{\log n \over n}
            \right), \quad \forall\ t\ge 0.
       \] 
    Similar deviation inequalities for the sample covariance matrix have been well understood, for instance,  \cite{vershynin_2012,lounici2014,bunea2015,koltchinskii2017}. However, the operator norm concentration inequalities of the sample correlation matrix is relatively less explored. \cite{elkaroui2009} studied the asymptotic behaviour of the limiting spectral distribution of the sample correlation matrix when $p/n\to (0,\i)$. \cite{wegkamp2016, han2017} prove a similar result for  Kendall's tau sample correlation matrix. 
     \end{remark}

	\subsection{Estimation of $A$ under the canonical parametrization}\label{sec_est_A}
	We first estimate $B = D_{\Sigma}^{-1/2} A$ by the estimator $\wh B$ presented below. Once this is done,  $A$ is easily estimated by 
	\begin{equation}\label{est_A}
	\wh A = \wh D_{\wh \Sigma}^{1/2}\wh B,\qquad 
	\text{with} \quad \wh D_{\wh \Sigma} = \diag(\wh \Sigma_{11},\cdots, \wh \Sigma_{pp}).
	\end{equation}
	
	\subsubsection{Estimation of $B_{I\cdot}$ and $\C$}\label{sec_est_AI}
Since Assumption \ref{ass_I} and $\diag(\C) = {\bm 1}$ imply that 	$M_{ii} = B_{ik}^2$ for any $i\in I_k$, and given the estimated partition of the pure variables $\wh \I = \{\wh I_1, \ldots, \wh I_{\wh K}\}$, we propose to estimate $B_{I\sbt}$ by 
	\begin{equation}\label{est_BI}
	\wh B_{ik}^2 = \wh M_{ii},\quad \forall i\in \wh I_k,~ k\in [\wh K]
	\end{equation}
	with $\wh M_{ii}$ defined in (\ref{est_M_LL_diag}). 
	Since we can only identify $B_{I\sbt}$ up to a signed permutation matrix, we use the convention that $i_k$ is the first element in $\wh I_k$.  Using the rationale  in (\ref{ident_BI_sign}), in the proof of Theorem \ref{thm_ident_I}, we set 
	\begin{equation}\label{est_BI_sign}
	\sgn(\wh B_{i_kk}) = 1,\qquad \sgn(\wh B_{jk}) = \sgn(\wh R_{i_kj}),\qquad \forall  j\in \wh I_k \setminus\{i_k\}.
	\end{equation}
	We estimate the diagonal of $\C$ by $\diag(\whC) = {\bm 1}$ and, following (\ref{ident_Gamma_II}) and (\ref{ident_C}), we estimate the off-diagonal elements of $\C$ by the corresponding entries of
	\begin{equation}\label{est_C}
	\wh B_{\wh I \sbt}^+\left(
	\wh R_{\wh I \wh I} - \wh \Gamma_{\wh I \wh I}\right)\left[\wh B_{\wh I \sbt}^{+}\right]^{\T}
	\end{equation}
	with $\wh B_{\wh I\sbt}^+ = \bigl[\wh B_{\wh I\sbt}^\T \wh B_{\wh I\sbt}\bigr]^{-1}\wh B_{\wh I\sbt}^\T$ being the left inverse of $\wh B_{\wh I\sbt}$, and $\wh \Gamma_{\wh I\wh I}$ estimated from (\ref{est_Gamma_II}).

	\bigskip

	We provide theoretical guarantees for the estimated loadings of the pure variables  and  of the covariance matrix $\C$ of $Z$, obtained from the above procedure for any $q\ge 2$. The  proofs are  deferred to Appendix \ref{app_proof_sec_est_A}.

	\begin{thm}\label{thm_BI}
	    Under the conditions of Theorem \ref{thm_I_post}, assume $|\bar J_1| \delta_n = o(1)$.  
		With probability tending to one, as $n\to \i$, 
		we have
		\begin{eqnarray*}
			&& \min_{P\in \P_K}\max_{i\in I_k} \|\wh B_{i\sbt} - PB_{i\sbt }\|_\i \lesssim \delta_n,\\ 
			&&  \min_{P\in \P_K}\| \whC - P^\T\C P \|_\i \lesssim \delta_n,\\
			&& \min_{P\in \P_K}\max_{i\in I_k} \|\wh A_{i\sbt} - PA_{i\sbt }\|_\i \lesssim \sqrt{\Sigma_{ii}}\ \delta_n.
		\end{eqnarray*}
		The minimum is taken over 
		the set $\P_K$ of $K\times K$ signed permutation matrices.  
	\end{thm}
	
	Similar to Theorem \ref{thm_I_post}, the results of Theorem \ref{thm_BI} can be easily stated non-asymptotically, 
	as explained in Remark \ref{rem_cond_I_post}.

	\subsubsection{Estimation of the submatrix $B_{J \cdot }$}\label{sec_est_AJ}
	
	Since the correlation matrix $R$ takes the form $B\C B^\T +\Gamma$, we can write  
	\begin{eqnarray*}
		\begin{bmatrix}
			R_{I I} & R_{I J}\\
			R_{J I}  & R_{J J}
		\end{bmatrix} &=& \begin{bmatrix}
			B_{I \sbt}\C B_{I \sbt} ^\T& B_{I \sbt}\C B_{J \sbt} ^\T
			\\
			B_{J \sbt}\C B_{I \sbt} ^\T& B_{J \sbt}\C B_{J \sbt} ^\T
		\end{bmatrix} + \begin{bmatrix}
			\Gamma_{II} & \\ & \Gamma_{JJ}
		\end{bmatrix}.
	\end{eqnarray*}
	In particular, from $R_{I J} = B_{I \sbt}\C B_{J \sbt} ^\T$, the submatrix $B_{J\sbt} $ is solved via
	\begin{eqnarray*} B_{J\sbt} ^\T &=&
		\C^{-1} [ B_{I\sbt}^\T B_{I\sbt} ]^{-1} B_{I\sbt}^\T\  R_{IJ}.
	\end{eqnarray*}
	Hence, after estimating $\C$ and $B_{\wh I\sbt}$, using the estimate $\wh R$, we obtain a plug--in estimate for $B_{\wh J\sbt}$ with $\wh J= [p]\setminus \wh I$.
	Alternatively, we can regress $\wh R_{\wh I \wh J}$ on $\wh B_{\wh I\sbt} \wh \Sigma_Z$, or  $[\wh B_{\wh I\sbt}^\T\wh B_{\wh I\sbt} ]^{-1} \wh B_{\wh I\sbt}^\T \wh R_{\wh I \wh J}$ on
	$\wh \Sigma_Z$. This approach has been adopted in \cite{bing2017adaptive}; \cite{Top}, and allows for incorporating sparsity restrictions on $B$ and hence on $A$, and ultimately can result in  estimation at minimax-optimal rates. 
	We refer to the aforementioned work for the details of estimating the entire assignment matrix $A$ under sparsity, and subsequent theoretical analysis. We mention that as soon as an estimator of $A_{I\sbt}$ and of $\C$ are found, the full estimation of $A$, and its  guarantees, are  almost identical to those  proposed by \cite{bing2017adaptive}; \cite{Top}, and are therefore omitted. 
	

	\section{Practical considerations \& simulations}\label{sec_sims}

	We discuss the selection of tuning parameters in Section \ref{sec_cv} and the pre-screening procedure in Section \ref{sec_pre_screen}. The results of extensive simulations are presented  in Sections \ref{sec_sim_only} and \ref{sec_pre_screen}.
	
	\subsection{Selection of tuning parameters}\label{sec_cv}
	
	Algorithm \ref{alg_whole} requires two tuning parameters: $\delta$ in Algorithm \ref{alg_I} for estimating $H$ and $\mu$ in (\ref{est_K}) for estimating $K$. We now explain how to choose these two tuning parameters, in turn.  
	
	\paragraph{Selection of $\delta$.}
	We discuss how to select the leading constant $c$ of $\delta(c) = c \sqrt{\log (p\vee n) / n}$ from Algorithm \ref{alg_I} in a fully data-driven way. We first split the data into two parts with equal sample size. Denote their corresponding correlation matrices as $\wh R^{(1)}$ and $\wh R^{(2)}$. 
	
	Given a pre-specified grid of $\mathcal{C} = \{c_1, \ldots, c_N\}$, for each $\ell \in [N]$, we estimate the set of parallel rows, $\wh H_{(\ell)}$, and its partition $\wh \H_{(\ell)}$ by applying Algorithm \ref{alg_I}  to $\wh R^{(1)}$ with $\delta(c_\ell) = c_{\ell}\sqrt{\log (p\vee n) / n}$. Then $\wh B_{\wh H_{(\ell)}\sbt}$ and $\whC$ are estimated from (\ref{est_BI}), (\ref{est_BI_sign}) and (\ref{est_C}). Finally, we compute
	\begin{equation}\label{def_L}
		L(\ell) = \left\|
			\wh R^{(2)} - \wt R_{(\ell)}^{(1)}
		\right\|_{\textrm{F-off}}
	\end{equation}
	where $\|\cdot \|_{\textrm{F-off}}$ denotes the Frobenius norm of all off-diagonal entries and 
	\[
		\wt R_{(\ell)}^{(1)} = \begin{bmatrix}
			\wh B_{\wh H_{(\ell)}\sbt}\whC \wh B_{\wh H_{(\ell)}\sbt}^\T & \wh R^{(1)}_{\wh H_{(\ell)} \wh H^c_{(\ell)}}\\
			\wh R^{(1)}_{\wh H^c_{(\ell)} \wh H_{(\ell)}} &  \wh R^{(1)}_{\wh H^c_{(\ell)} \wh H_{(\ell)}}\left[\wh B_{\wh H_{(\ell)}\sbt}\whC \wh B_{\wh H_{(\ell)}\sbt}^\T\right]^{-} \wh R^{(1)}_{\wh H_{(\ell)} \wh H^c_{(\ell)}}
		\end{bmatrix}
	\]
	with $M^-$ being the Moore-Penrose pseudo-inverse of any matrix $M$ and $\wh H^c_{(\ell)} = [p] \setminus \wh H_{(\ell)}$. The final $c_{\wh \ell}$ is chosen as 
	$
		\wh \ell = \arg\min_{1\le \ell \le N}L(\ell).
	$
	
	Since Lemma \ref{lem_I_hat} shows that $|\wh H(\delta)|$ is monotone in $\delta$, the range of the initial grid $\mathcal{C}$ can be chosen such that 
	$\delta(c_1) = \min_{i \ne j} \wh S_q(i,j)$ and 
	$\delta(c_N) = \max_{i\ne j}\wh S_q(i,j)$. This choice guarantees $2\le |\wh H(\delta)| \le p$ for all $\delta = \delta(c_\ell)$. When one has prior information of the range of $|H|$, the choice of $c_1$ and $c_N$ can be adapted correspondingly. In practice, we also recommend to adapt the above framework to $10$-fold cross-validation. 
	
To motivate the criterion $L(\ell)$ in (\ref{def_L}), note that if some $\ell$ yields good estimates  of $H$ and $\H$, then 
	$[\wt R^{(1)}_{(\ell)}]_{\wh H_{(\ell)}\wh H_{(\ell)}}$ is close to $B_{H\sbt} \C B_{H\sbt}^\T$. As a result 
	$
		[\wt R^{(1)}_{(\ell)}]_{\wh H^c_{(\ell)}\wh H^c_{(\ell)}}
	$
	is close to 
	$$
	R_{H^c H}(B_{H\sbt} \C B_{H\sbt}^\T)^- R_{HH^c} 
	= B_{H^c\sbt} \C B_{H^c \sbt}^\T.
	$$
	Therefore, the overall loss $L(\ell)$ should be small, as $\wh R^{(2)}$ can be well approximated by $\wt R_{(\ell)}^{(1)}$.
	
	\paragraph{Selection of $\mu$.}
	We discuss the selection of $\mu$ used in (\ref{est_K}). With $\wh \H$ given from Algorithm \ref{alg_I}, in light of the theoretical rate given in Proposition \ref{prop_K_parallel}, we propose to choose 
	$$
	    \mu = c \delta_n \left(\sqrt{\wh G }  + \wh G\delta_n\right)
	$$ 
	where $\wh G = |\wh \H|$, $\delta_n$ is selected from the procedure discussed above and the constant $c$ is chosen from a specified grid $[c_l, c_u]$ (our extensive simulation suggests to use $c_l = 0.1$ and  $c_u = 0.5$ with increment equal to $0.02$). We select $c$ by minimizing the prediction error on a validation dataset. Specifically, let $\wh \H$ and $\wh L$ be obtained, respectively, from Algorithm \ref{alg_I} and (\ref{def_L_hat}). We randomly split the data into two parts of equal sizes. For each $c$ in the grid, we compute $\wh M_{\wh L\wh L}^{(1)}$ and $\wh K$ from (\ref{est_M_LL_off_diag}) -- (\ref{est_M_LL_diag}) and (\ref{est_K}) by using the first dataset. Based on the spectral decomposition   \[
	\wh M_{\wh L\wh L}^{(1)}=\sum_k\sigma_k u_ku_k^\T\]
	with eigenvalues $\sigma_1 \ge \sigma_2 \ge \cdots$,
	we choose $c$ that minimizes the risk $\|\wt M - \wh M_{\wh L\wh L}^{(2)}\|_F^2$. Here $\wt M = \sum_{k=1}^{\wh K}\sigma_ku_ku_k^\T$  and $\wh M_{\wh L\wh L}^{(2)}$ is computed from (\ref{est_M_LL_off_diag}) -- (\ref{est_M_LL_diag}) by using the second dataset.\\
	
	Alternatively, instead of considering different values of $\mu$, one can directly choose  $k\in[\wh G]$ which minimizes the prediction loss. Specifically, we compute $\wt M$ by using the first dataset for $1\le k \le \wh G$ and then choose $k$ based on the smallest $\|\wt M - \wh M_{\wh L\wh L}^{(2)}\|_F^2$.\\
	
	Both procedures require the computation of a spectral decomposition only once.
	Our simulation shows these two procedures have almost the same performance.

	\subsection{Simulation study}\label{sec_sim_only}
	In this section, we conduct extensive simulations to support our theoretical findings from two aspects: the recovery of pure variables and the estimation errors of $A_{I\sbt}$, $\C$ and $[\Sigma_E]_{II}$. 
	We start by describing the metrics used to quantify these aspects.

	\paragraph{Metrics:}
	We first introduce the metrics that quantify the estimation performance  of the pure variable index set and respective partition. Recall that, for any estimate $\wh \I$ of $\I$, we need to evaluate both its estimation quality, and also that of  $\wh I = \cup_{k}\wh I_k$ as an estimate of $I$.  Towards this end, the first two metrics that we consider are the true positive rate (TPR) and the true negative rate (TNR), defined as  
	\begin{equation*}
		{\rm TPR} =  {|\wh I \cap I| \over |I|},\qquad {\rm TNR} = {|\wh J\cap J| \over |J|}.
	\end{equation*}
	Here, $I$ is the true set of pure variables, $J = [p]\setminus I$ and $\wh J := [p]\setminus \wh I$. TPR and TNR quantify the performance of the estimated set $\wh I$ but do not reflect how well we recover the partition of pure variables. %
	The pairwise approach of  \cite{wiwie2015comparing} measures how well two partitions are aligned with each other. This approach also allows $\wh K = |\wh \I|$ to be different from $K$. 
	For any pair $1\leq j<k\leq d := |\wh I \cup I|$, define 
	\begin{align*}
	{\rm TP}_{jk} &= 1\left\{ \textrm{$j,k\in I_a$ and $j,k\in \wh I_b$ for some $1\leq a\leq K$ and $1\leq b\leq \wh K$}\right\},\\
	{\rm TN}_{jk} &=1\left\{ \textrm{$j,k\notin I_a$ and $j,k\notin  \wh I_b$ for any $1\leq a\leq K$ and $1\leq b\leq \wh K$}\right\},\\
	{\rm FP}_{jk}&=1\left\{\textrm{$j,k\notin I_a$ for any $1\leq a\leq K$ and $j,k\in \wh I_b$ for some  $1\leq b\leq \wh K$}\right\},\\
	{\rm FN}_{jk} &=1\left\{\textrm{$j,k\in I_a$ for some $1\leq a\leq K$ and $j,k\notin \wh I_b$ for any $1\leq b\leq \wh K$}\right\}.
	\end{align*}
	and we define 
	\begin{align*}
	{\rm TP}&=\sum_{1\leq j<k\leq d} {\rm TP}_{jk},\quad  {\rm TN}=\sum_{1\leq j<k\leq d} {\rm TN}_{jk},\\
	{\rm FP}&=\sum_{1\leq j<k\leq d} {\rm FP}_{jk},\quad {\rm FN}=\sum_{1\leq j<k\leq d} {\rm FN}_{jk}.
	\end{align*}
	We use sensitivity (SN) and specificity (SP) to evaluate the performance of $\wh \I$, defined as
	\begin{equation}\label{eqspsn}
	{\rm SP=\frac{TN}{TN+FP}}, \qquad {\rm SN=\frac{TP}{TP+FN}}.
	\end{equation}
	To measure the errors of estimating $A_{I\sbt}$ and $\C$ (and indirectly $[\Sigma_E]_{II}$), we use 
	$$
	    {\rm Err}(A)= {1\over |\wh I|}\left\|\wh A_{\wh I\sbt}\wh A_{\wh I\sbt}^\T - A_{\wh I\sbt}A_{\wh I\sbt}^\T\right\|_F, \quad {\rm Err}(\C)= {1\over \wh K}\left\|\wh \Sigma_Z - \C\right\|_F.
	 $$
	 Note that Err$(\C)$ is only well-defined when $\wh K = K$, whereas Err$(A)$ allows $\wh I$ to differ from $I$.

	\paragraph{Methods:} 

    Throughout the simulation, we consider the PVS procedure given in Algorithm \ref{alg_whole} with $q = 2$.\footnote{The same procedure, with  $q = \i$,  is not included as it has similar performance but is much more computationally demanding. The supplement contains more simulation results regarding  the numerical comparison between $q = 2$ and $q = \i$.} The tuning parameters $\delta$ and $\mu$ are selected by the cross-validation approach outlined in Section \ref{sec_cv}.\footnote{We do not distinguish between the different possible  selections of $\mu$ described in Section \ref{sec_cv},  as they  yield almost the same results. The presented results are based on the second criterion of enumerating $1\le k\le \wh G$ directly. 
    } We compare its performance  with that of the LOVE algorithm proposed in \cite{bing2017adaptive}, developed for factor models in which  pure variables in the same partition group have equal factor loadings, an assumption relaxed  in this work.

	\paragraph{Data generating mechanism:} We first describe how we generate $\C$ and $\Sigma_E$. The number of latent factor $K$ is set to $10$. 
	The matrix $\C$ is generated as $\diag(\C) = {\bm 1}$ and 
	$[\C]_{ab} = (-1)^{a+b}\rho_Z^{|a-b|}$ for any $a,b\in [K]$, where the parameter $\rho_Z \in [0,1)$ controls the correlation of the latent factor $Z$. $\Sigma_E$ is a diagonal matrix with diagonal elements independently sampled from Unif$(1,3)$.
	
	We generate $A$ as follows.  First, we 
	generate a homogeneous matrix $\wt A$ with the same support as $A$ and satisfies $|\wt A_{i\sbt}| = \e_k$ for all $i\in I_k$, $k\in [K]$ and $\|\wt A_{j\sbt}\|_1\le 1$ for all $j\in J$. 
	Specifically, we consider the following configuration of signs for pure variables in each cluster:  $(3,2)$, $(4,1)$, $(2,3)$, $(1,4)$ and $(5,0)$, with the convention that the first number denotes the number of positive pure variables in that group and the second one denotes the number of negative pure variables. Among the 10 groups, each sign pattern is repeated twice. We construct  $\wt A_{J\sbt}$ by  independently sampling each entry from $N(0,1)$ and normalize each row to have $\ell_1$-norm equal to 1.   
	Finally, we create a heterogeneous $A$ via $A:= D_A \wt A$ using 
	$D_A = \diag(d_1, \ldots, d_p)$ with 
	\[
		d_j =  \alpha{ p \ v_j^\eta  \over  \sum_{i = 1}^p v_i^\eta }\qquad {\rm and } \qquad  v_1,\ldots, v_p \overset{i.i.d.}{\sim}  \text{Unif}(1, 2).
	\]
	The parameter $\alpha > 0$ controls the overall scale of $A$ as $p^{-1}\sum_{j=1}^p d_j = \alpha$. Hence the larger $\alpha$, the stronger signal we have. The parameter $\eta \ge 0$ controls the level of heterogeneity of the rows of $A$. The larger $\eta$, the more heterogeneous $A$ is. The case $\eta = 0$ corresponds to the homogeneous $A$, that is, $A = \alpha \wt A$. 
	
	Finally, the data matrix $\X$ is generated from $\X=\Z A^\T +\bE$ where the rows of  $\Z$ and $\bE$ are i.i.d. samples from $N(0, \C)$ and $N(0, \Sigma_E)$, respectively. For each setting in the simulation studies below, we
	repeat generating the data matrix 100 times and 
	report the averaged metrics that we introduced earlier based on 100 repetitions.

	\subsubsection{Performance of estimating the pure variables}
	
	We compare PVS with LOVE in terms of TPR, TNR, SP and SN under various scenarios. To illustrate the dependency on different parameters, we vary $n$, $p$, $\alpha$, $\rho_Z$ and  $\eta$ one at a time. We fix $n = 300$, $p = 500$, $\alpha = 2.5$, $\rho_Z = 0.3$ and $\eta = 1$, as a baseline setting,  when they are not varied.

	\subsubsection*{Varying $n$ and $p$ one at a time}
	To study the performance of PVS and LOVE under different combinations of $n$ and $p$, we first vary $p \in \{100, 300, 500, 700, 900\}$ with $n=300$, and then vary $n\in \{100, 300, 500, 700, 900\}$ with $p=500$. Figure \ref{fig_np} depicts the TPR, TNR, SP and SN metrics of PVS and LOVE in each setting.  
	
	PVS consistently recovers both the set and the partition of pure variables.  LOVE fails to control TPR and SN. Its often lower TPR means that some pure variables are not selected. The undercount of pure variables further leads to a low SN. The low performance of LOVE is to be expected when $\eta = 1$ as it has theoretical guarantees only when the rows of $A$ are homogeneous.  
	The performance of PVS improves further as $n$ increases, but it isn't  impacted much as $p$ gets larger, in line with 
	Theorem \ref{thm_score_ell_q} and Corollary \ref{cor_I}.

	\begin{figure}[ht]
		\centering
		\begin{tabular}{cc}
			\includegraphics[width = .4\textwidth]{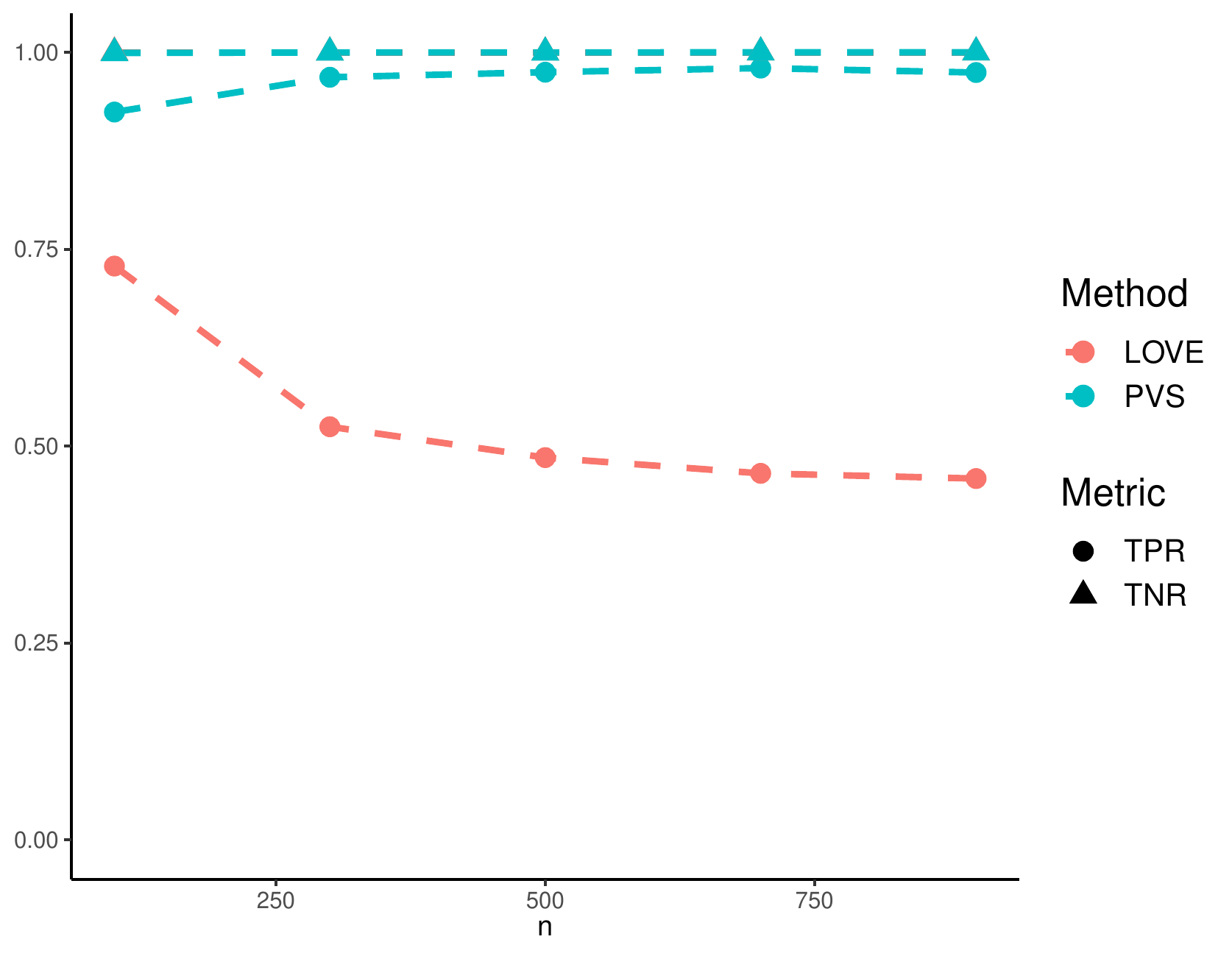} & 
			\includegraphics[width = .4\textwidth]{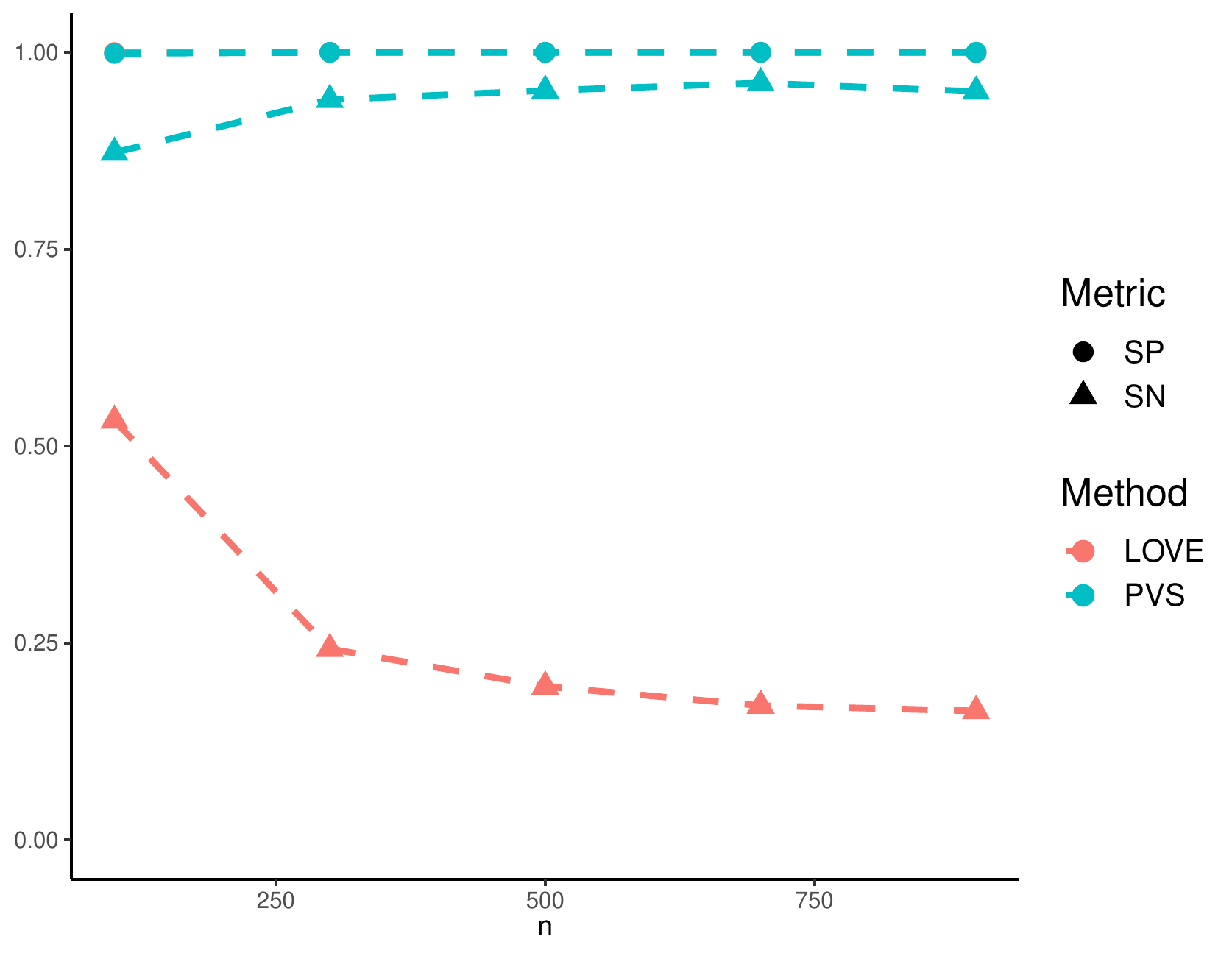}\\
			\includegraphics[width = .4\textwidth]{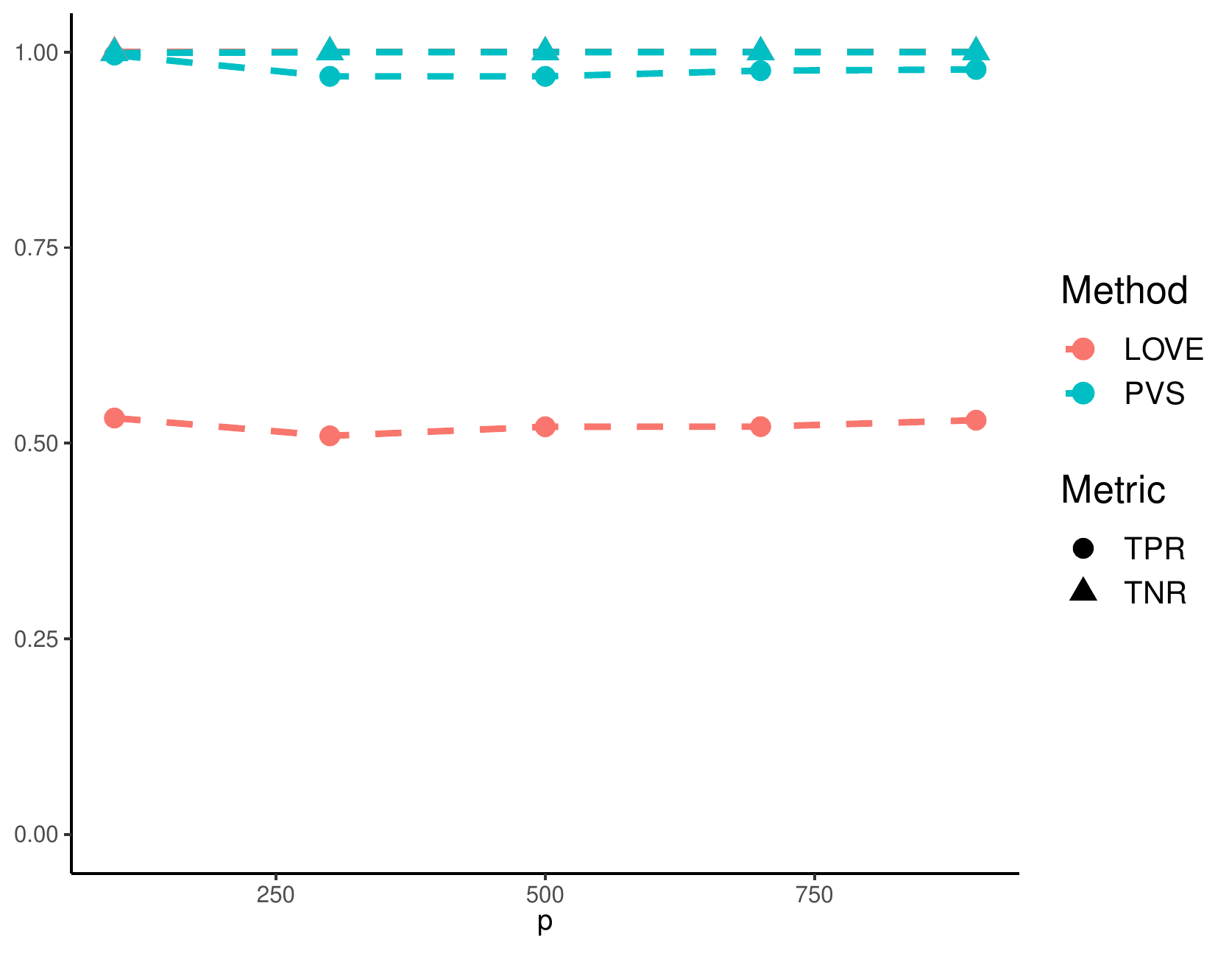} & 
			\includegraphics[width = .4\textwidth]{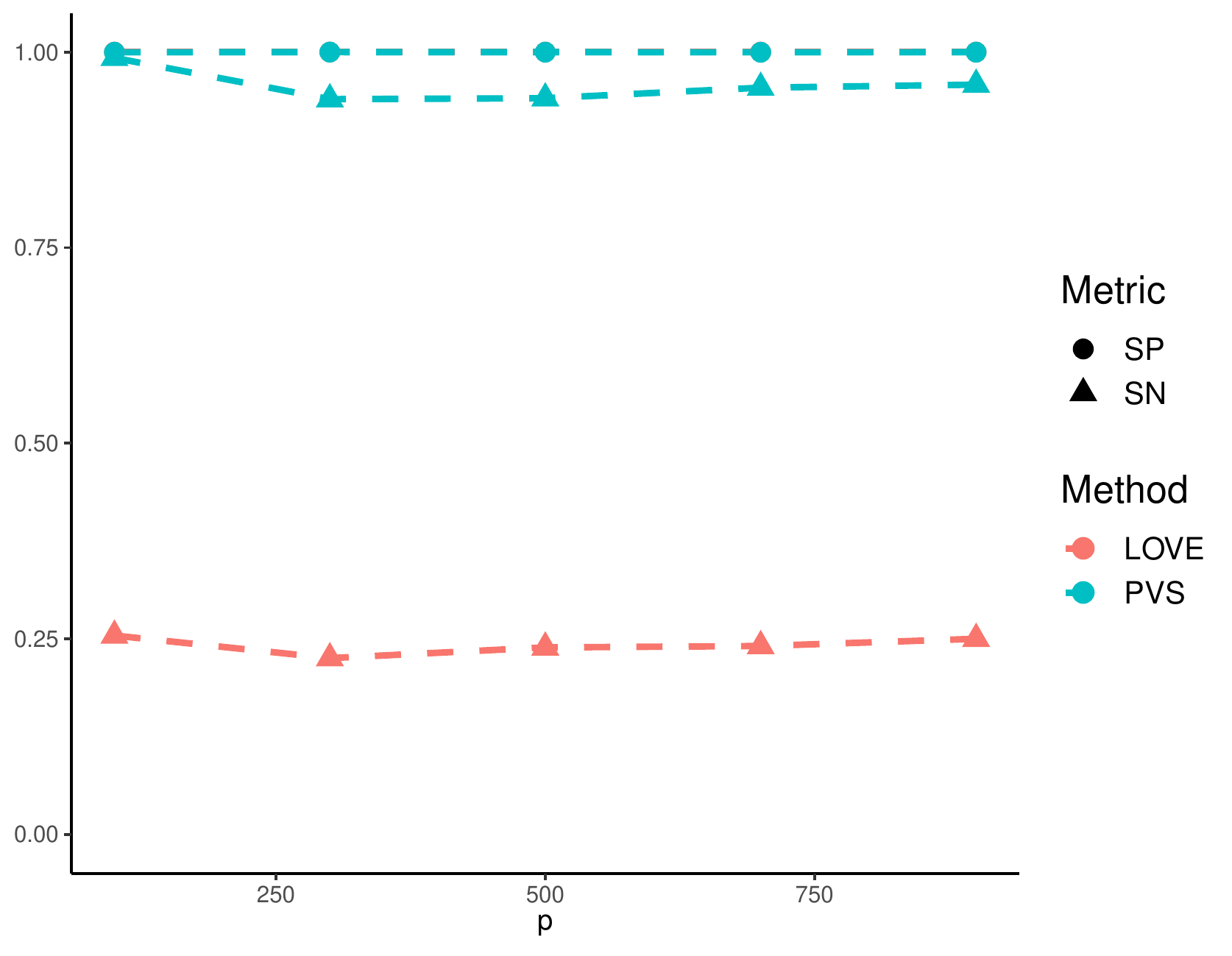}
		\end{tabular}
		\caption{Performance of PVS and LOVE as $n$ or $p$ varies.}
		\label{fig_np}
	\end{figure}

	\subsubsection*{Varying  $\alpha$ and $\rho_Z$ to change the signal strength}

	In this section we study how the performance of PVS changes for various levels of the signal strength. Recall from Section \ref{sec_theory_H_K} that $\min_i\|A_{i\sbt}\|_2$ and $\lambda_K(\C)$ jointly contribute to the signal strength. Further recall that $\alpha$ controls the overall scale of $A$ while $\rho_Z$ controls $\lambda_K(\C)$. We first set $\rho_Z = 0.3$ with $\alpha$ varying within $\{1.5, 2, 2.5, 3, 3.5\}$, then set $\alpha = 2.5$ and vary $\rho_Z \in \{0, 0.2, 0.4, 0.6, 0.8\}$. 
	
    From Figure \ref{fig_alpha_rhoZ}, PVS has near-perfect performance as long as $\alpha > 2$ and $\rho_Z < 0.8$. Its performance gets affected when $\alpha$ is too small or the latent factor $Z$ is too correlated. In particular, PVS tends to miss a few pure variables while still has a perfect TNR and SP.
	
	\begin{figure}[ht]
		\centering
		\begin{tabular}{cc}
			\includegraphics[width = .4\textwidth]{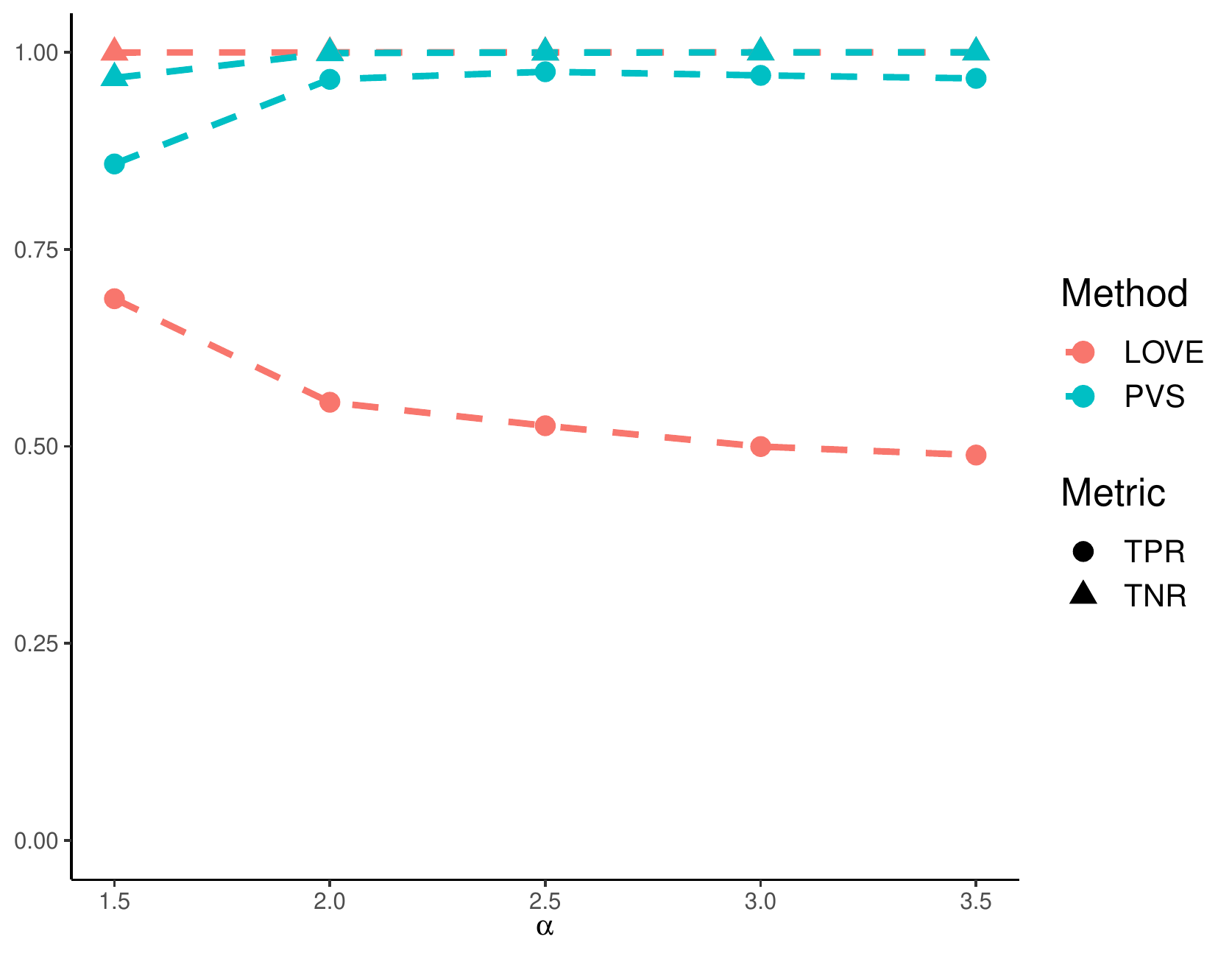} & 
			\includegraphics[width = .4\textwidth]{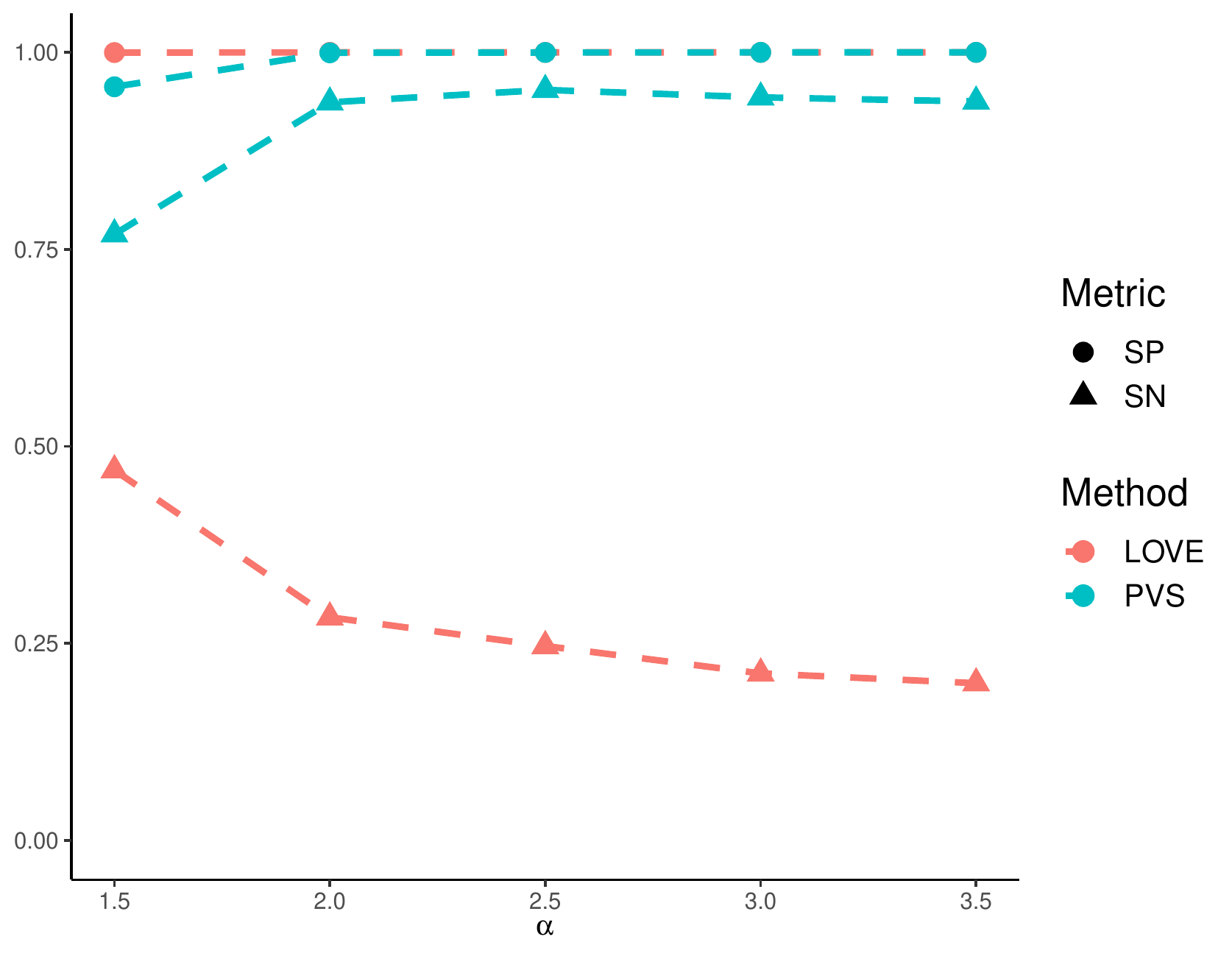}\\
			\includegraphics[width = .4\textwidth]{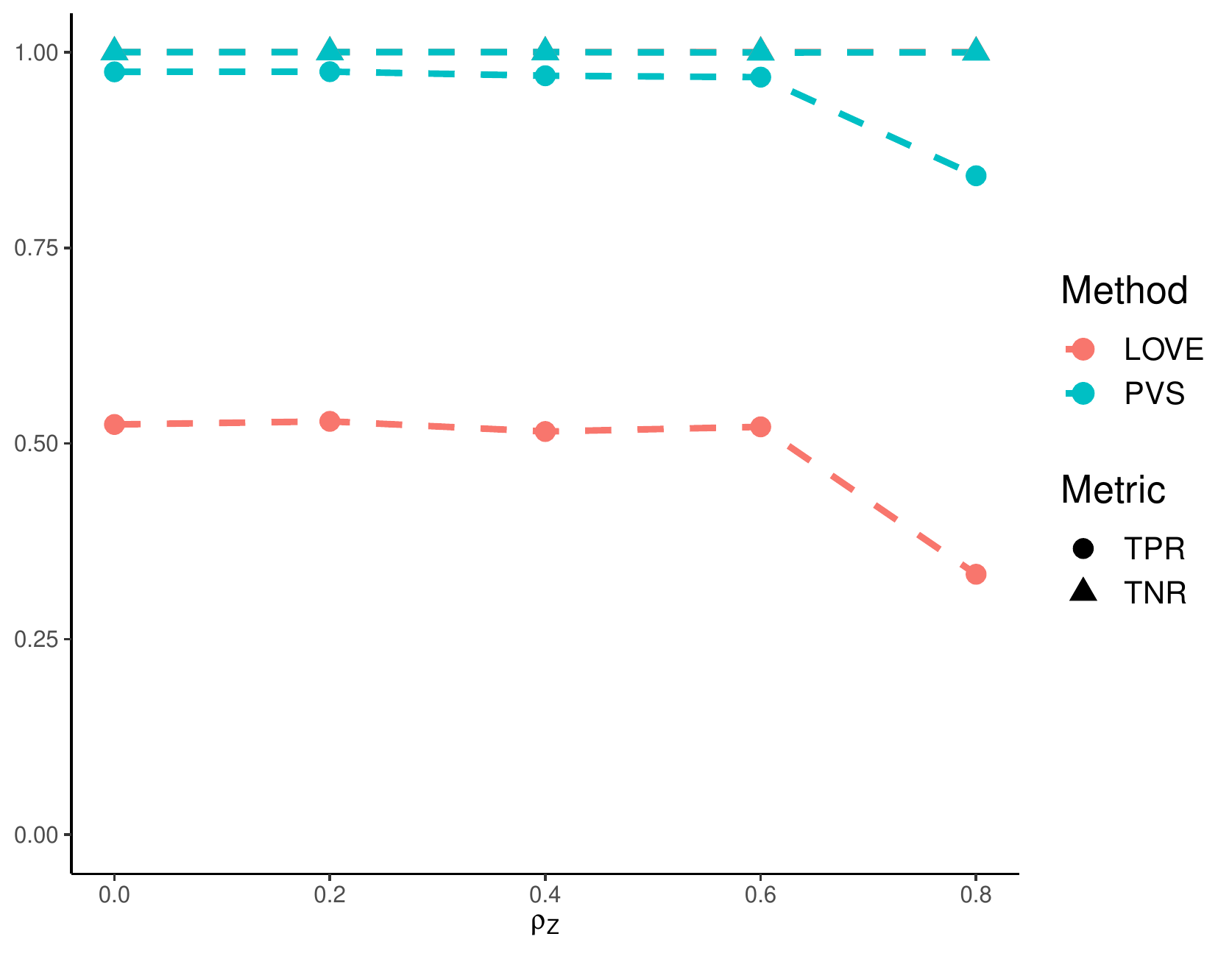} & 
			\includegraphics[width = .4\textwidth]{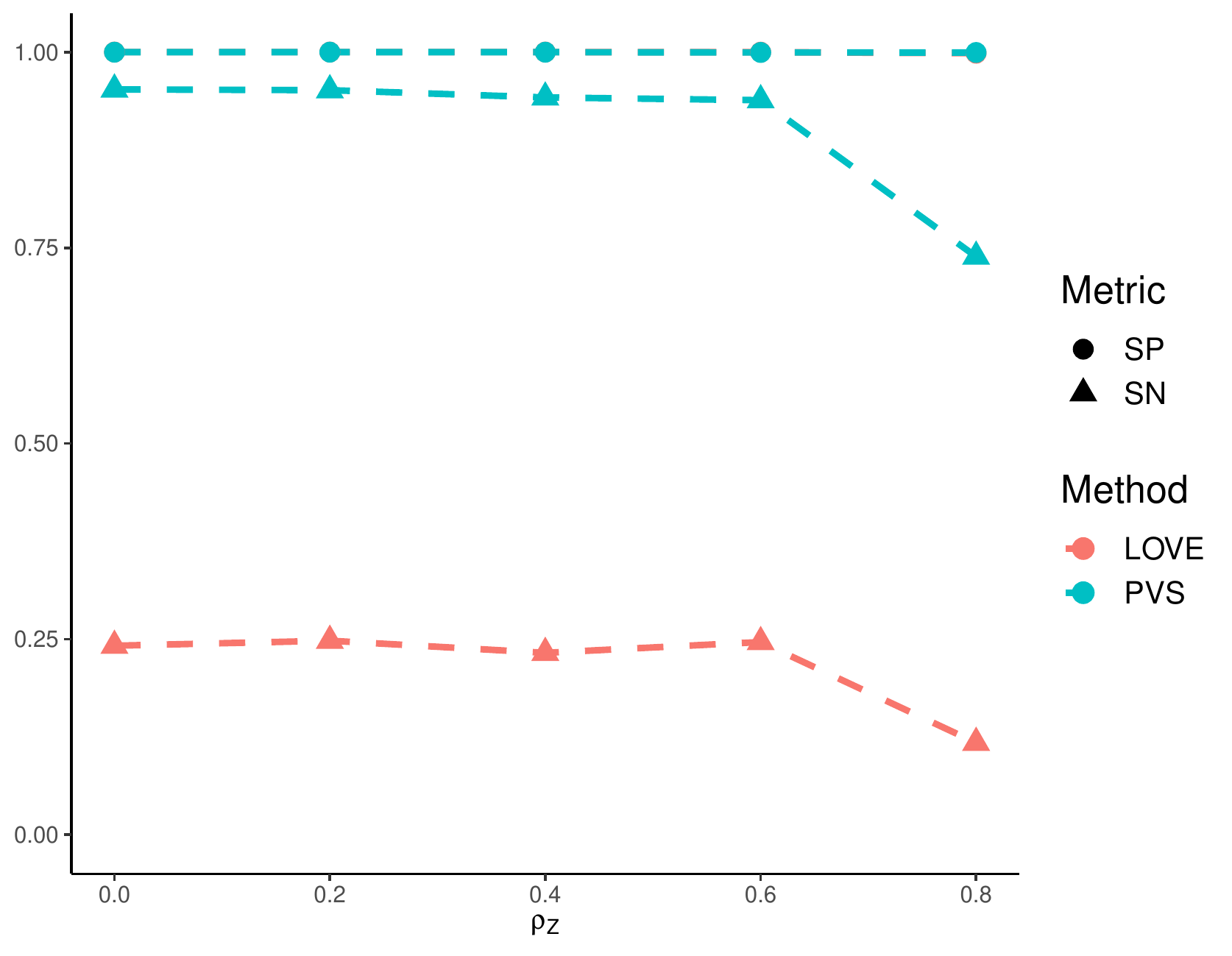}
		\end{tabular}
		\caption{Performance of PVS and LOVE as $\alpha$ or $\rho_Z$ varies.}
		\label{fig_alpha_rhoZ}
	\end{figure}

	\subsubsection*{Varying $\eta$ to change the heterogeneity of $A$}
	
	Finally, we investigate the effect of heterogeneity of $A$ on the performance of PVS and LOVE. For this end, we vary $\eta \in \{0, 0.5, 1, 1.5, 2\}$ to change the heterogeneity of $A$ (recall that small $\eta$ means less heterogeneity). Figure \ref{fig_eta} shows the SP, SN, TPR and TPN of both LOVE and PVS.  
	
	Starting from $\eta = 0$, that is, in the absence of heterogeneity, both LOVE and PVS perfectly recover the pure variables. As the heterogeneity increases, the performance of LOVE drops dramatically whereas PVS is only slightly affected. In the presence of strong heterogeneity, that is, $\eta \ge 1.5$, PVS starts to select less pure variables. The explanation is that some pure variables have too weak a signal to be captured by PVS when $\eta$ is large (recall that the total signal strength over all rows of $A$ is fixed for various of $\eta$).

	\begin{figure}[ht]
		\centering
		\begin{tabular}{cc}
			\includegraphics[width = .4\textwidth]{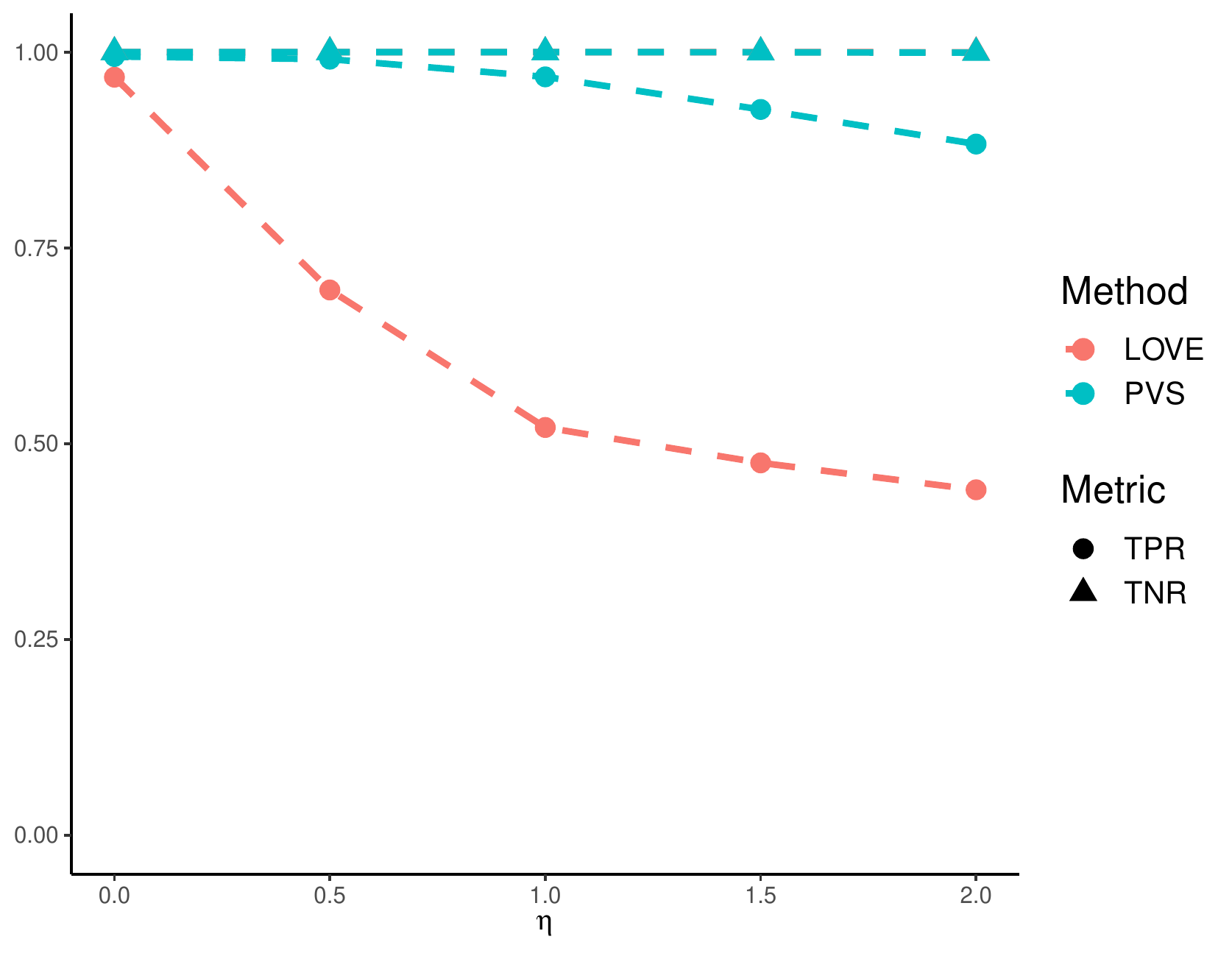} & 
			\includegraphics[width = .4\textwidth]{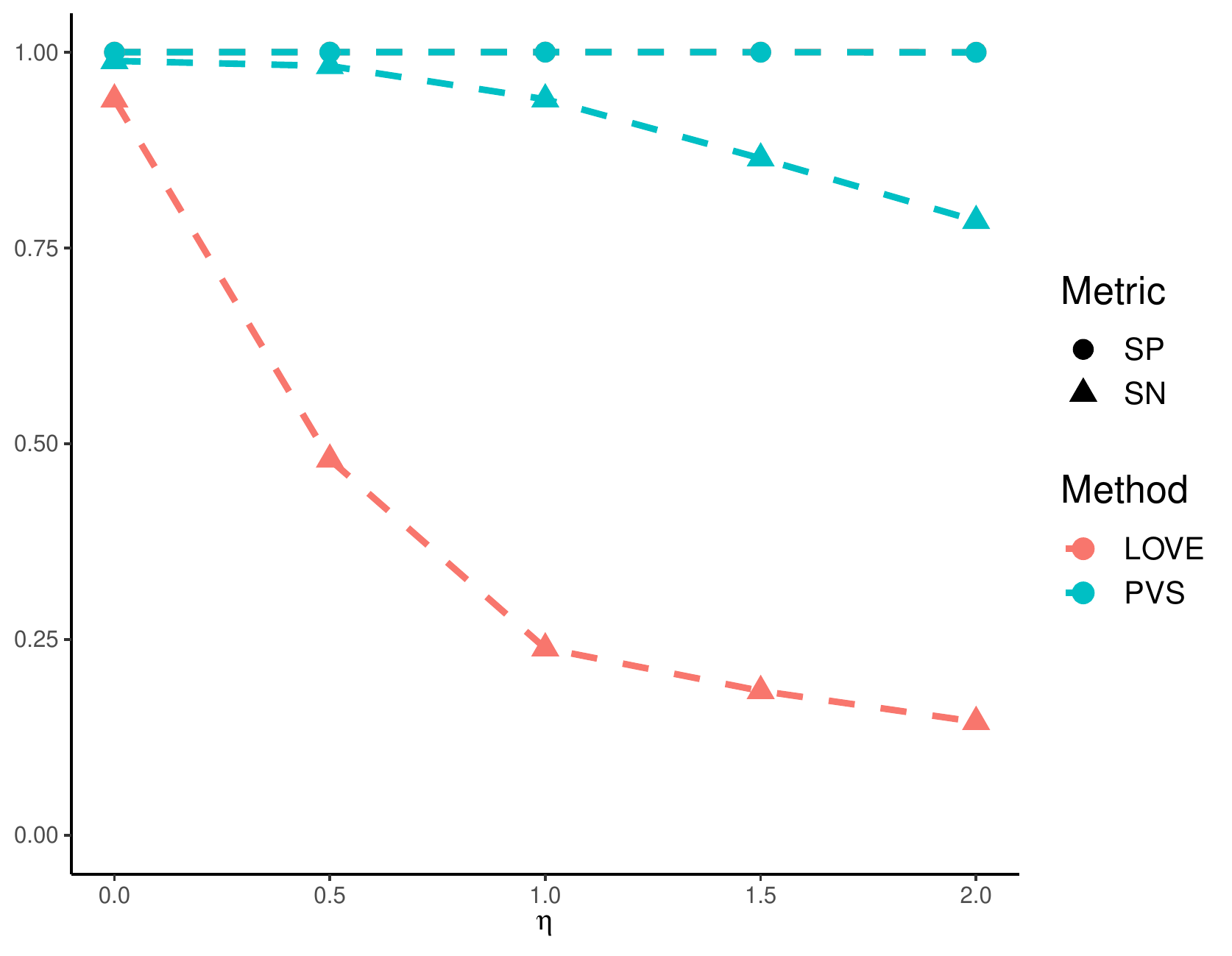}
		\end{tabular}
		\caption{Performance of PVS and LOVE as $\eta$ increases}
		
		\label{fig_eta}
	\end{figure}

	\subsubsection{Errors of estimating $A_{I\cdot}$ and $\C$}
	
	We evaluate PVS in terms of Err$(A)$ and Err$(\C)$ as the sample size $n$ changes. We set $p = 300$, $K=10$, $\rho_Z = 0.3$, $\eta = 1$, $\alpha = 2.5$ and vary $n\in \{100, 300, 500, 700, 900\}$. For each setting, the averaged Err$(A)$ and Err$(\C)$ together with their standard deviations are collected in Table \ref{tab_error}. As expected from Theorem \ref{thm_BI}, the errors for both metric decrease as $n$ increases.

	\begin{table}[ht]
        \centering
        \caption{The averaged errors of estimating $A_{I\cdot}$ and $\C$ as $n$ increases (the numbers in the parenthesis are standard deviations).}
        \label{tab_error}
        {\renewcommand{\arraystretch}{1.3}{
        \begin{tabular}{l|ccccc}
          \hline
         & $n=100$ & $n = 300$ & $n = 500$ & $n = 700$ & $n= 900$ \\ 
          \hline
        {\rm Err}$(A)$  
        & 0.133 (0.048) & 0.046 (0.019) & 0.023 (0.009) & 0.017 (0.008) & 0.014 (0.007)\\ 
       {\rm Err}$(\C)$ & 0.012 (0.006) & 0.005 (0.005) & 0.003 (0.003) & 0.002 (0.003) & 0.004 (0.006)\\
        \hline
        \end{tabular}
        }}
    \end{table}
    
    Finally, since Err$(\C)$ relies on $\wh K = K$, we remark that for almost all of the settings we have considered so far, Algorithm \ref{alg_whole} consistently selects $K$ (we remark that Algorithm \ref{alg_I} consistently selects $K$ as well). To save space, we defer  more simulation results regarding the selection of $K$ to the supplement, including some other settings where Algorithm \ref{alg_whole} consistently selects $K$ but Algorithm \ref{alg_I} does not.

	\subsection{Pre-screening}\label{sec_pre_screen}
	
	We assume $\|A_{j\sbt}\|_2> 0$ for $1\le j\le p$ to establish the identifiability of $H$ in Section \ref{sec_ident}. This can be guaranteed using a  pre-screening method, that is based on the following lemma. 
	\begin{lemma}\label{lem_pure_noise}
	    Under model (\ref{model}) and Assumption \ref{ass_parallel}, we have, for any $q\ge 0$, 
	    \[
	        \|A_{i\sbt}\|_q = 0 \quad \iff \quad \|R_{i, \setminus\{i\}}\|_q = 0.
	    \]
	\end{lemma}
	The proof of Lemma \ref{lem_pure_noise} follows straightforwardly from that of Proposition \ref{prop_I}, and is thus omitted.
	In practice, 
	we recommend to pre-screen the features by removing those in the set 
	\begin{equation}\label{def_N_0}
	   \mathcal{N}_0 := \left\{j\in [p]: \|\wh R_{j,\setminus \{j\}}\|_q \le  (p-1)^{1/q} \ c\sqrt{\log (p\vee n) / n}\right\}
	\end{equation}
	for some constant $c>0$.
	It is straightforward to show that, on the event $\E$, we have 
	\[ \max_{ j\in \mathcal{N}_0 } \|R_{j,\setminus\{j\}}\|_q \lesssim (p-1)^{1/q}\sqrt{\log (p\vee n) / n}.
	\]   We refer to  the features in $\mathcal{N}_0$ as pure-noise variables. 
	We propose the following strategy to select the leading constant $c$ from a specified grid $\mathcal{C} = \{c_1,\ldots, c_N\}$. By splitting the data into two parts with equal size, we compute their sample correlation matrices, denoted as $\wh R^{(1)}$ and $\wh R^{(2)}$. For each $c\in \mathcal{C}$, we find $\mathcal{N}_0$ from (\ref{def_N_0}) and set 
	\[
	     \wt R^{(1)}_{ij} = \left\{ \begin{array}{ll}  \wh R^{(1)}_{ij} & \mbox{for all $i\notin \mathcal{N}_0$, $j\notin \mathcal{N}_0$}\vspace{2mm}\\ 0 &\mbox{otherwise}
\end{array}\right..\]
 We choose the leading constant $c\in\mathcal{C}$ with the smallest loss $\|\wh R^{(2)} - \wt R^{(1)}\|_{\textrm{F-off}}$. The   grid $\mathcal{C}$ can be chosen in a similar way as that of selecting $\delta_n$ in Section \ref{sec_cv}.   Specifically, in our simulations, we choose $q = 2$ and first determine the range of $\mathcal{C}$ from the 0\% and 50\% quantiles of 
 \[
 \left\|\wh R^{(1)}_{j,\setminus\{j\}}\right\|_2\Big/   \sqrt{(p-1)\log (p\vee n) / n}, \qquad j\in [p] \]
 The full grid $\mathcal{C}$ is chosen as $50$ uniformly separated points within this range. The upper level 50\% can be adapted once one has prior information of the maximal percentage of pure-noise variables.\\

	
    We further verify the effectiveness of this pre-screening strategy via simulations. 
    We compare PVS with pre-screening (named as PVS-screen) and PVS without pre-screening (named as PVS-all) in terms of estimating the pure variables and its partition. 
	Under the data generating mechanism described above with  $n = 300$, $p = 100$, $K = 5$, $\eta = 1$, $\alpha = 2.5$ and $\rho_Z = 0.3$,
	we manually add $N_0$ pure-noise variables with $N_0$ varying within $\{0, 5, 10, 15, 20\}$.
  Figure \ref{fig_pure_noise} depicts the performance of PVS-screen and PVS-all. We see that when there exist pure-noise variables ($N_0>0$),   PVS-screen is better than PVS-all in terms of all metrics, while when there is no pure-noise variable $(N_0 = 0)$, PVS-screen has the same performance as PVS-all. 
	
	\begin{figure}[ht]
		\centering
		\begin{tabular}{cc}
			\includegraphics[width = .45\textwidth]{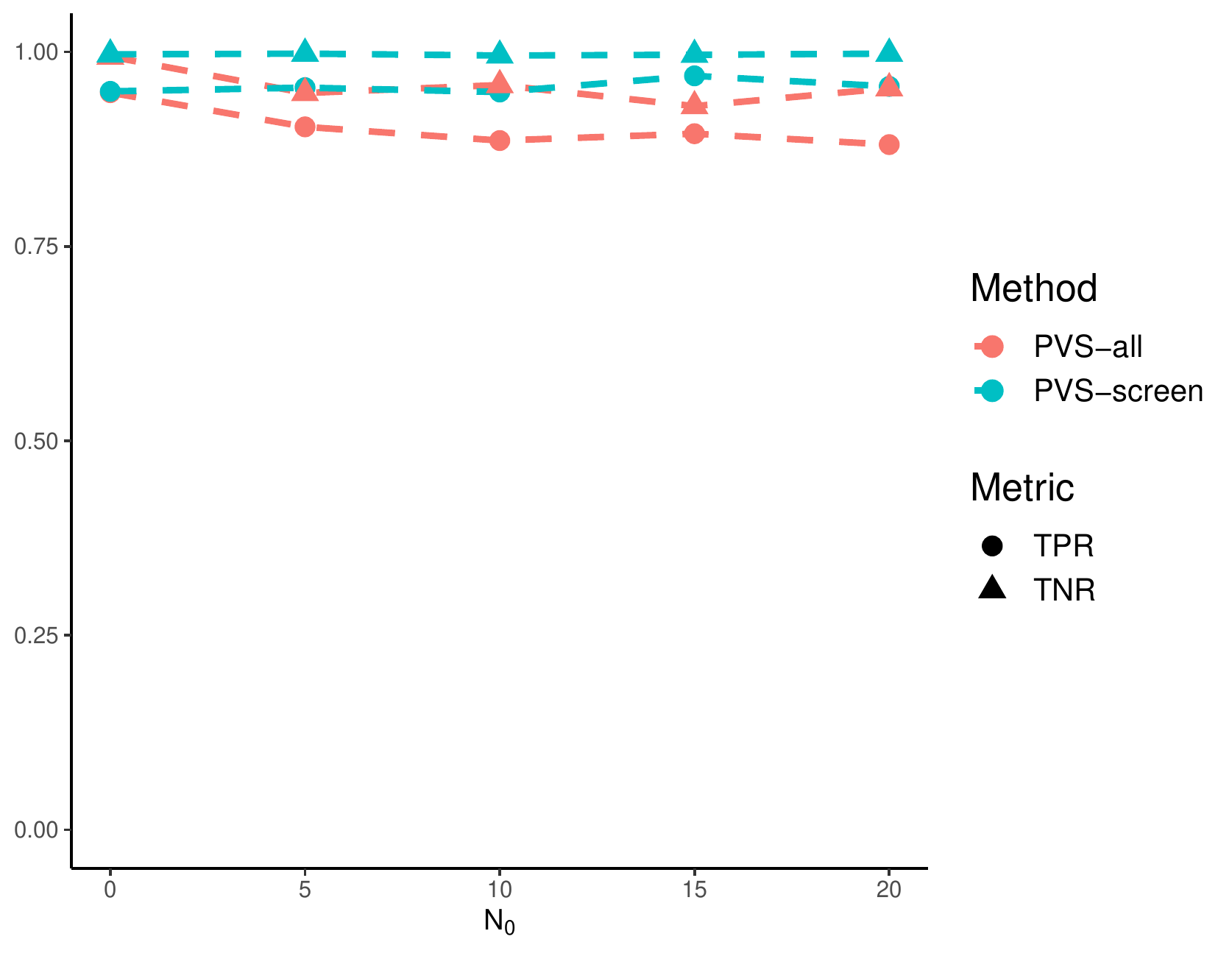} & 
			\includegraphics[width = .45\textwidth]{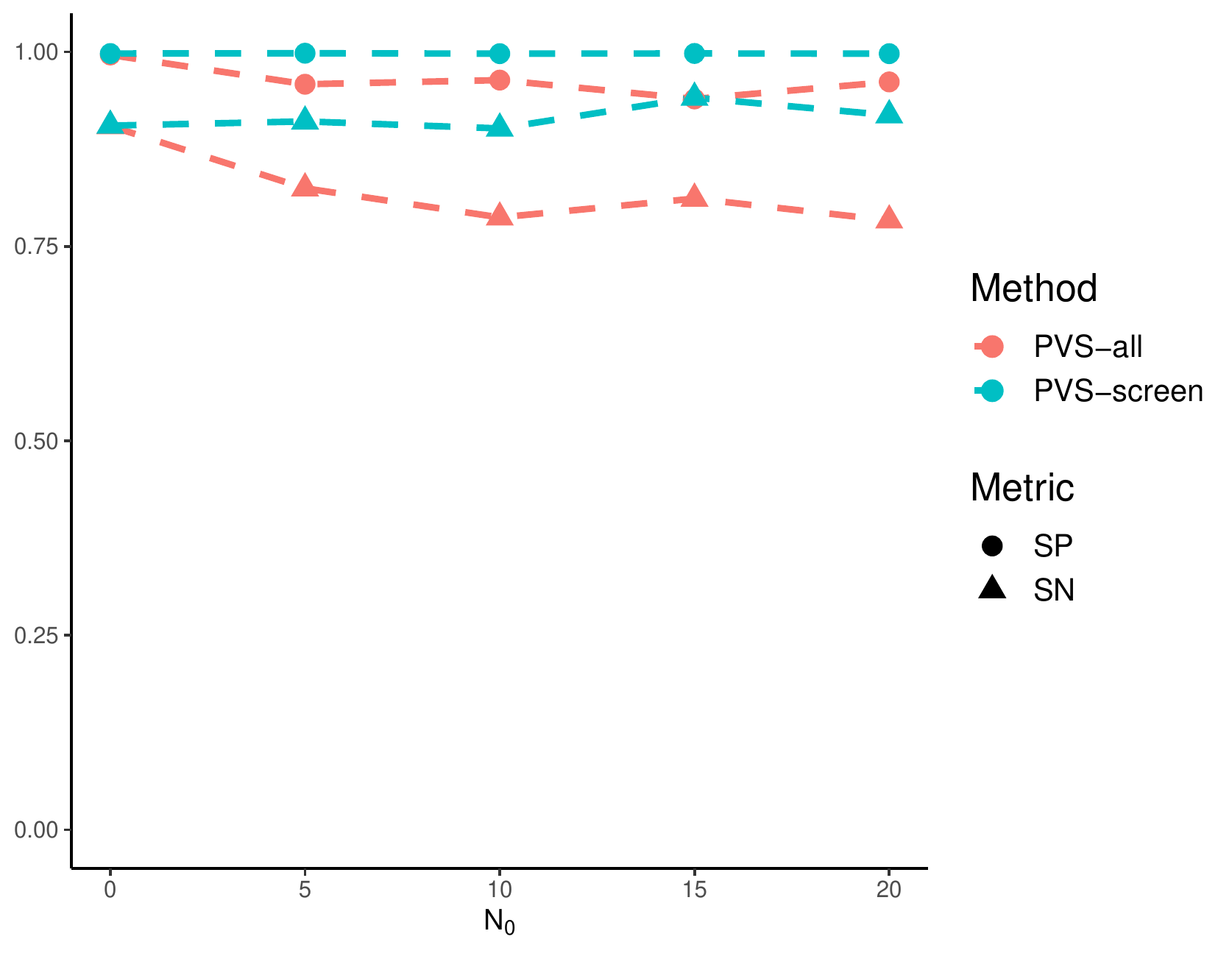}
		\end{tabular}
		\caption{Performance of PVS-all and PVS-screen as $N_0$ increases}
		
		\label{fig_pure_noise}
	\end{figure}

    \section*{Acknowledgements}
	  The authors were  supported in part by NSF grants DMS-1712709 and DMS-2015195.\\
	  We are grateful to Boaz Nadler for stimulating our interest in this problem, and for suggesting a preliminary version of the score function used in this work.

	\newpage

	\appendix 
	
	\section{Proofs}
	
	We provide section-by-section proofs in Appendix \ref{app_sec_ident} -- \ref{app_proof_sec_theory_overlap}. Appendix \ref{app_lemmas_M} contains main technical lemmas and Appendix \ref{app_concentration_R} contains the concentration inequalities of the sample correlation matrix.  Appendix \ref{app_dicuss} provides detailed discussions of condition (\ref{cond_signal_ell_q_relaxed}) and (\ref{cond_signal_ell_q_J}). All auxiliary results / proofs are collected in Appendix \ref{app_auxiliary_proof} while auxiliary lemmas of  concentration inequalities are stated in Appendix \ref{app_aux_lem}.
	
	\subsection{Proofs of Section \ref{sec_parallel}}\label{app_sec_ident}
	
	\subsubsection{Proof of Proposition \ref{prop_I}}\label{app_proof_prop_I}
	    The first statement is obvious by the definition of Assumption \ref{ass_I} and (\ref{def_J1}). We now prove 
	    \[
	        A_{i\sbt} \parallel A_{j\sbt}  \quad  \iff \quad  R_{i,\nij} \parallel R_{j,\nij}.
	    \]
		This holds trivially for any $i = j$. Pick any $i\ne j$. With the convention $0\parallel 0$ in mind, the statement also holds trivially when $K = 1$ and $p = 2$. We thus only consider either $K\ge 1$ or $K = 1$ with $p \ge 3$. 
		
		Recall that $\|A_{j\sbt}\|_2>0$ and $\rank(\C) = K$ imply $\Sigma_{jj}>0$ for all $1\le j\le p$. This implies
			\[
			A_{i\sbt} \parallel A_{j\sbt} \quad \iff \quad  B_{i\sbt}  \parallel B_{j\sbt}
			\]
			with $B = D_{\Sigma}^{-1/2}A$. Use decomposition  $R = B\C B^\T + \Gamma$ in (\ref{eq_R}), the diagonal structure of $\Gamma$ implies 
			\[
			R_{i,\setminus\{i,j\}} = B_{i\sbt }^\T M_{(ij)} ,\qquad  R_{j,\setminus\{i,j\}} = B_{j\sbt}^\T M_{(ij)} .
			\]
			We write for simplicity
			\[M_{(ij)} = \left[\C B^\T\right]_{\sbt \nij} \in \RR^{K\times (p-2)}.\]
			Observe that $K\ge 2$ and Assumption \ref{ass_I} guarantee $p\ge 2K \ge 4$ and $p-2 \ge K$. Further observe $p-2 \ge K$ also holds when $K =1$ and $p \ge 3$. 
		It suffices to show 
			\begin{align*}
			&B_{i\sbt} \parallel B_{j\sbt} \quad \Longrightarrow \quad  B_{i\sbt }^\T M_{(ij)} \parallel  B_{j\sbt }^\T M_{(ij)},\\
			&B_{i\sbt} \notparallel B_{j\sbt} \quad \Longrightarrow  \quad 
			B_{i\sbt }^\T M_{(ij)} \notparallel 	 B_{j\sbt }^\T M_{(ij)}.
			\end{align*}
			The first line follows trivially. To prove the second line, it suffices to show $\rank(M_{(ij)}) = K$ whenever $B_{i\sbt} \notparallel B_{j\sbt}$. To prove this, for any $i \ne j$ satisfying $B_{i\sbt} \notparallel B_{j\sbt}$, we have three possible cases: (1) $i\in I$, $j\in I^c$; (2) $i\in I_a$, $j\in I_b$ with $a\ne b$; (3) $i,j \in I^c$ and $i\ne j$. For any of these cases, Assumption \ref{ass_parallel} implies that we can find a $K\times K$ non-singular submatrix of $B_{\nij \sbt} \in \RR^{(p-2) \times K}$, which together with $\rank(\C) = K$ implies $\rank(M_{(ij)}) = K$. 
			This completes the proof. \qed

	\subsubsection{Proof of Proposition \ref{prop_score}}\label{app_proof_prop_score}  
	    For part (1), the result of $q = r = 2$ follows immediately from Proposition \ref{prop_I}. For general $q$ and $r$, since 
			\[
			S_{q,r}(i,j) = 
			(p-2)^{-1/q} 
			\cdot \min_{v\ne 0} {\left\|v_1 R_{i,\setminus\{i,j\}} + v_2 R_{j,\setminus\{i,j\}}\right\|_q \over \|v\|_r},
			\]
			the result follows by the equivalence of $\ell_q$--norms at the origin. 
			
			For part (2), we first show $(i)$. Pick any $0\le q\le \i$ and set $1_{p-2}=(1,\ldots,1)\in \RR^{p-2}$ so that $\| 1_{p-2}\|_q= (p-2)^{1/q}$. For any $r>0$, we have 
			\[
			S_{q,r}(i,j) = {1\over \|1_{p-2}\|_{q}}\cdot \min_{v\ne 0} {\left\|v_1 R_{i,\setminus\{i,j\}} + v_2 R_{j,\setminus\{i,j\}}\right\|_q \over \|v\|_r}.
			\]
			The result then follows from $\|v\|_r \ge \|v\|_t$ for any $0<r\le t$. 
			
			We proceed to show $(ii)$ of part (2). Pick any $0<r\le \i$ and $\|v\|_r = 1$. For all $t\ge q$, an application of H\"older's inequality gives 
			\begin{align*}
			{\left\|v_1 R_{i,\setminus\{i,j\}} + v_2 R_{j,\setminus\{i,j\}}\right\|_q
				\over \|1_{p-2}\|_q} & ~ \le  {\left\|v_1 R_{i,\setminus\{i,j\}} + v_2 R_{j,\setminus\{i,j\}}\right\|_t\over (p-2)^{1/q}} (p-2)^{{t-q \over tq}}\\
			& ~ =  {\left\|v_1 R_{i,\setminus\{i,j\}} + v_2 R_{j,\setminus\{i,j\}}\right\|_t
				\over \|1_{p-2}\|_t}.
			\end{align*}
			This completes the proof.\qed 
			
	\subsubsection{Proof of Proposition  \ref{prop_score_2}}\label{app_proof_prop_score_2}
	
			Pick any arbitrary $i,j\in [p]$ with $i\ne j$. We drop the superscripts for simplicity and write $V_{ii}=V_{ii}^{(ij)}$ and $V_{ij}= V_{ij}^{(ij)}$. 
			We consider two situations: (a)  $V_{jj}\ge V_{ii}$ and (b) $V_{ii}\ge V_{jj}$, and give first the proof under (a).

		     Let $v^* = (v^*_1, v^*_2)$ be the minimizer of	$S_2^2(i,j)$ given by  (\ref{score_ell_2}). Therefore,  $\max \{|v^*_1|, |v^*_2|   \} = 1$.  To begin with, assume  that  $|v^*_1| = 1$, and observe then that 
			\[
			v^*_2 = \arg\min_{|v_2|\le 1} \left\|v_1^*R_{i,\nij} + v_2 R_{j,\nij}\right\|_2^2.
			\]
			The unrestricted minimum of the quadratic function
			$
			v^*_2 = -v_1^*{V_{ij} / V_{jj}} 
			$
			 is a feasible solution: by Cauchy-Schwarz, $|V_{ij}|^2 \le V_{ii}V_{jj}$, and since we assumed  that $V_{jj} \ge V_{ii}$,  we have $|v_2^*| \le |V_{ij}|/V_{jj} \le 1$. Thus, the corresponding minimum  value of $(p-2)[S_2(i,j)]^2$ when $|v_1^*| = 1$ is given by:	\begin{equation}\label{eq_score_i}
	\left\|v_1^*R_{i,\nij} + v_2^* R_{j,\nij}\right\|_2^2 = \left\|
			R_{i,\nij} - {V_{ij} \over V_{jj}} R_{j,\nij}\right\|_2^2 = V_{ii}\left(
			1 - {V_{ij}^2 \over V_{ii}V_{jj}}\right).
			\end{equation}
			
			We  argue below that this is indeed the minimum value of the objective function, by studying its value in the complementary case,  $|v^*_2| = 1$.  By analogous arguments to those above, we have  \[
			v^*_1 = \arg\min_{|v_1|\le 1} \left\|v_1R_{i,\nij} + v_2^* R_{j,\nij}\right\|_2^2,
			\]
		and therefore 
			\[
			v_1^* =\left\{
			\begin{array}{ll}
			-v_2^*(V_{ij}/V_{ii}),   & \textrm{if }|V_{ij}| \le V_{ii}; \vspace{2mm}\\
			-\sgn(v_2^*V_{ij}),  &    \textrm{if }|V_{ij}| > V_{ii}.
			\end{array}
			\right.
			\]

			 When $|V_{ij}| > V_{ii}$, then  $|v_1^*| = 1$, a case covered above.  Hence, the corresponding minimum of the objective function  is no smaller than (\ref{eq_score_i}). When $|V_{ij}| \le V_{ii}$, 
			\[
			\left[S_2(i,j)\right]^2 =  {V_{jj}\over p-2}\left(
			1 - {V_{ij}^2 \over V_{ii}V_{jj}}\right),\quad \textrm{ if }|V_{ij}| \le V_{ii};
			\]
			which is again also no smaller than (\ref{eq_score_i}), since we work in the case $V_{jj} \geq V_{ii}$.  This concludes 
			\[
			\left[S_2(i,j)\right]^2 = {V_{ii} \over p-2}\left(
			1 - {V_{ij}^2 \over V_{ii}V_{jj}}\right), \quad \text{when} \quad  V_{ii} \le V_{jj}
			\]
		 The proof is completed by  repeating the same arguments under the assumption (b) $V_{ii}\ge V_{jj}$, and  noting that the score is symmetric in $i$ and $j$. \qed\\

	\subsubsection{Proof of Theorem \ref{thm_partial_ident}}\label{sec_proof_thm_partial_ident}
	We only consider $J_1 \ne \emptyset$ as otherwise the result holds automatically from Proposition \ref{prop_I}.  
    
    First, the proof of Proposition \ref{prop_I} reveals that, under Assumption \ref{ass_I}, 
    \[
        A_{i\sbt} \parallel A_{j\sbt} \quad \iff \quad B_{i\sbt} \parallel B_{j\sbt} \quad \iff \quad R_{i,\nij} \parallel R_{j,\nij}.
    \]
    Together with the proof of Proposition \ref{prop_score}, we further have 
    \[
         A_{i\sbt} \parallel A_{j\sbt} \quad \iff \quad S_q(i,j) = 0
    \]
    for any $0\le q\le \i$. Therefore, by comparing $S_q(i,j)$ to zero for all $i\ne j$, we can identify
    $H$ and its partition $\H$.\footnote{Of course, we can only identify the partition up to a group permutation. For simplicity, we assume the identity permutation.} 
    In the following, we show that $\Theta_{HH} := [A\C A^\T]_{HH}$ is identifiable. The uniqueness of $K$ then follows from the fact that $\rank(\Theta_{HH}) = K$.

    To identify $\Theta_{HH}$, we first identify $M_{HH} := [B\C B^\T]_{HH}$. For
    any $i, j\in I_k$ (or $i,j\in J_1^a$) with $i\ne j$ and $k\in [K]$ (or $a\in [N])$, since $B_{i\sbt} \parallel B_{j\sbt}$, there exists $\rho_{ij}\ne 0$ such that 
    \[
            \rho_{ij}B_{j\sbt} = B_{i\sbt}.
    \]
    Note that 
	\begin{equation}\label{eq_MI_offdiag}
	    R_{ij} = M_{ij} = B_{i\sbt}^\T \C B_{j\sbt} = \rho_{ij}B_{j\sbt}^\T \C B_{j\sbt} = \rho_{ij}M_{jj}
	\end{equation}
	and, by recalling that $R_{i,\nij} = B_{i\sbt}^\T \C B^\T_{\nij\sbt}$ and $R_{i,\nij} = B_{j\sbt}^\T \C B^\T_{\nij\sbt}$,
	\begin{equation*}
	|\rho_{ij}| = {\left\|R_{i,\nij}\right\|_q \over \left\|R_{j,\nij}\right\|_q},\qquad \forall q\ge 1.
	\end{equation*}
	The   two displays above together with the fact that $M_{jj}>0$   yield
	\begin{equation}\label{ident_MI}
	 M_{jj} = |R_{ij}| {\|R_{j,\nij}\|_q \over \|R_{i,\nij}\|_q}. 
	\end{equation}
	This identifies $M_{HH}$, and in turn identifies $\Theta_{HH}$, because 
	$$
	    \Theta = A\C A^\T = D_{\Sigma}^{1/2} B\C B^\T D_{\Sigma}^{1/2} = D_{\Sigma}^{1/2} M D_{\Sigma}^{1/2}.
	$$
	The proof is complete. \qed

	\subsubsection{Proof of Theorem \ref{thm_score_ell_q}}\label{app_proof_thm_score_q}
			Fix any $i\ne j$. 
			Recall from (\ref{score}) and (\ref{est_score_ell_q}) that 
			\begin{align*}
			&(p-2)^{1/q}\cdot S_q(i,j) = \min_{\|v\|_\i=1} \left\|v_1R_{i,\nij} + v_2 R_{j,\nij}\right\|_q := \left\|v^*_1 R_{i,\nij} + v_2^* R_{j,\nij}\right\|_q.
			\end{align*}
			Hence, using $\| v^*\|_\infty=1$, and the triangle inequality, we find
			\begin{align*}
			(p-2)^{1/q}\cdot \wh S_q(i,j) &=~  \min_{\|v\|_\i=1} \left\|v_1\wh R_{i,\nij} + v_2 \wh R_{j,\nij}\right\|_q\\
			&:= ~ \left\|\wh v^*_1 \wh R_{i,\nij} + \wh v_2^* \wh R_{j,\nij}\right\|_q\\
			&~\le~  \left\| v^*_1 \wh R_{i,\nij} +  v_2^* \wh R_{j,\nij}\right\|_q\\
			&~\le~ \left\| v^*_1  R_{i,\nij} +  v_2^*  R_{j,\nij}\right\|_q + \left\| v^*_1 \left(\wh R_{i,\nij} - R_{i,\nij}\right)\right\|_q\\
			&\quad + \left\| v^*_2 \left(\wh R_{j,\nij} - R_{j,\nij}\right)\right\|_q\\
			&~\le~ 	(p-2)^{1/q}\cdot S_q(i,j) + \left\|\wh R_{i,\nij} - R_{i,\nij}\right\|_q+\left\|\wh R_{j,\nij} - R_{j,\nij}\right\|_q.
			\end{align*}
		Analogously, 
			\begin{align*}
			(p-2)^{1/q}\cdot S_q(i,j) &~\le~ \left\| \wh v^*_1 R_{i,\nij} +  \wh v_2^* R_{j,\nij}\right\|_q\\
			&~\le~ 	(p-2)^{1/q}\cdot \wh S_q(i,j) + \left\|\wh R_{i,\nij} - R_{i,\nij}\right\|_q+\left\|\wh R_{j,\nij} - R_{j,\nij}\right\|_q.
			\end{align*}
			Invoking Lemma \ref{lem_dev_R} concludes the proof of the first statement. The rest of the proof follows by noting that $S_q(i,j) = 0$ for all $i\group j$ from Proposition \ref{prop_score}. \qed

      \subsubsection{Proof of Proposition \ref{prop_K_parallel}}\label{app_proof_prop_K_parallel}	
	     We prove the statement on the event $\E \cap \E_H$ where 
	    \[
            \E_H:= \bigcup_{\substack{L = \{\ell_1,\ldots, \ell_G\}\\ \ell_k\in H_k, k \in [G]}}\left\{
			\left\|\wh M_{LL} -  M_{LL}\right\|_{\rm op} \lesssim ~  \ol c(L)(\sqrt{G}\delta_n + G\delta_n^2) \right\}
		\]
		holds with probability $1-c(n\vee p)^{-c'}$ from Lemma \ref{lem_op_M_H}. Since Corollary \ref{cor_H} ensures $\wh H = H$ and $\wh H_k = H_k$ (without loss of generality, we assume the identity permutation $\pi$), we further have 
		\[
		    \left\|\wh M_{\wh L\wh L} -  M_{\wh L\wh L}\right\|_{\rm op} \lesssim   ~ \ol c(L)(\sqrt{G}\delta_n + G\delta_n^2).
		\]
		 Weyl's inequality and the fact that $M_{\wh L\wh L}$ has rank $K$ yield 
	    \[
	        \lambda_{K+1}(\wh M_{\wh L\wh L})\le \lambda_{K+1}(M_{\wh L\wh L}) + \|\wh M_{\wh L \wh L} - M_{\wh L\wh L}\|_{\rm op} \lesssim  \ol c(L) (\sqrt{G\delta_n^2} + G\delta_n^2)< \mu.
	    \] 
	    This implies $\wh K \le K$. On the other hand, by Weyl's inequality again,
	    \[
	        \lambda_K(\wh M_{\wh L\wh L}) \ge   \lambda_K(M_{\wh L\wh L}) - \|\wh M_{\wh L \wh L} - M_{\wh L\wh L}\|_{\rm op}.
	    \]
	    Since the definition of $\ul c(L)$ in (\ref{def_cbI}) ensures that
	    \[
	        \lambda_K(M_{\wh L\wh L}) \ge  \ul c(L),
	    \]
	    by using $\max_L[\ol c(L) / \ul c(L)]G\delta_n^2\le c'$ for sufficiently small $c'>0$ and (\ref{cond_regularity_parallel}), 
	    we conclude 
	    $
	        \lambda_K(\wh M_{\wh L\wh L}) > \mu
	    $
	    for all $L$, 
	    which implies $\wh K \ge K$ and completes the proof. \qed

    \subsection{Proofs of Section \ref{sec_overlap}}\label{app_proof_sec_theory_overlap}

    \subsubsection{Proof of Theorem \ref{thm_ident_I}}\label{app_proof_thm_ident_A}
	 The identifiability of the matrix $A$ follows from that of its scaled version $B=D_{\Sigma}^{-1/2}A$. To derive this property for $B$, we discuss the identifiability of $B_{I\sbt}$, $\C$ and $B_{J\sbt}$.\\ 
	
	Suppose $\I = \{I_1, \ldots, I_K\}$ is 
	recovered (for instance, as a consequence from Assumptions \ref{alg_I} \& \ref{ass_J_prime}). Using the fact that, w.l.o.g, 
$\diag(\C) = 1$, we have, for 
any $i, j\in I_k$ with $i\ne j$ and $k\in [K]$, 
	\[
	R_{ij} = B_{ik}[\C]_{kk}B_{jk} = B_{ik} B_{jk},
	\]
	and, by recalling that $R_{i,\nij} = B_{ik}[\C B^\T]_{k,\nij}$ and $R_{j,\nij} = B_{jk}[\C B^\T]_{k,\nij}$,
	\begin{equation}\label{eq_ratio}
	{|B_{ik}| \over |B_{jk}|} = {\left\|R_{i,\nij}\right\|_q \over \left\|R_{j,\nij}\right\|_q},\qquad \forall q\ge 1.
	\end{equation}
	The   two displays above yield
	\begin{equation}\label{ident_BI}
	B_{ik}^2 = |R_{ij}| {\|R_{i,\nij}\|_q \over \|R_{j,\nij}\|_q}, \quad \forall i,j\in I_k,\ i\ne j. 
	\end{equation}
	Since we can only identify $B$ up to a $K\times K$ signed permutation matrix, we choose 
	\begin{equation}\label{ident_BI_sign}
	\sgn(B_{ik}) = 1,\qquad \sgn(B_{jk}) = \sgn(R_{ij}),\qquad \forall j\in I_k \setminus \{i\}. 
	\end{equation}
	To identify $\C$, we first 
	identify $\Gamma_{II}$ as 
	\begin{equation}\label{ident_Gamma_II}
	\Gamma_{ii} = R_{ii} - B_{ik}^2 = 1-B_{ik}^2,\qquad \forall i\in I_k, k\in [K].
	\end{equation}
	From $R_{II} = B_{I\sbt} \C B_{I\sbt}^\T + \Gamma_{II}$, we then identify $\C$ by 
	\begin{equation}\label{ident_C}
	\C = B_{I\sbt}^+\left(R_{II} - \Gamma_{II}\right)B_{I\sbt}^{+\T},
	\end{equation}
	where $B_{I\sbt}^+ = (B_{I\sbt}^\T B_{I\sbt})^{-1}B_{I\sbt}^\T$ is the left inverse of $B_{I\sbt}$. Finally, to identify $B_{J\sbt}$, by noting that $R_{IJ} = B_{I\sbt}\C B_{J\sbt}^\T$, one has 
	\[
	B_{j\sbt} = \C^{-1} B_{I\sbt}^+ R_{Ij},\qquad \forall  j\in J.
	\]
	This completes our proof.
	\qed\\

	\subsubsection{Proof of Lemma \ref{lem_find_L}}\label{app_proof_lem_find_L}	 
	
	For convenience, we make the following definition.  
	\begin{definition}\label{def_Sk}
		Define $S_1 = \emptyset$ and, for $2\le k\le K$,
		$S_k = \{i_1, \ldots, i_{k-1}\}$  satisfying $i_a \in I_{\pi(a)}$  for all $1\le a\le k-1$. Further define the set of all possible $S_k$ for any fixed $k$ as 
		\[
		    \S_k := \left\{
			    S_k =\{i_1, \ldots, i_{k-1}\}: \forall\ i_a \in I_{\pi(a)}, \forall\ 1\le a\le k-1
			\right\}.
		\]
	\end{definition}
	We remark $|\S_k| = \prod_{a=1}^{k-1}|I_{\pi(a)}| \le p^{k-1}$.

		\begin{proof}[Proof of Lemma \ref{lem_find_L}]
		It suffices to prove the result for any given $1\le k\le K$. Without loss of generality,  assume $S_k=\{i_1,\ldots, i_{k-1}\}$ satisfying $i_a \in I_a$ for $1\le a\le k-1$. Write $T_k = \{1,\ldots, k-1\}$ (with $T_1 = \emptyset$) as the set of groups corresponding to $S_k$ and write the rest of groups as $T_k^c = [K]\setminus T_k$. 
		
		We aim to show 
		\begin{equation}\label{eq_target}
		    \arg\max_{j \in H} \Theta_{jj| S_k} \in I_{k^*}, \quad \textrm{for some }k^*\in T^c.
		\end{equation}
		Write
		\begin{equation}\label{def_tildes}
		    \wt A_{j\sbt} = D_\xi^{-1} A_{j\sbt},\qquad V = D_\xi\C D_\xi
		\end{equation} 
		such that $\Theta = A\C A^\T = \wt A V \wt A^\T$, where
		\begin{equation}\label{def_D_A}
		    D_\xi = \diag\left( \xi_1, \ldots, \xi_K\right),\quad \xi_k = \max_{i\in I_k} |A_{ik}|,\quad \forall 1\le k\le K.
		\end{equation}
		Note that Assumption \ref{ass_J_prime} implies 
		\begin{equation}\label{cond_scale}
		    \max_{j\in J_1} \| \wt A_{j\sbt}\|_1  = \max_{j\in J_1} \sum_{k=1}^K |A_{jk}|/ \xi_k \le 1.
		\end{equation}
		Further note that $\max_{i\in I_k}|\wt A_{ik}| = 1$ for all $1\le k\le K$. For any $j\in [p]$,
		\begin{eqnarray}\label{ident_Theta_jjsk}\nonumber
		\Theta_{jj|S_k} &=& \Theta_{jj} - \Theta_{jS_k}^\T \Theta_{S_kS_k}^{-1}\Theta_{S_k j}\\\nonumber
		&=& \wt A_{j\sbt}^\T \left(V - V_{\sbt T_k}V_{T_kT_k}^{-1}V_{T_k\sbt}\right) \wt A_{j\sbt} \\
		&=& \wt A_{jT_k^c}^\T V_{T_k^cT_k^c| T_k}\wt A_{jT_k^c}
		\end{eqnarray}
		with 
		$$
		V_{T_k^cT_k^c| T_k} = V_{T_k^cT_k^c} - V_{T_k^c T_k}V_{T_kT_k}^{-1}V_{T_kT_k^c}.
		$$
	    We partition the set $H=I\cup J_1$ as 
		$$
		    H = \left(\bigcup_{b\in T_k^c} I_b\right) \cup   \left(\bigcup_{a\in T_k}I_a\right) \cup J_1.
		$$
		Notice that $\wt A_{j\sbt}$ has the same support as $A_{j\sbt}$. From (\ref{ident_Theta_jjsk}), observe that
		\begin{alignat}{2}\label{eq_Theta_schur_I}
		    \Theta_{jj| S_k} & = 0 \qquad &&\text{ for all } j\in I_a, \ a\in T_k;\\\nonumber
		    \Theta_{ii| S_k} & = \wt A_{ib}^2\left[V_{T_k^cT_k^c|T_k}\right]_{bb} = \wt A_{ib}^2V_{bb| T_k}\qquad && \text{ for all } i\in I_b, \ b\in T_k^c.
		\end{alignat}
	The first claim follows immediately as $\wt A_{jT_k^c}=\0$ for $j\in I_a$ and $a\in T_k$.
		Thus, to prove (\ref{eq_target}), it suffices to show 
		\[
			\max_{j\in J_1} \Theta_{jj| S_k} < \max_{b\in T_k^c}\max_{i\in I_b}\Theta_{ii| S_k} =  \max_{b\in T_k^c}\max_{i\in I_b} \wt A_{ib}^2V_{bb| T_k} = \max_{b\in T_k^c}V_{bb| T_k}
		\] 
		by using $\max_{i\in I_k}\wt A_{ik}^2 = 1$ in the last step. 
	 	To show this, for any $j \in J_1$, we consider two cases:\\
	 	{\bf Case 1:  $\|\wt A_{jT_k}\|_0 \ge 1$.} We have
	 		\[
	 		\Theta_{jj| S_k} \le \|\wt A_{jT_k^c}\|_1^2 \left\| V_{T_k^cT_k^c| T_k}\right\|_\i  < \|\wt A_{j\sbt}\|_1^2 \max_{b\in T_k^c}V_{bb| T_k} \overset{(\ref{cond_scale})}{\le} \max_{b\in T_k^c} V_{bb| T_k}.
	 	\]
	 	\noindent
	 	{\bf Case 2: $\|\wt A_{jT_k}\|_0 = 0$.} By the definition of $J$,  $\|\wt A_{jT_k^c}\|_0 \ge 2$ and $
	 	\|\wt A_{jT_k^c}\|_1^2 > \|\wt A_{jT_k^c}\|_2^2
	 	$. 
	 		Applying Lemma \ref{lem_quad} to $\Theta_{jj| S_k}$ with $\alpha = \wt A_{jT_k^c}$ and $M = V_{T_k^cT_k^c| T_k}$ yields
		\begin{align}\label{bd_Theta_jjsk}\nonumber
			\Theta_{jj| S_k}  &\le \|\wt A_{jT_k^c}\|_1^2 \max_{b\in T_k^c}V_{bb| T_k} - \left(
			\|\wt A_{jT_k^c}\|_1^2 - \|\wt A_{jT_k^c}\|_2^2\right)\Delta(V_{T_k^cT_k^c| T_k}) \\
			&< \|\wt A_{jT_k^c}\|_1^2 \max_{b\in T_k^c}V_{bb| T_k}.
		\end{align}
	 	Since 
	 	\begin{align*}
	 	    \Delta(V_{T_k^cT_k^c| T_k}) &\ge \lambda_{\min}(V_{T_k^cT_k^c| T_k}) & (\text{ by Lemma \ref{lem_Delta}})\\
	 	    & \ge \lambda_{\min}(V) \\
	 	    &> 0 & (\text{using $\lambda_K(\C)>0$)}.
	 	\end{align*}
	 	%
	 	Hence, (\ref{bd_Theta_jjsk}) implies
	 	\[
	 		 \Theta_{jj| S_k} <\|\wt A_{jT_k^c}\|_1^2 \max_{b\in T_k^c}V_{bb| T_k} \le \max_{b\in T_k^c}V_{bb| T_k} .
	 	\]
	 	This completes the proof. 
	\end{proof}
	
	\bigskip

	The following two lemmas are used in the proof of Lemma \ref{lem_find_L}.
	
	\begin{lemma}\label{lem_quad}
		Let $M\in \RR^{d\times d}$ be any symmetric matrix. Define 
		\begin{equation}\label{def_delta_M}
			\Delta(M) := \min_{i\ne j} {1\over 2}(M_{ii} + M_{jj}) - |M_{ij}|.
		\end{equation}
		Then for any $\alpha \in \RR^d$, we have 
		\[
			\alpha^\T M \alpha \le \|\alpha\|_1^2 \max_i M_{ii} -  \left(
			\|\alpha\|_1^2 - \|\alpha\|_2^2
			\right)\Delta(M) .
		\]
	\end{lemma}
	\begin{proof}
	    Write $M_\i := \max_i |M_{ii}|$. Expanding the quadratic form gives
	    \begin{align*}
	       \alpha^\T M \alpha  &= \sum_{i=1}^d M_{ii}\alpha_i^2 + \sum_{i\ne j}\alpha_i\alpha_j M_{ij}\\
	       &\le M_\i \sum_{i=1}^d \alpha_i^2 + \sum_{i\ne j}|\alpha_i\alpha_j||M_{ij}|\\
	       &\le M_\i\sum_{i=1}^d \alpha_i^2 + \sum_{i\ne j}|\alpha_i\alpha_j|\left[{1\over 2}(M_{ii}+M_{jj}) - \Delta(M)\right]\\
	       &\le M_\i\left(\sum_{i=1}^d \alpha_i^2 +\sum_{i\ne j}|\alpha_i\alpha_j|\right) - \sum_{i\ne j}|\alpha_i\alpha_j|\Delta(M).
	    \end{align*}
	    The result follows by noting that $\|\alpha\|_1^2 - \|\alpha\|_2^2 = \sum_{i\ne j}|\alpha_i\alpha_j|$.
	\end{proof}
	
	\begin{lemma}\label{lem_Delta}
		Let $M\in \RR^{d\times d}$ be symmetric matrix with smallest eigenvalue equal to $\lambda_{\min}(M)$. With $\Delta(M)$ defined in (\ref{def_delta_M}), one has
		\[
		  \Delta(M) \ge \lambda_{\min}(M)
		\]
	\end{lemma}
	\begin{proof}
	    By definition, for any $i\ne j$,
	    \[
	        \lambda_{\min}(M) = \inf_{\|u\|_2=1}u^\T M u \le {1\over 2}(M_{ii} + M_{jj}) - |M_{ij}|
	    \]
	    where we take $u = \e_i/\sqrt{2} + \sgn(M_{ij})\e_j/\sqrt{2}$. This finishes the proof. 
	\end{proof}
    	 
    \bigskip

	\subsubsection{Proof of Theorem \ref{thm_ident_I_post}}\label{app_proof_thm_ident_I_post}
    
    Recall that $\Theta_{HH}$ is identified from the proof of Theorem \ref{thm_partial_ident}. Given $\H$, to identify $I$ and $\I$, it suffices to show that we can find at least one variable in $I_k$ for each $k\in [K]$. This follows immediately from Lemma \ref{lem_find_L}. The identifiability of $A$ follows from the same arguments in the proof of Theorem \ref{thm_ident_I}.
	\qed\\

      \subsubsection{Proof of Proposition \ref{prop_K}}\label{app_proof_prop_K}
    
        We work on the event $\E \cap \E_M'$ where 
        \[
            \E_M':= \bigcup_{S_K\in \S_K}\left\{
			\left\|\wh M_{S_KS_K} -  M_{S_KS_K}\right\|_{\rm op} \le  c_1(\sqrt{K}\delta_n + K\delta_n^2) \right\}
		\]
		holds with probability $1-c(n\vee p)^{-c'}$ from Lemma \ref{lem_op_M}. 
		
		Corollary \ref{cor_I_prime} ensures that $I_k = \wh H_{\pi(k)}$ for all $1\le k\le K$ with some permutation $\pi$. This implies that $\wh L = S_{K+1} \cup S_{K+1}^c$ where $S_{K+1}$ satisfies Definition \ref{def_Sk}, that is, 
	    \[
	        S_{K+1} = \{i_1,\ldots, i_K\},\quad i_k \in I_{\pi(k)}, \quad \forall k\in [K]
	    \]
	    and $S_{K+1}^c = \wh L \setminus S_{K+1} \subset \bar J_1$. For notational simplicity, we drop the subscripts $K+1$. We have 
	    \begin{align*}
	        \|\wh M_{\wh L \wh L} - M_{\wh L\wh L}\|_{\rm op} &\le 
	        \|\wh M_{SS} - M_{SS}\|_{\rm op} + \|\wh M_{S^cS^c} - M_{S^c S^c}\|_{\rm op}.
	    \end{align*}
	    By the string of inequalities $\|Q\|_{\rm op} \le \|Q\|_F \le d\|Q\|_\i$ for any matrix $Q\in \RR^{d\times d}$, we have
	    \begin{align*}
	         \|\wh M_{S^cS^c} - M_{S^cS^c}\|_{\rm op} 
	        &\le  |S^c| \max_{i,j\in \bar J_1}|\wh M_{ij} - M_{ij}| \lesssim (\wh G - K)\delta_n
	    \end{align*}
	    by invoking Lemma \ref{lem_Gamma_M_Ihat} in the last step. Invoking $\E_M'$ to bound $\|\wh M_{SS} - M_{SS}\|_{\rm op}$ together with the fact that $\wh G - K \le |\bar J_1| / 2$ yields 
	    \[
	        \|\wh M_{\wh L \wh L} - M_{\wh L\wh L}\|_{\rm op} \lesssim \sqrt{K\delta_n^2} + K\delta_n^2 + |\bar J_1| \delta_n.
	    \]
	    Weyl's inequality and the fact that $M_{\wh L\wh L}$ has rank $K$ yield 
	    \[
	        \lambda_{K+1}(\wh M_{\wh L\wh L})\le \lambda_{K+1}(M_{\wh L\wh L}) + \|\wh M_{\wh L \wh L} - M_{\wh L\wh L}\|_{\rm op} \lesssim \sqrt{K\delta_n^2} + K\delta_n^2 + |\bar J_1| \delta_n< \mu,
	    \]
	    by choosing $\mu = C(\sqrt{K\delta_n^2} + K\delta_n^2 + |\bar J_1| \delta_n)$ for some constant $C>0$. This implies $\wh K \le K$. 
	    
	    On the other hand, by Weyl's inequality again,
	    \[
	        \lambda_K(\wh M_{\wh L\wh L}) \ge   \lambda_K(M_{\wh L\wh L}) - \|\wh M_{\wh L \wh L} - M_{\wh L\wh L}\|_{\rm op}.
	    \]
	    Since 
	    \[
	        \lambda_K(M_{\wh L\wh L}) \ge \lambda_K(M_{SS}) \overset{(\ref{lb_eigen_Msksk})}{\ge} c_zc_{b,I},
	    \]
	    by using $\max\{K\delta_n^2, |\bar J_1|\delta_n\}\le c'$ for sufficiently small $c'>0$, 
	    we conclude 
	    \[
	        \lambda_K(\wh M_{\wh L\wh L}) \gtrsim c_zc_{b,I} > \mu,
	    \]
	    which implies $\wh K \ge K$ and completes the proof. \qed 
	    
	 \bigskip

    \subsubsection{Proof of Theorem \ref{thm_I_post}}\label{app_proof_thm_I_post}
	

	We work on the event $\E$ which ensures that the results of Corollary \ref{cor_I_prime}  hold. For any given $1\le r \le K$, it suffices to show that Lemma \ref{lem_find_L} holds with $\Theta, H$  and $K$ replaced by $\wh \Theta$, $\wh H$ and $r$,  respectively, that is,  to show
	\[
		i_k :=\arg\max_{j \in \wh H} \wh\Theta_{jj| S_k} \in I_{\pi(k)},\qquad \text{for all }1\le k\le r,
	\]
	with $S_k$ satisfying Definition \ref{def_Sk}.
	This is guaranteed by invoking Lemmas \ref{lem_noise} and \ref{lem_signal}. Indeed, for any $1\le k\le r$, recall that $S_k = \{i_1, \ldots, i_{k-1}\}$ is defined in Definition \ref{def_Sk} meanwhile  $T_k = \{1,\ldots, k-1\}$\footnote{This is assumed without loss of generality.} is the set of groups corresponding to elements of $S_k$ and $T^c_k = [K]\setminus T_k$. With probability $1-c(n\vee p)^{-c'}$ for some constant $c,c'>0$, for all $1\le k\le r$ and $S_k\in \S_k$, there exists $i_k \in I_a$ for some $a\in T_k^c$ such that
	\begin{align*}
			&\wh\Theta_{i_ki_k | S_k} - \max_{j \in  \cup_{b\in T_k}I_b \cup \bar J_1 }\wh\Theta_{jj | S_k}\\
			&\ge \left(\Theta_{i_ki_k | S_k} - \Sigma_{i_ki_k}\cdot \varepsilon_n\right) - \max_{j \in  \cup_{b\in T_k}I_b \cup \bar J_1 }\left(\Theta_{jj | S_k} + \Sigma_{jj} \cdot \varepsilon_n\right) &(\text{by Lemma \ref{lem_noise}})\\
			&=
			\min_{j \in  \cup_{b\in T_k}I_b \cup \bar J_1 }
			\left(\Theta_{i_ki_k | S_k} -\Theta_{jj | S_k} - (\Sigma_{i_ki_k} + \Sigma_{jj}) \cdot \varepsilon_n\right)\\
			& > 0 &(\text{by Lemma \ref{lem_signal}})
	\end{align*}
	where $\varepsilon_n = c_*\sqrt{K}\delta_n$.
	This implies that there exists some permutation $\pi:[K]\to [K]$ such that the output $\{i_1, \ldots, i_r\}$  satisfies 
	$
		i_a \in I_{\pi(a)}
	$
	for $1\le a\le r$, as desired.\qed \\

    \subsubsection*{Lemmas used in the proof of Theorem \ref{thm_I_post}}
    
    We now state and prove Lemmas \ref{lem_noise} and \ref{lem_signal}. 	The following lemma provides the upper bounds of $|\wh \Theta_{ii|S_k} - \Theta_{ii|S_k}|$ over all $i\in \wh H$, $S_k\in \S_k$ and $1\le k\le K$.

	\begin{lemma}\label{lem_noise}
		Under conditions of Theorem \ref{thm_I_post}, with probability $1-c(p\vee n)^{-c'}$ for some constant $c,c'>0$,
		\[
			\max_{1\le k\le K} \max_{S_k \in \S_k}\max_{i\in \wh H}  \Sigma_{ii}^{-1}\left|
			\wh\Theta_{ii | S_k} - \Theta_{ii | S_k}
			\right|  \le c_*\sqrt{K}\delta_n
		\]
		where $c_* = c_*(c_z,c_{b,I},\bar c_z)>0$ is some constant. 
	\end{lemma}
	\begin{proof}
		We work on the event $\E$ intersected with 
		\begin{equation}\label{def_event_E_M}
		\E_M:= \bigcup_{k=1}^K \bigcup_{S_k\in \S_k}\left\{
			\left\|\wh M_{S_kS_k} -  M_{S_kS_k}\right\|_{\rm op} \le  c_1\sqrt{K}\delta_n,~ \lambda_{k-1}\left(\wh M_{S_kS_k}\right) \gtrsim c_2
		\right\}
       \end{equation}
		Here we use the convention that $\lambda_0(M) = \i$ and  $c_1$ and $c_2$ are positive constants depending on $c_z,c_{b,I}$ and $\bar c_z$. Lemma \ref{lem_op_M} ensures $\PP(\E_M \cap \E) \ge 1-c(n\vee p)^{-c'}$.\\
		
		Pick any $1\le k\le K$, $S_k \in \S_k$ and $i\in \wh H$. Recall that 
		$$
		\Theta = A\C A^\T = D_{\Sigma}^{1/2}M D_{\Sigma}^{1/2},\qquad M = B\C B^\T
		$$ 
		and
		\[
			\wh\Theta = \wh\Sigma - D_{\wh \Sigma}^{1/2}\wh\Gamma D_{\wh \Sigma}^{1/2}=  D_{\wh \Sigma}^{1/2}\left(\wh R - \wh\Gamma \right)D_{\wh \Sigma}^{1/2} = D_{\wh \Sigma}^{1/2}\wh M D_{\wh \Sigma}^{1/2}.
		\]
		Observe
		\[
		    \Theta_{ii | S_k} = \Theta_{ii} - \Theta_{iS_k}^\T\Theta_{S_kS_k}\Theta_{S_ki} = \Sigma_{ii}\left(
		    M_{ii} - M_{iS_k}^\T M_{S_kS_k}^{-1}M_{S_ki}\right) = \Sigma_{ii} M_{ii|S_k}
		\]
		and similarly
		\[
		    \wh\Theta_{ii | S_k} = \wh\Theta_{ii} - \wh\Theta_{iS_k}^\T\wh\Theta_{S_kS_k}\wh\Theta_{S_ki} = \wh \Sigma_{ii} \wh M_{ii|S_k}.
		\]
		We find
		\begin{align*}
	    	\wh\Theta_{ii | S_k} - \Theta_{ii | S_k} 
	    	&= \wh \Sigma_{ii}\wh M_{ii|S_k} - \Sigma_{ii} M_{ii|S_k}.
		\end{align*}
		We first bound from above $|\wh M_{ii|S_k} - M_{ii|S_k} |$. 
		By definition and adding and subtracting terms, 
		\begin{align*}
			\left|\wh M_{ii | S_k} - M_{ii | S_k}\right| & = \left|\wh M_{ii} - \wh M_{iS_k}^\T \wh M_{S_kS_k}^{-1}\wh M_{S_ki} - M_{ii} + M_{iS_k}^\T M_{S_kS_k}^{-1}M_{S_ki}\right|\\
			&\le  | \wh M_{ii} - M_{ii}| + \left| \wh M_{iS_k}^\T \wh M_{S_kS_k}^{-1}(\wh M_{S_ki} - M_{S_ki})\right|\\
			&\quad   +\left| \wh M_{iS_k}^\T \left(\wh M_{S_kS_k}^{-1} - M_{S_kS_k}^{-1}\right)M_{S_ki}\right| + \left|(\wh M_{iS_k} - M_{iS_k})^\T M_{S_kS_k}^{-1}M_{S_ki}\right|\\
			&:= |\wh M_{ii} - M_{ii}| + T_1 + T_2 + T_3. 
		\end{align*}
	 Note that Lemma \ref{lem_Gamma_M_Ihat} gives
	 \begin{equation}\label{bd_M_sup}
	 |\wh M_{ij} - M_{ij}|\lesssim \delta_n\quad \forall i,j\in \wh H.
	 \end{equation}
	 We proceed to bound $T_1, T_2$ and $T_3$, separately. 
	 
	 \paragraph{Bound $T_1$:} By the Cauchy-Schwarz inequality, 
	 \begin{align*}
	 	T_1^2 &\le \left\|\wh M_{S_ki} - M_{S_ki}\right\|_2^2   \wh M_{iS_k}^\T\wh M_{S_kS_k}^{-2} \wh M_{S_ki}\\
	 	&\le |S_k| \max_{j\in S_k}|\wh M_{ji} - M_{ji}|^2 \cdot  {\wh M_{iS_k}^\T\wh M_{S_kS_k}^{-1} \wh M_{S_ki} \over  \lambda_{k-1}(\wh M_{S_kS_k})} & (\text{by }(\ref{bd_M_sup}))\\
	 	&\overset{(i)}{\lesssim} (k-1)\delta_n^2\cdot  {1 \over \lambda_{k-1}(\wh M_{S_kS_k})}\\
	 	&\lesssim (k-1)\delta_n^2 & (\text{on } \E_M).
	 \end{align*}
	 In $(i)$, we used $|S_k| = k-1$, (\ref{bd_M_sup}) and the following result from Lemma \ref{lem_M_schur}
	 \[
	 	 \wh M_{iS_k}^\T\wh M_{S_kS_k}^{-1} \wh M_{S_ki} \lesssim 1.
	 \]

	 \paragraph{Bound $T_3$:} By similar arguments and Lemma \ref{lem_M_schur}, we have 
	 \begin{align*}
	 T_3^2 &\le \left\|\wh M_{S_ki} - M_{S_ki}\right\|_2^2    M_{iS_k}^\T M_{S_kS_k}^{-2}  M_{S_ki} \lesssim   {(k-1)\delta_n^2 M_{ii} \over \lambda_{k-1}(M_{S_kS_k})} \lesssim (k-1)\delta_n^2.
	 \end{align*}
	 The last inequality uses
		\begin{align}\label{lb_eigen_Msksk}\nonumber
			\lambda_{k-1}(M_{S_kS_k})  &= \lambda_{k-1}(B_{S_k\sbt}\C B_{S_k\sbt}^\T)\\\nonumber
			&\ge c_z  \lambda_{k-1}(B_{S_k\sbt} B_{S_k\sbt}^\T)\\\nonumber
			&\ge c_z \min_{1\le a\le K}\min_{i\in I_a}B_{ia}^2\\ &\ge c_z c_{b,I},
		\end{align}
	 with $c_z$ and $c_{b,I}$ defined in (\ref{cond_regularity}).
	 
	 \paragraph{Bound $T_2$:} The identify $P^{-1} = Q^{-1} + P^{-1}(Q-P) Q^{-1}$ for two invertible matrices gives 
	 \begin{align*}
	 	T_2^2  & = \left| \wh M_{iS_k}^\T \wh M_{S_kS_k}^{-1}\left(M_{S_kS_k} -  \wh M_{S_kS_k}\right)M_{S_kS_k}^{-1}M_{S_ki}\right|^2\\
	 	&\le  \wh M_{iS_k}^\T \wh M_{S_kS_k}^{-2} \wh M_{S_ki}  \left\|\wh M_{S_kS_k} -  M_{S_kS_k}\right\|_{\rm op}^2 M_{iS_k}^\T M_{S_kS_k}^{-2}M_{S_ki}\\
	 	&\lesssim \left\|\wh M_{S_kS_k} -  M_{S_kS_k}\right\|_{\rm op}^2
	 	&(\textrm{by Lemma \ref{lem_M_schur}})\\
	 	&\lesssim K\delta_n^2 &(\textrm{ on $\E_M$}).
	 \end{align*}
	 
	 Collecting all the bounds on $T_1$, $T_2$ and $T_3$ concludes 
	 \begin{align*}
	 		\left|\wh M_{ii | S_k} - M_{ii | S_k}\right|  &\lesssim \delta_n + \sqrt{K\delta_n^2} \lesssim \sqrt{K\delta_n^2}.
	 \end{align*}
	 Finally, notice that 
	 \[
	 			\left|\wh \Theta_{ii | S_k} - \Theta_{ii | S_k}\right| \le \wh\Sigma_{ii} 	\left|\wh M_{ii | S_k} - M_{ii | S_k}\right|  + \left|\wh\Sigma_{ii} - \Sigma_{ii}\right|  M_{ii | S_k}.
	 \]
	 The result then follows by Lemma \ref{lem_find_L} and the inequalities $M_{ii|S_k} \le M_{ii}\le 1$ and 
	 $$
	 	\wh\Sigma_{ii} \le \Sigma_{ii}+ |\wh\Sigma_{ii}-\Sigma_{ii}| \le \Sigma_{ii} (1+ \delta_n / 2) \lesssim \Sigma_{ii}.
	 $$
	\end{proof}
	
	\medskip
	
	\begin{lemma}\label{lem_signal}
		Under conditions of Theorem \ref{thm_I_post}, for any $1\le k\le K$ and $S_k$ defined in Definition \ref{def_Sk}, with $\varepsilon_n = c_*\sqrt{K}\delta_n$, we have
		\[
			\max_{i\in \cup_{b = k}^KI_{\pi(b)}}\min_{j\in  \cup_{a=1}^{k-1}I_{\pi(a)} \cup \bar J_1}\left( \Theta_{ii|S_k} - \Theta_{jj|S_k}-\left( \Sigma_{ii}+ \Sigma_{jj}\right)  \varepsilon_n\right) >0.
		\]
	\end{lemma}
	\begin{proof}
		Pick any $1\le k\le K$ and $S_k = \{i_1, \ldots, i_{k-1}\}\in \S_k$. Recall that $T_k = \{1,\ldots, k-1\}$ is the set of groups corresponding to elements of $S_k$ and further recall $T_k^c = [K]\setminus T_k$. Notice that
		\begin{equation}\label{ident_sets}
			\bigcup_{a=1}^{k-1} I_{\pi(a)} = \bigcup_{a' \in T_k}I_{a'},\qquad 
			\bigcup_{b =k}^K I_{\pi(b)} = \bigcup_{b'\in T_k^c}I_{b'}.
	    \end{equation}
		From (\ref{eq_Theta_schur_I}), recall that
		\begin{equation}\label{disp_Theta_I_T}
			\Theta_{jj | S_k} = 0,\qquad \forall 	j\in  I_a,  a\in T_k,
 		\end{equation}
 		and 
 		\begin{eqnarray}\label{def_max}\nonumber
 				\max_{i\in \cup_{b = k}^K I_{\pi(b)}} \Theta_{ii | S_k} 
 				&=&\max_{i\in  \cup_{b'\in T_k^c}I_{b'}} \Theta_{ii | S_k}\\ \nonumber
 				&=& \max_{b'\in T_k^c}\max_{i\in I_{b'}} A_{ib'}^2\cdot [\C]_{b'b'|T_k}\\\nonumber
 			& :=& \max_{i\in I_{b^*}} A_{ib^*}^2\cdot [\C]_{b^*b^*|T_k}\\
 			& :=&  A_{i^*b^*}^2 [\C]_{b^*b^*|T_k}.
 		\end{eqnarray}
 		On the other hand, for any $j\in \bar J_1$, (\ref{bd_Theta_jjsk}) yields 
 		\begin{align}\label{disp_Theta_Jbar}\nonumber
 			 \Theta_{jj | S_k} &~\le~ \|\wt A_{jT_k^c}\|_1^2 \cdot \max_{b\in T_k^c}G_{bb| T_k}\\ \nonumber
 			 & ~=~ \|\wt A_{jT_k^c}\|_1^2 \cdot  \max_{b\in T_k^c}\max_{i\in I_b}A_{ib}^2 [\C]_{bb|T_k}\\\nonumber
 			 & ~=~ \|\wt A_{jT_k^c}\|_1^2 \cdot A_{i^*b^*}^2 [\C]_{b^*b^*|T_k} & (\text{by }(\ref{def_max}))\\
 			 &~<~ (1 - \varepsilon)\cdot  A_{i^*b^*}^2 [\C]_{b^*b^*|T_k}
 		\end{align}
 		where we invoke condition (\ref{ass_J1_prime}) in the last step. We consider two cases:\\
 		{\bf Case 1:}   $j\in \cup_{a=1}^{k-1}I_{\pi(a)}$. 
 		Based on (\ref{disp_Theta_I_T}) -- (\ref{disp_Theta_Jbar}), we immediately have 
 		\begin{align*}
 			&\max_{i\in \cup_{b = k}^KI_{\pi(b)}}\min_{j\in  \cup_{a=1}^{k-1}I_{\pi(a)} }\left( \Theta_{ii|S_k} - \Theta_{jj|S_k}-\left( \Sigma_{ii}+ \Sigma_{jj}\right) \cdot \varepsilon_n\right)\\
 			&=  A_{i^*b^*}^2[\C]_{b^*b^*|T_k} - \left( \max_{i\in \cup_{b' \in T^c_k} I_{b'}}\Sigma_{ii}+ \max_{j\in  \cup_{a'\in T_k}I_{a'} }\Sigma_{jj}\right) \cdot \varepsilon_n
 		\end{align*}
 		by also using (\ref{ident_sets}).\\ 
 		{\bf Case 2:}   $j\in \bar J_1$. Similar arguments yield 
 		\begin{align*}
 		&\max_{i\in \cup_{b = k}^KI_{\pi(b)}}\min_{j\in \bar J_1 }\left( \Theta_{ii|S_k} - \Theta_{jj|S_k}-\left( \Sigma_{ii}+ \Sigma_{jj}\right) \cdot \varepsilon_n\right)\\
 		&>  \max_{i\in \cup_{b = k}^KI_{\pi(b)}}\min_{j\in \bar J_1}\left( \Theta_{ii|S_k} - (1-\varepsilon)A_{i^*b^*}^2[\C]_{b^*b^*|T_k} -\left( \Sigma_{ii}+ \Sigma_{jj}\right) \cdot \varepsilon_n\right) & (\text{by }(\ref{disp_Theta_Jbar}))\\
 		& \ge  \varepsilon \cdot A_{i^*b^*}^2[\C]_{b^*b^*|T_k} -\left( \max_{i\in \cup_{b'\in T^c_k} I_{b'}}\Sigma_{ii}+ \max_{j\in \bar J_1} \Sigma_{jj}\right) \cdot \varepsilon_n & (\text{by }(\ref{def_max})).
 		\end{align*}
 		The proof is accomplished by bounding from below the displays in both cases. To this end, we first observe that 
 		\begin{equation}\label{lb_C_bbT}
 			[\C]_{b^*b^*|T_k} \ge \lambda_K(\C) = c_z > 0.
 		\end{equation}
 		Since (\ref{cond_regularity}) implies 
 		\begin{align}\label{cond_signal_ratio}
 		    {\Sigma_{jj}  \over A_{j\sbt}^\T \C A_{j\sbt}} = {1\over B_{j\sbt}^\T \C B_{j\sbt}} \le {1\over c_zc_{b,I}} := C_1,\quad \forall j\in I\cup \bar J_1,
 		\end{align}
 		we then obtain
 		\begin{align}\label{ub_Sigma_IT}\nonumber
 			\max_{j\in \cup_{a'\in T_k}I_{a'}}\Sigma_{jj} &\overset{(\ref{cond_signal_ratio})}{\le} C_1\max_{j\in \cup_{a'\in T_k}I_{a'}}A_{j\sbt}^\T\C A_{j\sbt}\\\nonumber
 			& ~~ = ~~ C_1\max_{a'\in T_k}\max_{j\in I_{a'}}A_{ja'}^2\\
 			& \overset{(\ref{cond_pure_ratio})}{\le} C_1C_0 A_{i^*b^*}^2
 		\end{align}
 		where the last step also uses $(\ref{def_max})$ to invoke (\ref{cond_pure_ratio}). Similarly, 
 		\begin{align}\label{ub_Sigma_ITc}
 		\max_{i\in \cup_{b'\in T^c_k}I_{b'}}\Sigma_{ii} &\overset{(\ref{cond_signal_ratio})}{\le} C_1\max_{i\in \cup_{b'\in T^c_k}I_{b'}}A_{i\sbt}^\T\C A_{i\sbt}\overset{(\ref{cond_pure_ratio})}{\le} C_1C_0A_{i^*b^*}^2
 		\end{align}
 		and, finally, recalling that $A\C A^\T = \wt AD_\xi \C D_\xi \wt A^T = \wt A V \wt A^\T $ with $D_\xi$ defined in (\ref{def_D_A}), 
 	 	\begin{align}\label{ub_Sigma_Jbar}\nonumber
 		\max_{j\in \bar J_1}\Sigma_{jj} \overset{(\ref{cond_signal_ratio})}{\le} C_1\max_{j\in \bar J_1}A_{j\sbt}^\T\C A_{j\sbt}
 		& ~~ = ~~ C_1\max_{j\in  \bar J_1}\wt A_{j\sbt}^\T V  \wt A_{j\sbt}\\\nonumber
 		&~~ \le~~ C_1\max_{j\in  \bar J_1}\|\wt A_{j\sbt}\|_1^2 \max_{1\le b\le K}G_{bb} \\\nonumber
 		& ~~ \overset{(i)}{<} ~~  C_1\max_{1\le b\le K}\max_{i\in I_b} A_{ib}^2\\
 		&\overset{(\ref{cond_signal_ratio})}{\le}C_1C_0 A_{i^*b^*}^2.
 	\end{align}
 	Step $(i)$ uses $\|\wt A_{j\sbt}\|_1^2 < 1$ from condition (\ref{ass_J1_prime})
 	and $G_{bb} = \max_{i\in I_b}A_{ib}^2[\C]_{bb} =\max_{i\in I_b}A_{ib}^2$. 
 	Collecting (\ref{lb_C_bbT}) -- (\ref{ub_Sigma_Jbar}) concludes 
 	\begin{align*}
 		& A_{i^*b^*}^2[\C]_{b^*b^*|T_k} - \left( \max_{i\in \cup_{b \in T^c_k} I_{b}}\Sigma_{ii}+ \max_{j\in  \cup_{a\in T_k}I_{a} }\Sigma_{jj}\right) \cdot \varepsilon_n\\
 		&\ge A_{i^*b^*}^2 \left(
 		  c_z - C_1C_0\varepsilon_n
 		\right)
 	\end{align*}
 	and 
 	\begin{align*}
 		&\varepsilon\cdot A_{i^*b^*}^2[\C]_{b^*b^*|T_k} -\left( \max_{i\in \cup_{b\in T^c_k} I_{b}}\Sigma_{ii}+ \max_{j\in \bar J_1} \Sigma_{jj}\right)  \varepsilon_n \ge A_{i^*b^*}^2\left(c_z \varepsilon -2C_1C_0 \varepsilon_n\right).
 	\end{align*}
 	We finish the proof by invoking 
 	\[
 		\varepsilon_n = c_*\sqrt{K}\delta_n <  {c_z \over C_0C_1} = {c_z^2 c_{b,I} \over C_0}
 	\]
 	and, for sufficiently large $C'>0$,  
 	\[
 		\varepsilon = C'\sqrt{K}\delta_n > {2C_1C_0 \over c_z}c_*\sqrt{K}\delta_n = {2C_0 \over c_z^2 c_{b,I}}\varepsilon_n.
 	\]
	\end{proof}

      \subsubsection{Proof of Theorem \ref{thm_BI}}\label{app_proof_sec_est_A}
      
		The entire proof is deterministic under the event $\E$ intersected with $\{\wh I_k = I_{\pi(k)}\}_{1\le k\le K}$ for some group permutation $\pi$. Recall from Proposition \ref{prop_K} and Theorem \ref{thm_I_post} that this event holds with probability tending to one as $n\to \i$. We assume the identity $\pi$ for simplicity. Further we write
		\[
		\wh V_i = \|\wh R_{i,\nij}\|_q,\qquad V_i = \|R_{i,\nij}\|_q,\qquad \forall i\in I.
		\]
		Pick any $k\in [K]$ and $i\in I_k$. Also recall from (\ref{est_BI}) that 
		\[
			\wh B_{ik}^2 \le \max_{j\in I_k\setminus \{i\}}|\wh R_{ij}| {\|\wh R_{i,\nij}\|_q  \over\|\wh R_{j,\nij}\|_q } =  \max_{j\in I_k\setminus \{i\}}|\wh R_{ij}| {\wh V_i \over \wh V_j}.
		\]
		As we can only identify $B_{I\sbt}$ up to a signed permutation matrix, we further assume 
		$
			\sgn(\wh B_{ik}) = \sgn(B_{ik}) = 1.
		$ 
		On the other hand, from (\ref{est_BI_sign}), for any $j\in I_k\setminus\{i\}$, we have $\sgn(\wh B_{jk}) = \sgn(\wh R_{ij})$. Note that  $\E$ together with $\min_{i\in I_k}|B_{ik}|\ge c$ from (\ref{cond_regularity}) and $\log p \le c_0 n$ ensures 
		\[
			  \wh R_{ij}\ge R_{ij} - \delta_n  = B_{ik}B_{jk} - \delta_n \ge  c' B_{ik}B_{jk} = c' R_{ij}
		\] 
		for some constant $c'>0$. This further implies $\sgn(\wh B_{jk}) = \sgn(R_{ij}) = \sgn(B_{ik}B_{jk}) = \sgn(B_{jk})$. We thus conclude 
		\[
			\sgn(\wh B_{ik}) = \sgn(B_{ik}),\qquad \forall i\in I_k, k\in [K].
		\]
		We then prove the upper bound for 
		\begin{equation}\label{bd_BI_diff}
		|\wh B_{ik}^2 - B_{ik}^2| \le \max_{j\in I_k\setminus \{i\}}\left|
		|\wh R_{ij}| {\wh V_i \over \wh V_j}  - |R_{ij}| {V_i \over V_j} 
		\right|.
		\end{equation}
		Pick any $j\in I_k \setminus \{i\}$. By adding and subtracting terms, we have 
		\begin{align}\label{disp_B}
			\left|
			|\wh R_{ij}| {\wh V_i \over \wh V_j}  - |R_{ij}| {V_i \over V_j} 
			\right| &\le {V_i \over V_j}|\wh R_{ij}-R_{ij}| + |\wh R_{ij}|\left\{
				{\wh V_i \over \wh V_j}{|\wh V_j - V_j| \over V_j} + {|\wh V_i  - V_i| \over V_j}
			\right\}.
		\end{align}
		Note that the event $\E$ and $\log p \le c_0 n$ imply
		\begin{equation}\label{bd_R_hat}
			|\wh R_{ij}| \le |R_{ij}| + \delta_n \le c'' |R_{ij}| = c''|B_{ik}||B_{jk}|
		\end{equation}
		for some constant $c''>0$. The event $\E$ also gives
		\begin{equation}\label{bd_Vi_diff}
			\max_{1\le i\le p}|\wh V_i - V_i| \le (p-2)^{1/q}\delta_n.
		\end{equation}
		Since H\"older's inequality, (\ref{def_cbzr}) and (\ref{def_cbI}) guarantee that, for any $q\ge 2$,
		\begin{equation}\label{bd_lower_Vi}
		    {V_i \over (p-2)^{1/q}} \ge {\|R_{i,\nij}\|_2 \over \sqrt{p-2}} \ge \sqrt{c_{b,I}c_z c_r} \delta_n
			,\quad \forall i,j\in I_k, i\ne j,
		\end{equation}
		one can deduce from (\ref{cond_regularity}) that, for some constants $c',c''>0$,
		\begin{equation}\label{bd_Vi_hat}
				c'V_i \le \wh V_i \le c'' V_i,\qquad \forall i,j\in I_k, i\ne j.
		\end{equation}
		Plugging (\ref{bd_R_hat}), (\ref{bd_Vi_diff}) and (\ref{bd_Vi_hat}) into (\ref{disp_B}) gives
		\begin{align*}
			\left|
			|\wh R_{ij}| {\wh V_i \over \wh V_j}  - |R_{ij}| {V_i \over V_j} 
			\right| &\lesssim {V_i\over V_j}\delta_n + |B_{ik}||B_{jk}|\left\{
			{V_i \over V_j}{ (p-2)^{1/q} \over V_j} + {(p-2)^{1/q}\over V_j}
			\right\}\delta_n\\
			&\lesssim  {V_i\over V_j}\delta_n + |B_{ik}||B_{jk}|\left\{
			{V_i \over V_j} +1
			\right\}\delta_n &(\textrm{by }(\ref{bd_lower_Vi}))\\
			&\le {|B_{ik}|\over |B_{jk}|}\delta_n + |B_{ik}||B_{jk}|\left\{
			{|B_{ik}| \over |B_{jk}|} +1
			\right\}\delta_n &(\textrm{by } (\ref{eq_ratio}))\\
			&\lesssim |B_{ik}|\delta_n,
		\end{align*}
		We used  $|B_{ik}| \le 1$ and $|B_{jk}|\ge c$ from (\ref{cond_regularity}) to arrive at the last line. In view of (\ref{bd_BI_diff}), we conclude 
		\[
			|\wh B_{ik}^2 - B_{ik}^2| \lesssim  |B_{ik}|\delta_n.
		\]
		The first result then follows from $|\wh B_{ik}^2 - B_{ik}^2| = |\wh B_{ik} + B_{ik}|\cdot |\wh B_{ik} - B_{ik}|$ and $\sgn(\wh B_{ik}) = \sgn(B_{ik})$. 
		
		We proceed to prove the second result. Recall that $\whC$ is estimated from (\ref{est_C}) and satisfies
		\begin{align}\label{disp_Chat}\nonumber
			[\whC]_{ab} &=\sum_{i\in I_a} {\wh B_{ia} \over \sum_{i\in I_a}\wh B_{ia}^2}\sum_{j\in \wh I_b} { \wh B_{jb} \over \sum_{j\in \wh I_b}\wh B_{jb}^2}\wh R_{ij}\\
			& =  {1\over \wh d_a \wh d_b}\cdot{1\over |I_a||I_b|} \sum_{i\in I_a, j\in I_b}\wh B_{ia}\wh B_{jb} \wh R_{ij},\qquad \forall a\ne b \in [K],
		\end{align}
		meanwhile (\ref{ident_C}) gives 
		\[
			[\C]_{ab} = {1\over d_a  d_b}\cdot{1\over |I_a||I_b|} \sum_{i\in I_a, j\in I_b} B_{ia} B_{jb}  R_{ij},\qquad \forall a\ne b \in [K]
		\]
		where we write 
		$$
		\wh d_a = {1\over |I_a|}\sum_{i\in I_a} \wh B_{ia}^2,\qquad d_a = {1\over |I_a|}\sum_{i\in I_a} B_{ia}^2\qquad \forall a\in [K].
		$$
		We then have, by adding and subtracting terms, 
		\begin{align*}
			\left|
				[\whC]_{ab} - [\C]_{ab}
			\right| &\le {|d_ad_b - \wh d_a\wh d_b|\over \wh d_a \wh d_b d_a d_b}\cdot{1\over |I_a||I_b|} \left|\sum_{i\in I_a, j\in I_b}\wh B_{ia}\wh B_{jb} \wh R_{ij}\right|\\
			&\quad  + {1\over d_a  d_b}\cdot{1\over |I_a||I_b|} \sum_{i\in I_a, j\in I_b} \left|\wh B_{ia} \wh B_{jb}  \wh R_{ij} -  B_{ia} B_{jb}  R_{ij}\right|.
		\end{align*}
		We bound from above the two terms separately.
		Note that, by (\ref{disp_Chat}),
		\begin{align*}
			{|d_ad_b - \wh d_a\wh d_b|\over \wh d_a \wh d_b d_a d_b}\cdot{1\over |I_a||I_b|} \left|\sum_{i\in I_a, j\in I_b}\wh B_{ia}\wh B_{jb} \wh R_{ij}\right|
			&=  {|d_ad_b - \wh d_a\wh d_b|\over d_a d_b}\left|[\whC]_{ab}\right|\\
			&\le  {|\wh d_b - d_b |\over  d_b}+ {\wh d_b|\wh d_a- d_a|\over d_a d_b}.
		\end{align*}
		Note that (\ref{cond_regularity}) guarantees
		$$
		    d_a \ge  {c\over |I_a|}\sum_{i\in I_a}|B_{ia}|.
		$$
		Since
		\[
			 |\wh d_a - d_a| \le {1\over |I_a|}\sum_{i\in I_a}|\wh B_{ia}^2 - B_{ia}^2| \lesssim  {1\over |I_a|}\sum_{i\in I_a}|B_{ia}|\delta_n,
		\] 
		and
		\[
			{\wh d_b \over d_b} \le 1 + {|\wh d_b - d_b| \over d_b} \lesssim 1 +   \delta_n \lesssim 1
		\]
		by using $\log p \le c_0 n$,
		we conclude 
		 \[
		 			{|d_ad_b - \wh d_a\wh d_b|\over \wh d_a \wh d_b d_a d_b}\cdot{1\over |I_a||I_b|} \left|\sum_{i\in I_a, j\in I_b}\wh B_{ia}\wh B_{jb} \wh R_{ij}\right| \lesssim \delta_n.
		 \]
		 On the other hand, 
		 \begin{align*}
		 	  \left|\wh B_{ia} \wh B_{jb}  \wh R_{ij} -  B_{ia} B_{jb}  R_{ij}\right| &\le  \left|\wh B_{ia} \wh B_{jb} -  B_{ia} B_{jb} \right||\wh R_{ij}|+  |B_{ia}| |B_{jb}|\left|\wh R_{ij} -    R_{ij}\right|\\
		 	  &\le \left|\wh B_{ia}-B_{ia}\right| |\wh B_{jb}| +  |B_{ia}|\left|\wh B_{jb}- B_{jb} \right|+  |B_{ia}| |B_{jb}|\delta_n\\
		 	  &\lesssim (|B_{ia}| + |B_{jb}|)\delta_n
		 \end{align*}
		 by using $|B_{ia}|\le 1$, $|B_{ia}| \ge c$ from (\ref{cond_regularity}) and $|\wh B_{jb}| \lesssim |B_{jb}| + \delta_n \lesssim |B_{jb}|$  in the last line. Using $d_a \gtrsim |I_a|^{-1}\sum_{i\in I_a}|B_{ia}|$ again together with  $|B_{ia}| \ge c$ concludes 
		 \[
		 	{1\over d_a  d_b}\cdot{1\over |I_a||I_b|} \sum_{i\in I_a, j\in I_b} \left|\wh B_{ia} \wh B_{jb}  \wh R_{ij} -  B_{ia} B_{jb}  R_{ij}\right| \lesssim \delta_n
		 \]
		which finishes the proof of the second result. 
		
		Finally, we prove the rates of $\|\wh A_{i\sbt} - A_{i\sbt}\|_\i$. We first remark that Lemma \ref{lem_dev_Sigma} holds on the $\E$ by inspecting the proof of Lemma \ref{lem_dev_R}. 
		By definition of $\wh A = D_{\wh \Sigma}^{1/2}\wh B$, we have 
		\begin{align*}
		    |\wh A_{ik}^2 - A_{ik}^2| & = |\wh\Sigma_{ii}\wh B_{ik}^2 - \Sigma_{ii}B_{ik}^2|\\
		    &\le \wh\Sigma_{ii}|\wh B_{ik}^2 -B_{ik}^2| +  B_{ik}^2|\wh \Sigma_{ii} - \Sigma_{ii}|\\
		    &\le (\Sigma_{ii} + |\wh \Sigma_{ii}-\Sigma_{ii}|)|\wh B_{ik}^2 -B_{ik}^2| +  B_{ik}^2|\wh \Sigma_{ii} - \Sigma_{ii}|.
		\end{align*}
		Invoking the rates of $|\wh B_{ik}^2 -B_{ik}^2|$ and Lemma \ref{lem_dev_Sigma} together with $\log p \le n$ yields
		\[
		     |\wh A_{ik}^2 - A_{ik}^2|  \lesssim |B_{ik}| \Sigma_{ii} \delta_n  + B_{ik}^2\Sigma_{ii}\delta_n=  |A_{ik}| \sqrt{\Sigma_{ii}} \delta_n  + A_{ik}^2\delta_n \lesssim |A_{ik}| \sqrt{\Sigma_{ii}} \delta_n
		\]
		where we used $A_{ik}= \sqrt{\Sigma_{ii}}B_{ik}$ in the equality and $\Sigma_{ii} = A_{ik}^2+[\Sigma_E]_{ii} \ge A_{ik}^2$ in the last step. Since 
		\begin{align*}
		    \|\wh A_{i\sbt} - A_{i\sbt}\|_\i  = |\wh A_{ik} - A_{ik}| = {|\wh A_{ik}^2 - A_{ik}^2| \over |\wh A_{ik}|+|A_{ik}|} \le {|\wh A_{ik}^2 - A_{ik}^2| \over |A_{ik}|},
		\end{align*} the result follows.
		\qed	\\

 	\subsection{Technical Lemmas for controlling quantities related with $\wh M$}\label{app_lemmas_M}
	
	In this section, we provide technical lemmas that control different quantities related with $\wh M_{\wh H\wh H}$ constructed from (\ref{est_M_LL_off_diag}) -- (\ref{est_M_LL_diag}). Also note that
	$\wh M_{\wh H\wh H} =  \wh R_{\wh H\wh H} - \wh \Gamma_{\wh H\wh H}$
	with $\wh\Gamma_{\wh H\wh H}$ constructed from (\ref{est_Gamma_II}). 

	The first lemma provides an element-wise control of $\wh M_{ij} - M_{ij}$ for all $i,j\in \wh H$. The second lemma provides both the upper bounds of $\wh M_{S_kS_k} - M_{S_kS_k}$ in operator norm and the lower bounds for $\lambda_{k-1}(\wh M_{S_kS_k})$. Both results hold uniformly for all $S_k\in \S_k$ and $1\le k\le K$. The third lemma provides upper bounds for the quadratic terms $M_{iS_k}^\T M_{S_kS_k}^{-1}M_{S_ki}$ and its sample level counterpart, as well uniformly over $i\in \wh H$, $S_k\in \S_k$ and $1\le k\le K$.
	
	Finally, the last two lemmas are variations of Lemmas \ref{lem_Gamma_M_Ihat} and \ref{lem_op_M} under conditions in Section \ref{sec_theory_H_K}. 
	
	\begin{lemma}\label{lem_Gamma_M_Ihat}
	     Under model (\ref{model}), assume Assumptions \ref{ass_I} \& \ref{ass_distr}, (\ref{cond_signal_ell_q_relaxed}), $\log p\le cn$ for sufficiently small constant $c>0$ and 
	     \[
	        \min(c_{b,I}, c_z, c_r) > c'
	     \]
	     for some absolute constant $c'>0$. Then on the event $\E$, 
	    \begin{align*}
	        \max_{i\in \wh H}\left|\wh \Gamma_{ii} - \Gamma_{ii}\right| \lesssim \delta_n,\qquad
	        &\max_{i,j\in \wh H}\left|\wh M_{ij} - M_{ij}\right| \lesssim \delta_n.
	    \end{align*}
	\end{lemma}
	\begin{proof}
	    The entire proof is on the event $\E$. Recall that $\wh H\subseteq I \cup \bar J_1$ from Corollary \ref{cor_I_prime}. Pick any $i\in \wh H_k$ with some $k\in [\wh G]$.
	    Recall from (\ref{est_Gamma_II}) and (\ref{est_BI}) that 
	    \[
	        \wh \Gamma_{ii} = 1 - \wh B_{ik}^2 = 1- |\wh R_{i\wh j}| {\|\wh R_{i,\setminus\{i,\wh j\}}\|_2  \over\|\wh R_{\wh j,\setminus\{i,\wh j\}}\|_2}
	    \]
	   with 
	   \[\wh j = \arg\min_{\ell \in \wh H_k\setminus \{i\}} \wh S_2(i,\ell).\]
	   Further recall that $\Gamma_{ii} = 1 - B_{i\sbt}^\T \C B_{i\sbt}$ and, for simplicity, write 
	   $$
	   \wh v_{i} = \|\wh R_{\wh j,\setminus\{i,\wh j\}}\|_2, \quad \wh v_{\wh j} = \|\wh R_{\wh j,\setminus\{i,\wh j\}}\|_2,  \quad v_{i} = \| R_{\wh j,\setminus\{i,\wh j\}}\|_2, \quad v_{\wh j} = \| R_{\wh j,\setminus\{i,\wh j\}}\|_2.
	   $$
	   We have 
	   \[
	        \left| \wh \Gamma_{ii} - \Gamma_{ii}\right| \le \left|
	        |\wh R_{i\wh j}| { \wh v_i \over \wh v_{\wh j}} - |R_{i\wh j}|{v_i \over v_{\wh j}}
	        \right| + \left| |R_{i\wh j}|{v_i \over v_{\wh j}} - B_{i\sbt}^\T \C B_{i\sbt}\right|.
	   \]
	   To bound the right hand side of the above display,  definitions in (\ref{def_cbzr}) and (\ref{def_cbI}) yield
	   \begin{equation}\label{bd_lower_vi}
	    v_i^2  = \left\| B_{i\sbt}^\T \C B_{\setminus\{i,\wh j\}}^\T\right\|_2^2 \ge (p-2)c_{b,I} c_zc_r.
	   \end{equation}
	   Together with (\ref{bd_Vi_diff}) by taking $q = 2$ and $\delta_n\lesssim 1$ (implied by $\log p\lesssim n$), we have 
	   \[
	        c'v_i \le \wh v_i \le c''v_i,\qquad \forall i\in [p].
	   \]
	   By similar arguments in the proof of Theorem \ref{thm_BI} (see (\ref{disp_B}) -- (\ref{bd_Vi_hat}) and their subsequent arguments), we obtain 
	   \begin{align*}
	        \left|
	        |\wh R_{i\wh j}| { \wh v_i \over \wh v_{\wh j}} - |R_{i\wh j}|{v_i \over v_{\wh j}}
	        \right| &\lesssim {v_i\over v_{\wh j}}\delta_n + (|R_{i\wh j}| + \delta_n)\left\{
			{v_i \over v_{\wh j}}{\sqrt{p-2} \over v_{\wh j}} + {\sqrt{p-2}\over v_{\wh j}}
			\right\}\delta_n\\
			&\lesssim  {v_i\over \sqrt{p-2}}\delta_n + (|R_{i\wh j}| + \delta_n)\left(
			{v_i \over \sqrt{p-2}} +1
			\right)\delta_n &(\textrm{by }(\ref{bd_lower_vi}))\\
			&\le \delta_n.
		\end{align*}
		where in the last step we used $|R_{i\wh j}|\le 1$, $\delta_n \lesssim 1$ (implied by $\log p\le cn$) together with 
		\[
		    \max_iv_i^2  \le \sum_{\ell \ne i,\wh j}(B_{i\sbt}^\T \C B_{\ell \sbt})^2 \le \sum_{\ell \ne i,\wh j}B_{i\sbt}^\T\C B_{i\sbt}B_{\ell\sbt}^\T\C B_{\ell \sbt}\le p-2.
		\]
		 It then suffices to show 
		\[
		    |\Delta_{i\wh j}| := \left| |R_{i\wh j}|{v_i \over v_{\wh j}} - B_{i\sbt}^\T \C B_{i\sbt}\right| \lesssim \delta_n.
		\]
		By adding and subtracting terms, for any $\rho \ne  0$ to be chosen later, we have 
		\begin{align*}
		    \Delta_{i\wh j} & = |R_{i\wh j}|{v_i \over v_{\wh j}} - B_{i\sbt}^\T \C B_{i\sbt}\\
		    & = |R_{i\wh j}|{v_i \over v_{\wh j}} - |R_{i\wh j}|{v_i \over |\rho| v_i} + |R_{i\wh j}|{v_i \over |\rho| v_i} - B_{i\sbt}^\T \C B_{i\sbt}\\
		    &= |R_{i\wh j}| {|\rho| v_i - \wh v_j \over |\rho| v_{\wh j}} + {|B_{i\sbt}^\T \C B_{\wh j \sbt}| \over |\rho|} - B_{i\sbt}^\T \C B_{i\sbt},
		\end{align*}
		which yields
		\begin{align}\label{bd_Delta_ij}\nonumber
		    |\Delta_{i\wh j}| &\le {|R_{i\wh j}|\over |\rho| v_{\wh j}} \left| |\rho|\cdot \left\|R_{i,\setminus\{i,\wh j\}}\right\|_2 -  \left\|R_{\wh j,\setminus\{i,\wh j\}}\right\|_2 \right| + {1\over |\rho|}\left|
		        |\rho|\cdot |B_{i\sbt}^\T \C B_{\wh j \sbt}| - B_{i\sbt}^\T \C B_{i\sbt}
		    \right|\\\nonumber
		    &\overset{(i)}{\lesssim} {1\over |\rho|\sqrt{p-2}} \left\|
		    \left(\rho \wh B_{i\sbt} - \wh B_{\wh j\sbt}\right)^\T \C B_{\setminus\{i,\wh j\}}^\T
		    \right\|_2 + {1\over |\rho|}\left|
		         B_{i\sbt}^\T \C (\rho B_{i\sbt} -  B_{\wh j \sbt})
		    \right|\\
		    &\overset{(ii)}{\lesssim} {1\over |\rho|}\left\|\C^{1/2}(\rho B_{i\sbt} -  B_{\wh j \sbt})
		    \right\|_2
		\end{align}
		In $(i)$, we used $|R_{i\wh j}|\le 1$ and (\ref{bd_lower_vi}) and in $(ii)$ we applied the Cauchy-Schwarz inequality and used
		\[
		    {1\over p-2} \left\|
		      \C^{1/2} B_{\setminus\{i,\wh j\}}^\T
		    \right\|_2^2 = {1\over  p-2} \sum_{\ell \ne i,\wh j}B_{\ell\sbt}^\T\C B_{\ell \sbt} \le 1
		\]
		from the fact that $\max_i B_{i\sbt}^\T \C B_{i\sbt} \le 1$.
		We proceed to bound $\|\C^{1/2}(\rho B_{i\sbt} -  B_{\wh j \sbt})\|_2 / |\rho|$ from above. Recall that 
		\[
		    \wh j = \arg\min_{\ell \in \wh H_k\setminus \{i\}} \wh S_2(i,\ell).
		\]
		It implies $\wh S_2(i,\wh j) \lesssim \delta_n$ which further gives 
		$$
		    S_2(i,\wh j) \le \wh S_2(i,\wh j)  + |\wh S_2(i,\wh j)-S_2(i,\wh j)|\lesssim \delta_n
		$$ by Theorem \ref{thm_score_ell_q}. On the other hand, the third line in (\ref{lb_score}) yields 
		\begin{align}\label{lb_score_2}\nonumber
		    [S_2(i,\wh j)]^2 &\ge c_r \min_{\|v\|_\i=1}\left\|(v_1 B_{i\sbt} + v_2 B_{\wh j\sbt})^\T \C^{1/2}\right\|_2^2\\\nonumber
		    & = c_r \min\left\{\min_{|v|\le 1}\left\|(B_{i\sbt} + vB_{\wh j\sbt})^\T \C^{1/2}\right\|_2^2, \min_{|v|\le 1}\left\|(vB_{i\sbt} + B_{\wh j\sbt})^\T \C^{1/2}\right\|_2^2\right\}\\
		    &:= c_r \left\|(v_1^* B_{i\sbt} + v_2^* B_{\wh j\sbt})^\T \C^{1/2}\right\|_2^2 
		\end{align}
		from the definitions in (\ref{def_cbzr}). When $v_1^* = 1$, choose $\rho = -(1 / v_2^*)$ in (\ref{bd_Delta_ij}) to obtain 
		\[
		    |\Delta_{i\wh j}| \lesssim {1\over |\rho|}\left\|\C^{1/2}(\rho B_{i\sbt} -  B_{\wh j \sbt})
		    \right\|_2 =  \left\|\C^{1/2}(B_{i\sbt} + v_2^*  B_{\wh j \sbt}) 
		    \right\|_2 \le {S_2(i,\wh j) \over  \sqrt c_r} \lesssim \delta_n.
		\]
		By similar arguments, when $v_2^* = 1$, 
		choose $\rho = -v_1^*$ in (\ref{bd_Delta_ij}) to obtain
		\[
		    |\Delta_{i\wh j}| \lesssim {1\over |\rho|}\left\|\C^{1/2}(\rho B_{i\sbt} -  B_{\wh j \sbt})
		    \right\|_2 = {1\over |v_1^*|}\left\|\C^{1/2}(v_1^* B_{i\sbt} + B_{\wh j \sbt})
		    \right\|_2 \le {S_2(i,\wh j) \over  |v_1^*|\sqrt c_r} \lesssim {\delta_n \over |v_1^*|}.
		\]
		To conclude $ |\Delta_{i\wh j}| \lesssim \delta_n$, it remains to lower bound $|v_1^*|$ when $|v_1^*|\ne 1$ and $v_2^* = 1$. Since the proof of Proposition \ref{prop_score_2} reveals that 
		\[
		    \left\|(v_1^* B_{i\sbt} + B_{\wh j\sbt})^\T \C^{1/2}\right\|_2^2 = B_{\wh j\sbt}^\T \C B_{\wh j\sbt} \left(
		     1 - {|B_{i\sbt}^\T \C B_{\wh j\sbt}|^2 \over B_{i\sbt}^\T\C B_{i\sbt}B_{\wh j\sbt}^\T\C B_{\wh j\sbt}}
		    \right) = B_{\wh j\sbt}^\T \C B_{\wh j\sbt}  - |v_1^*|^2
		\]
		with 
		\[
		    |v_1^*| = {|B_{i\sbt}^\T \C B_{\wh j\sbt}| \over B_{i\sbt}^\T \C B_{i\sbt}},
		\]
		we conclude that 
		\[
		    |v_1^*|^2 = B_{\wh j\sbt}^\T \C B_{\wh j\sbt} - \left\|(v_1^* B_{i\sbt} + B_{\wh j\sbt})^\T \C^{1/2}\right\|_2^2  \overset{(\ref{lb_score_2})}{\ge} c_{b,I}c_z - {[S_2(i,\wh j)]^2 \over c_r} \gtrsim c_{b,I}c_z
		\]
		by also invoking the definitions of the quantities in (\ref{def_cbzr}) and the inequality $S_2(i, \wh j) \lesssim \delta_n$. This concludes 
		$$
		\max_{i\in \wh H}|\wh \Gamma_{ii} - \Gamma_{ii}| \lesssim \delta_n.
		$$
		The rate of $\wh M_{ij} - M_{ij}$ follows immediately by noting that
		\[
		    \left|\wh M_{ij} - M_{ij} \right| \le \left|\wh R_{ij} - R_{ij} \right| + |\wh \Gamma_{ij} - \Gamma_{ij}|\cdot 1_{\{i=j\}},\quad \forall i,j\in \wh H.
		\]
	\end{proof}

	\begin{lemma}\label{lem_op_M}
		Under conditions of Lemma \ref{lem_Gamma_M_Ihat},  on the event $\E$,
		$$
				\max_{1\le k\le K}\max_{S_k \in \S_k}\left\|\wh M_{S_kS_k} -  M_{S_kS_k}\right\|_{\rm op} \lesssim \bar c_z(\sqrt{K\delta_n^2} + K\delta_n^2).
		$$
		Furthermore, if
		$
			\bar c_z (\sqrt{K\delta_n^2} + K\delta_n^2) \le c_1
		$
		for some sufficiently small constant $c_1>0$, then
		\[
		    \min_{1\le k\le K}\min_{S_k\in \S_k}\lambda_{k-1}(\wh M_{S_kS_k}) \gtrsim c_z c_{b,I}\]
		Both statements are valid on an event that holds with   probability greater than $1-C(p\vee n)^{-\alpha}$, for some $0< C, \alpha<\infty$.
	\end{lemma}
	\begin{proof}
		To prove   the first assertion, 
		pick any $k\in [K]$ and $S_k\in \S_k$. We have 
		\[
			\|\wh M_{S_kS_k} -  M_{S_kS_k}\|_{\rm op} \le 
			\|\wh R_{S_kS_k} -  R_{S_kS_k}\|_{\rm op} + \|\wh \Gamma_{S_kS_k} -  \Gamma_{S_kS_k}\|_{\rm op}. 
		\]
		On the event $\E$, the inclusion $I\subseteq \wh H$ from Corollary \ref{cor_I_prime}, and Lemma \ref{lem_Gamma_M_Ihat} ensure
		\begin{equation}\label{bd_Gamma_sksk}
			\|\wh \Gamma_{S_kS_k} -  \Gamma_{S_kS_k}\|_{\rm op} = 	\|\wh \Gamma_{S_kS_k} -  \Gamma_{S_kS_k}\|_{\i} \le \max_{i\in I} |\wh \Gamma_{ii} - \Gamma_{ii}| \lesssim \delta_n.  
		\end{equation}
        To bound $\|\wh R_{S_kS_k} -  R_{S_kS_k}\|_{\rm op}$, note that, on the event $\E\cap \E_{S_k,t}$ by taking $S=S_k$ in Lemma \ref{lem_op_R_SS}, we have, for all $t\ge 0$, 
        \[
		        \|\wh R_{S_kS_k} -  R_{S_kS_k}\|_{\rm op} \lesssim \|R_{S_kS_k}\|_{\op} \left(
		        \sqrt{|S_k| \over n}+{|S_k|\over n} + \sqrt{t\over n} + {t\over n} + \delta_n
		        \right).
		 \]
		Therefore, the above 
		holds uniformly over $1\le k\le K$ and $S_k\in \S_k$ with probability at least 
		\[
		    \PP\left\{
		        \E \cap \left(\bigcap_{k=1}^K \bigcap _{S_k\in\S_k}\E_{S,t}\right)
		    \right\} \ge 1 - \PP(\E^c) - \sum_{k=1}^K|\S_k| \cdot 2\exp(-ct).
		\]
		Recall that $\sum_{k=1}^K|\S_k| \le Kp^{K-1}\le p^{K}$. Taking $t = CK\log(n\vee p)$ with $C$ large enough and using $|S_k| < K$ conclude
		\[
	        \|\wh R_{S_kS_k} -  R_{S_kS_k}\|_{\rm op} \lesssim \|R_{S_kS_k}\|_{\op} \left(
	        \sqrt{K\delta_n^2} + {K\delta_n^2}
	        \right)
		\]
		uniformly over $1\le k\le K$ and $S_k\in \S_k$ with probability at least $1-4(n\vee p)^{-1} - 2(n\vee p)^{-c'K}$. Then observe that 
		\begin{equation}\label{bd_R_sksk}
			\|R_{S_kS_k}\|_{\rm op}\le \|B_{S_k\sbt}\C B_{S_k\sbt}^\T\|_{\rm op} + \|\Gamma\|_{\rm op} \le  \max_{i\in I_a, a\in [K]}\left(\bar c_z B_{ia}^2 + \Gamma_{ii}\right) \le 1+\bar c_z
		\end{equation}
		by using the fact that $B_{S_k\sbt} B_{S_k\sbt}^\T$ is diagonal with diagonal elements equal to $B_{i_a\pi(a)}^2$ for $1\le a\le k-1$. 
		In conjunction with (\ref{bd_Gamma_sksk}), this completes the proof of the first claim.
		
	The second claim follows from 
		Weyl's inequality, 
		\begin{align*}
			\lambda_{k-1}(\wh M_{S_kS_k}) & \ge 	\lambda_{k-1}(M_{S_kS_k})  - \left\|
				\wh M_{S_kS_k} - M_{S_kS_k}
			\right\|_{\rm op},
		\end{align*}
		the first assertion, $(\ref{lb_eigen_Msksk})$ and  $\bar c_z \sqrt{K\delta_n^2} \le c_1$ for sufficiently small $c_1$.
	\end{proof}
	
	\bigskip

	   	 \begin{lemma}\label{lem_M_schur}
	    Under conditions of Lemma \ref{lem_op_M}, we have
	    \[
	       \max_{1\le k\le K}\max_{S_k\in \S_k}\max_{i\in \wh H} M_{iS_k}^\T M_{S_kS_k}^{-1} M_{S_ki} \le 1.
	    \]
	    Furthermore, on the event $\E \cap \E_M$ with $\E_M$ defined in (\ref{def_event_E_M}), we have 
	    \[
	        \max_{1\le k\le K}\max_{S_k\in \S_k}\max_{i\in \wh H} \wh M_{iS_k}^\T\wh M_{S_kS_k}^{-1} \wh M_{S_ki} \lesssim 1.
	    \]
	 \end{lemma}
	 \begin{proof}
	    Pick any $k\in [K]$, $S_k\in \S_k$ and $i\in \wh H$. To prove the first result, note that $|S_k| = k-1$. When $i\in S_k$, we have 
	    \[
	        M_{ii} - M_{iS_k}^\T M_{S_kS_k}^{-1} M_{S_ki} \ge \lambda_{k-1}(M_{S_kS_k})\overset{(\ref{lb_eigen_Msksk})}{>} 0
	    \]
	    which in conjunction with $M_{ii} = R_{ii} - \Gamma_{ii} \le 1$ implies the result. When $i\notin S_k$, write $\bar S_k = S_k \cup \{i\}$ and 
	    observe that
	    \[
	        M_{ii} - M_{iS_k}^\T M_{S_kS_k}^{-1} M_{S_ki} \ge \lambda_{k}(M_{\bar S_k\bar S_k}) \ge 0.
	    \]
	    Using $M_{ii} \le 1$ again concludes the first result. 
	    
	    We proceed to prove the second result by working on $\E\cap \E_M$. By adding and subtracting terms,
	    \begin{align}\label{disp_0}\nonumber
	        \wh M_{iS_k}^\T\wh M_{S_kS_k}^{-1} \wh M_{S_ki} &\le \left|(\wh M_{iS_k} -M_{iS_k})^\T\wh M_{S_kS_k}^{-1} \wh M_{S_ki}\right| + \left|M_{iS_k}^\T\left(\wh M_{S_kS_k}^{-1} - M_{S_kS_K}^{-1}\right) \wh M_{S_ki}\right|\\
	        &\quad + \left|M_{iS_k}^\T M_{S_kS_K}^{-1}(\wh M_{S_ki}-M_{S_k i})\right| + M_{iS_k}^\T M_{S_kS_K}^{-1} M_{S_ki}.
	    \end{align}
	    Note that 
    	\begin{align}\label{disp_1}\nonumber
    	    \left|(\wh M_{iS_k} -M_{iS_k})^\T\wh M_{S_kS_k}^{-1} \wh M_{S_ki}\right|^2 &\le \left\|\wh M_{iS_k} -M_{iS_k}\right\|_2^2 \wh M_{iS_k}^\T \wh M_{S_kS_k}^{-2} \wh M_{S_ki}\\\nonumber
    	    &\lesssim K\delta_n^2 {\wh M_{iS_k}^\T \wh M_{S_kS_k}^{-1} \wh M_{S_ki}\over \lambda_{k-1}(\wh M_{S_kS_k})} &(\textrm{by (\ref{bd_M_sup})})\\
    	    &\lesssim K\delta_n^2 \cdot \wh M_{iS_k}^\T \wh M_{S_kS_k}^{-1} \wh M_{S_ki} &(\textrm{on $\E_M$})
    	\end{align}
    	and, similar arguments together with (\ref{lb_eigen_Msksk}) yield 
    	\begin{align}\label{disp_2}
    	    \left|M_{iS_k}^\T M_{S_kS_K}^{-1}(\wh M_{S_ki}-M_{S_k i})\right| &\lesssim \sqrt{K\delta_n^2} \cdot M_{iS_k}^\T  M_{S_kS_k}^{-1} M_{S_ki}\le \sqrt{K\delta_n^2}.
    	\end{align}
	    Furthermore, 
	    \begin{align}\label{disp_3}\nonumber
	        &\left|M_{iS_k}^\T\left(\wh M_{S_kS_k}^{-1} - M_{S_kS_K}^{-1}\right) \wh M_{S_ki}\right|\\ \nonumber
	        & = \left|M_{iS_k}^\T  M_{S_kS_K}^{-1}\left( M_{S_kS_K}-\wh M_{S_kS_k}\right) \wh M_{S_kS_k}^{-1}\wh M_{S_ki}\right|\\\nonumber
	        &\le \sqrt{M_{iS_k}^\T  M_{S_kS_k}^{-2}M_{S_ki}}\left\| M_{S_kS_K}-\wh M_{S_kS_k}\right\|_{\rm op}\sqrt{ \wh M_{iS_k}^\T \wh M_{S_kS_k}^{-2}\wh M_{S_ki}}\\\nonumber
	        &\overset{(i)}{\lesssim}\left\| M_{S_kS_K}-\wh M_{S_kS_k}\right\|_{\rm op}\left(\wh M_{iS_k}^\T \wh M_{S_kS_k}^{-1}\wh M_{S_ki}\over \lambda_{k-1}(\wh M_{S_kS_k})\right)^{1/2} \\
	        &\lesssim \sqrt{K\delta_n^2}\sqrt{ \wh M_{iS_k}^\T \wh M_{S_kS_k}^{-1}\wh M_{S_ki}} &(\textrm{on $\E_M$})
	    \end{align}
	    where $(i)$ uses the first result of this lemma together with (\ref{lb_eigen_Msksk}). Plugging (\ref{disp_1}) -- (\ref{disp_3}) into (\ref{disp_0}) and using the condition that $K\delta_n^2$ is sufficiently small complete the proof. 
	 \end{proof}
	 
    \bigskip
    
    \begin{lemma}\label{lem_op_M_H}
	    Under conditions of Proposition \ref{prop_K_parallel}, on the event $\E$, for all $L=\{\ell_1,\ldots,\ell_G\}$ with $\ell_k \in H_k$ and $k\in [G]$,
		$$
			\left\|\wh M_{LL} -  M_{LL}\right\|_{\rm op} \lesssim \ol c(L)(\sqrt{G\delta_n^2} + G\delta_n^2).
		$$
	\end{lemma}
    \begin{proof}
        Notice that 
        \[
            \wh M_{\wh L \wh L} = \wh R_{\wh L\wh L} - \wh \Gamma_{\wh L\wh L} 
        \]
        where $\wh \Gamma_{\wh L\wh L}$ is a diagonal matrix defined in (\ref{est_Gamma_II}).
        The proof follows the same lines of arguments as that of the first statement in Lemma \ref{lem_op_M}, provided that,
        \[
            \max_{i\in \wh H}\left|\wh \Gamma_{ii} - \Gamma_{ii}\right| \lesssim \delta_n
        \]
        on the event $\E$  under conditions in Proposition \ref{prop_K_parallel}. To verify the above claim, recall that $\wh H = H$ from Corollary \ref{cor_H}. By inspecting the proof of Lemma \ref{lem_Gamma_M_Ihat}, we have 
        \[
	        \left| \wh \Gamma_{ii} - \Gamma_{ii}\right| \le \left|
	        |\wh R_{i\wh j}| { \wh V_i \over \wh V_{\wh j}} - |R_{i\wh j}|{V_i \over V_{\wh j}}
	        \right| + \left| |R_{i\wh j}|{V_i \over V_{\wh j}} - B_{i\sbt}^\T \C B_{i\sbt}\right|
	   \]
	   where 
	   \[
	            V_i = \|R_{i,\nij}\|_q,\quad \wh V_i = \|\wh R_{i,\nij}\|_q.
	   \]
	   Since (\ref{cond_regularity_parallel}) and (\ref{bd_Vi_diff}) ensure that 
	   \[
	         c'V_i \le \wh V_i \le c''V_i,\qquad \forall i,j\in H_k, i\ne j
	   \]
	   on the event $\E$. By similar arguments of the proof of Lemma \ref{lem_Gamma_M_Ihat} together with $V_i \le (p-2)^{1/q}$, we obtain
	   \begin{align*}
	        \left|
	        |\wh R_{i\wh j}| { \wh V_i \over \wh V_{\wh j}} - |R_{i\wh j}|{V_i \over V_{\wh j}}
	        \right| &\lesssim {V_i\over V_{\wh j}}\delta_n + (|R_{i\wh j}| + \delta_n)\left\{
			{V_i \over V_{\wh j}}{(p-2)^{1/q} \over V_{\wh j}} + {(p-2)^{1/q}\over V_{\wh j}}
			\right\}\delta_n\\
			&\lesssim  {V_i\over (p-2)^{1/q}}\delta_n + (|R_{i\wh j}| + \delta_n)\left(
			{V_i \over(p-2)^{1/q}} +1
			\right)\delta_n &(\textrm{by }(\ref{bd_lower_vi}))\\
			&\le \delta_n.
		\end{align*}
		Furthermore, (\ref{ident_MI}) together with $\wh H_k = H_k$ implies
		\[
		    \left| |R_{i\wh j}|{V_i \over V_{\wh j}} - B_{i\sbt}^\T \C B_{i\sbt}\right| = 0.
		\]
		The proof is complete.
    \end{proof}
    
    \bigskip

	\subsection{Concentration inequalities of the sample correlation matrix}\label{app_concentration_R}
	Throughout this section, let $\X\in \RR^{n\times p}$ be any random matrix whose rows $\X_{i\sbt}$ are $n$ independent sub-Gaussian random vectors in $\RR^p$ with zero mean and second moment $\Sigma$. We provide non-asymptotic upper bounds for its sample correlation matrix 
	$$
	    \wh R := D_{\wh \Sigma}^{-1/2}\wh \Sigma D_{\wh \Sigma}^{-1/2}
	$$ where $\wh\Sigma = n^{-1}\X^\T \X$ and $D_{\wh \Sigma} = \diag(\wh\Sigma_{11}, \ldots, \wh\Sigma_{dd})$. 
	
    \begin{lemma}\label{lem_op_R_SS}
        Let $\X\in \RR^{n\times p}$ be a random matrix with $\Sigma^{-1}\X_{i\sbt}$ being $n$ independent sub-Gaussian random vectors with finite sub-Gaussian constant $\gamma>0$, zero-mean and second moment $\bI_p$. 
        Let $S$ be any fixed subset of $[d]$ with $|S| = s$.\\
        Then the event 
        \[
            \E_{S,t}:= \left\{
                \left\|\left[D_{\Sigma}^{-1/2}(\wh\Sigma - \Sigma)D_{\Sigma}^{-1/2}\right]_{SS}\right\|_{\rm op} \le  C\|R_{SS}\|_{\op}\left(
		        \sqrt{s \over n}+{s\over n} + \sqrt{t\over n} + {t\over n}
		        \right)
            \right\}
        \]  
        has probability $1-2\exp(-ct)$ for some positive constants $c=c(\gamma)$ and $C= C(\gamma)$ depending on $\gamma$ only.\\
        Furthermore, if $\log p \le c'n$ for sufficiently small constant $c'>0$, we have
        \[
            \|\wh R_{SS} - R_{SS}\|_{\rm op} \le C \|R_{SS}\|_{\rm op}\left(
                \sqrt{s\over n} + {s\over n} + \sqrt{ t\over n} + {t\over n} + \sqrt{\log (n\vee p) \over n}
            \right)
        \]
        on the event $\E\cap \E_{S,t}$ with $\E$ defined in (\ref{event:E}).\\
    \end{lemma}

    An immediate application of Lemma \ref{lem_op_R_SS} with $S = [p]$ in conjunction with $\PP(\E) \ge 1-4(n\vee p)^{-1}$ from Lemma \ref{lem_dev_R} yields the following upper bounds for the operator norm of $\wh R -R$.
    
    \begin{lemma}\label{lem_op_R}
	    Let $\X\in \RR^{n\times d}$ be a random matrix with $\Sigma^{-1}\X_{i\sbt}$ being $n$ independent sub-Gaussian random vectors with finite sub-Gaussian constant $\gamma>0$, zero-mean and second moment $\bI_d$.  
	    Assume $\log d \le c n$ for some sufficiently small constant $c>0$. Then, for all $t\ge 0$, with probability at least $1 - 2 \exp(-c't) - 4(p\vee n)^{-1}$, we have
        \[
            \|\wh R - R\|_{\rm op} \le C \|R\|_{\rm op}\left(
                \sqrt{p\over n} + {p\over n} + \sqrt{t\over n} + {t\over n} + \sqrt{\log n \over n}
            \right)
        \]
        where $C = C(\gamma)$ and $c' = c'(\gamma)$ are positive constants depending only on $\gamma$.\\
    \end{lemma}
    
        \begin{proof}[Proof of Lemma \ref{lem_op_R_SS}]
    We first prove $\PP(\E_{S,t}) \ge 1-2e^{-ct}$ for all $t\ge 0$. First note that
		\begin{align*}
			&\left\|\left[D_{\Sigma}^{-1/2}(\wh\Sigma - \Sigma)D_{\Sigma}^{-1/2}\right]_{SS}\right\|_{\rm op} \\
			& = \left\|\left[D_{\Sigma}^{-1/2}\right]_{SS}\Sigma_{SS}^{1/2}\left[\Sigma_{SS}^{-1/2}(\wh\Sigma_{SS} - \Sigma_{SS})\Sigma_{SS}^{-1/2}\right]\Sigma_{SS}^{1/2}\left[D_{\Sigma}^{-1/2}\right]_{SS}\right\|_{\rm op}\\ 
			&\le \left\|
	    	[D_{\Sigma}^{-1/2}]_{SS}\Sigma_{SS}	[D_{\Sigma}^{-1/2}]_{SS}
			\right\|_{\rm op} \left\|
			\Sigma_{SS}^{-1/2}(\wh\Sigma_{SS} - \Sigma_{SS})\Sigma_{SS}^{-1/2}
			\right\|_{\rm op}\\
			&= \left\|R_{SS}\right\|_{\rm op} \left\|
			\Sigma_{SS}^{-1/2}(\wh\Sigma_{SS} - \Sigma_{SS})\Sigma_{SS}^{-1/2}
			\right\|_{\rm op}.
		\end{align*}
    	Here we used the inequality $\| ABC\|_{\rm op} \le \| AC\|_{\rm op} \|B\|_{\rm op}.$
		Then an application of
		Remark 5.40 in \cite{vershynin_2012} gives 
		\[
			  \left\|\Sigma_{SS}^{-1/2}(\wh\Sigma_{SS} - \Sigma_{SS})\Sigma_{SS}^{-1/2}\right\|_{\rm op} \le  C \left(
		  		\sqrt{|S| \over n} + {|S|\over n} + \sqrt{t\over n} + {t\over n}
		  \right)
		\]
		with probability $1-2e^{-ct}$ for all $t\ge 0$ and some constants $c = c(\gamma)$ and $C = C(\gamma)$. This finishes the proof of the first result.\\

    To prove the second statement, we use the following identity obtained from adding and subtracting terms
		\begin{align*}
			\wh R - R & = D^{-1/2}_{\wh \Sigma}\wh\Sigma D^{-1/2}_{\wh \Sigma} - D^{-1/2}_{ \Sigma}\Sigma D^{-1/2}_{ \Sigma} \\
			&= D^{-1/2}_{ \Sigma}\left(\wh \Sigma - \Sigma\right) D^{-1/2}_{ \Sigma} + \left(
			D^{-1/2}_{\wh \Sigma} - D^{-1/2}_{\Sigma}
			\right)\wh \Sigma D^{-1/2}_{\wh \Sigma}\\
			& \quad + D^{-1/2}_{\Sigma}\wh \Sigma \left(
			D^{-1/2}_{\wh \Sigma} - D^{-1/2}_{\Sigma}
			\right).
		\end{align*}
The first term has been studied above. For the other two terms,  		note that 
		\begin{align*}
			 \left(
			D^{-1/2}_{\wh \Sigma} - D^{-1/2}_{\Sigma}
			\right)\wh \Sigma D^{-1/2}_{\wh \Sigma} &= D_{\Sigma}^{-1/2}\left(D^{1/2}_{\Sigma}-D^{1/2}_{\wh \Sigma}\right) D^{-1/2}_{\wh \Sigma}\wh \Sigma D^{-1/2}_{\wh \Sigma}\\
			&= D_{\Sigma}^{-1/2}\left(D^{1/2}_{\Sigma}-D^{1/2}_{\wh \Sigma}\right) \wh R\\
			& = D_{\Sigma}^{-1/2}\left(D^{1/2}_{\Sigma}-D^{1/2}_{\wh \Sigma}\right) R + D_{\Sigma}^{-1/2}\left(D^{1/2}_{\Sigma}-D^{1/2}_{\wh \Sigma}\right) (\wh R-R)
		\end{align*}
		and, similarly, 
		\begin{align*}
			D^{-1/2}_{\Sigma}\wh \Sigma \left(
			D^{-1/2}_{\wh \Sigma} - D^{-1/2}_{\Sigma}
			\right) &= D^{-1/2}_{\Sigma}\wh \Sigma D_{\Sigma}^{-1/2} \left(
			D^{1/2}_{ \Sigma} - D^{1/2}_{\wh\Sigma}
			\right)D_{\wh \Sigma}^{-1/2}\\
			&= D^{-1/2}_{\Sigma} \Sigma D_{\Sigma}^{-1/2} \left(
			D^{1/2}_{ \Sigma} - D^{1/2}_{\wh\Sigma}
			\right)D_{\wh \Sigma}^{-1/2}\\
			&\quad + D^{-1/2}_{\Sigma}\left(\wh \Sigma - \Sigma\right) D_{\Sigma}^{-1/2} \left(
			D^{1/2}_{ \Sigma} - D^{1/2}_{\wh\Sigma}
			\right)D_{\wh \Sigma}^{-1/2}\\
			&= R \left(
			D^{1/2}_{ \Sigma} - D^{1/2}_{\wh\Sigma}
			\right)D_{\wh \Sigma}^{-1/2}\\
			&\quad + D^{-1/2}_{\Sigma}\left(\wh \Sigma - \Sigma\right) D_{\Sigma}^{-1/2} \left(
			D^{1/2}_{ \Sigma} - D^{1/2}_{\wh\Sigma}
			\right)D_{\wh \Sigma}^{-1/2}.
		\end{align*}
		Plugging these two identities into $\wh R- R$ and writing 
		\[
		    \Delta_S := \left\|
			\left[\left(
			D^{1/2}_{ \Sigma} - D^{1/2}_{\wh\Sigma}
			\right)D_{\Sigma}^{-1/2}\right]_{SS}
			\right\|_{\rm op},\quad \wh\Delta_S := \left\|
			\left[\left(
			D^{1/2}_{ \Sigma} - D^{1/2}_{\wh\Sigma}
			\right)D_{\wh \Sigma}^{-1/2}\right]_{SS}
			\right\|_{\rm op} 
		\]
		give 
		\begin{align*}
			\left\|
				\wh R_{SS} - R_{SS} 
			\right\|_{\rm op}
			&\le \left\|\left[
			D^{-1/2}_{\Sigma}\left(\wh \Sigma - \Sigma\right) D_{\Sigma}^{-1/2} \right]_{SS}
			\right\|_{\rm op}\left(
			1+ \wh \Delta_S
			\right)\\
			&\quad +  \left\|R_{SS}\right\|_{\rm op}\left(\wh \Delta_S + \Delta_S\right) + \Delta_S\left\|\wh R_{SS} - R_{SS}\right\|_{\op}. 
		\end{align*}
		On the event $\E$, observe that 
		\[
		 		\left|
		 			\sqrt{\wh\Sigma_{jj}} - \sqrt{\Sigma_{jj}}
		 		\right| = 
		 		{|\wh\Sigma_{jj} - \Sigma_{jj}| \over \sqrt{\wh\Sigma_{jj}} + \sqrt{\Sigma_{jj}}} \le 
		 		\sqrt{\Sigma_{jj}}\delta_n/2\qquad \forall 1\le j\le p,
		\]
		by invoking Lemma \ref{lem_dev_Sigma}, which, together with $\log p\le c'n$ for sufficiently small $c'$, implies
		\[
			\sqrt{\wh\Sigma_{jj}} \ge \sqrt{\Sigma_{jj}} - \left|
			\sqrt{\wh\Sigma_{jj}} - \sqrt{\Sigma_{jj}}
			\right| \ge c\sqrt{\Sigma_{jj}} 
		\]
		for some $c \in (0,1)$.  We thus obtain, 	on the event $\E$, 
		\begin{align*}
			\Delta_S &=   \max_{1\le j\le p}{|\wh\Sigma_{jj}^{1/2} - \Sigma_{jj}^{1/2}| \over \Sigma_{jj}^{1/2}} \le {\delta_n\over 2};\\
		\wh \Delta_S  &\le   \max_{1\le j\le p}{|\wh\Sigma_{jj}^{1/2} - \Sigma_{jj}^{1/2}| \over c\Sigma_{jj}^{1/2}} \le {\delta_n \over 2c}.
		\end{align*}
	    These two bounds and $\E_{S,t}$ together with the fact that $\log (p\vee n)/n$ is sufficiently small imply the desired result. 
    \end{proof}
    
    \bigskip
    
    \subsection{Discussion of conditions (\ref{cond_signal_ell_q_relaxed}) and (\ref{cond_signal_ell_q_J})}\label{app_dicuss}

    We 
    provide insight into  (\ref{cond_signal_ell_q_relaxed})	and
     (\ref{cond_signal_ell_q_J}), and  discuss when they are likely to  hold. Since the criterion $S_q(i,j)$ itself is not given in terms of the model parameters, to aid intuition, 
     we first establish a  lower bound for 
     $S_q(i,j)$ in terms of the angle 
     \[ \theta_{ij} ~:=~  \angle (A_{i\sbt}, A_{j\sbt})\]
     between the rows $A_{i\sbt}$ and $A_{j\sbt}$
     of $A$ directly. 
     \begin{lemma}\label{lem_lb_score}
     	Under model (\ref{model}) and Assumption \ref{ass_I},
     	we have
     	\begin{align}\label{lb_score} 
     	[S_{q}(i,j)]^2  &\ge c_b c_z c_r  
     	\sin^2(\theta_{ij})
     	\end{align}
     	for any $i\ne j$ and any $q\ge 2$, with
     	\begin{equation}\label{def_cbzr} 
     	c_b:= \min_{1\le i\le p}\|B_{i\sbt}\|_2^2 ,\quad c_z:= \lambda_K(\C) , \quad c_r:= \min_{i\ne j}{1\over p-2}\lambda_K\left(B_{\nij\sbt}\C B_{\nij\sbt}^\T\right).  \\
     	\end{equation}
     \end{lemma}
     The proof of Lemma \ref{lem_lb_score} can be found in Appendix \ref{app_auxiliary_proof}. 
    From Lemma \ref{lem_lb_score}, we see    that
	\begin{equation}\label{cond_angle_rows_A}
	    \sin^2(\theta_{ij})
	    ~ > ~{16 \delta_n^2 \over c_bc_zc_r}, \qquad  \forall{i\not\group j}
	\end{equation}
	implies condition (\ref{cond_signal_ell_q_relaxed}) and (\ref{cond_signal_ell_q_J}).
	Of course, (\ref{cond_angle_rows_A}) is feasible only if 
    \begin{equation}\label{cond_BC}
    c_b c_z c_r >  16  \delta_n^2.
    \end{equation}
    Sufficient signal present in $B$ and $\C$ 
    translate into $c_b$ and $c_z$ being absolute constants bounded away from zero. Precisely, $c_b \ge \epsilon$, for some $\epsilon > 0$, is equivalent to  
    	$$
    	\max_{i\in [p]}{[\Sigma_E]_{ii} \over A_{i\sbt}^\T \C A_{i\sbt}} \le {1-\epsilon \over \epsilon},
    	$$
    a signal ($A_{i\sbt}^\T \C A_{i\sbt}$) to noise ($[\Sigma_E]_{ii}$) ratio condition. 
    The quantity $c_r$ 
    controls the effect that leaving out the $i$-th and $j$-th coordinates of $X$ has on the correlation matrix $R$, via its $K$th largest eigenvalue. Notice that $K$ is the dimension of the low-rank part $B \C B^\T$ of $R$ in its decomposition (\ref{eq_R}).
	 For example,
		suppose that the rows of the submatrix $B_{J\sbt}$ are i.i.d. copies of a sub-Gaussian random vector that has a  covariance matrix with uniformly bounded eigenvalues. 
	In this case, we prove in Appendix \ref{app_proof_lb_cr} that, if $\min\{c_b,c_z\}>c$ for some constant $c>0$ and $K\lesssim |J|$, the inequality 
			$(p-2) c_r  \gtrsim |J| +\min_k |I_k|$ holds with probability at least $1-2\exp(-c'K)$ for some constant $c'>0$. If $|J|\asymp p$, $c_r$ is further bounded from below by some positive constant. 
			
Now, consider that $c_b$, $c_z$ and $c_r$ are lower bounded by some positive constant. We enumerate all three possible cases for $i\not\group j$, and provide lower bounds for $\sin^2(\theta_{ij})$ in light of (\ref{cond_angle_rows_A}).
The first two cases correspond to (\ref{cond_signal_ell_q_relaxed}) while the third one corresponds  to (\ref{cond_signal_ell_q_J}).


	\begin{enumerate}
	    \item [(1)] $i\in I_a$, $j\in I_b$ for any $a,b\in [K]$ with $a\ne b$.\\ In this case,  
        $
        \sin^2(\theta_{ij}) = 1
        $ as the two vectors are orthogonal.
        \item[(2)]   $i\in I_a$, for some $a\in [K]$ and $j\notin I$.\\ This case boils down to be able to separate 
         pure rows of $A$ from non-pure rows of $A$.
         For any $\epsilon>0$,  
        $$A_{ja}^2 \le (1-\epsilon)\|A_{j\sbt}\|_2^2$$ implies that
	    \[
	        \sin^2(\theta_{ij}) = 1 - {A_{ja}^2 \over \|A_{j\sbt}\|_2^2} \ge \epsilon. 
	    \] 
        
        \item[(3)]  $i,j \in J$ and $i\ne j$.\\
        The separation between two non-pure rows of $A$ 
        is hard to verify in general. If the entries of the submatrix $A_{J\sbt}$ are independent $\gamma_A$--sub-Gaussian random variables and $K \ge c''\log |J|  $ for some sufficiently large constant $c''>0$, we prove in Appendix \ref{app_proof_angle_A} that
	    \begin{equation}\label{bd_sintheta_nonpure}
		      \min_{i,j\in J,~ i\ne j}\sin^2(\theta_{ij}) = 1 - \max_{i,j\in J,~ i\ne j}{|A_{i\sbt}^\T A_{j\sbt}|^2 \over A_{i\sbt}^\T A_{i\sbt}A_{j\sbt}^\T A_{j\sbt}} \ge 1- c
		\end{equation}  
		holds
		    with probability at least $1-8|J|^{-c'}$ for some constants $c = c(\gamma_A)$ and $c' = c'(\gamma_A)$.

	\end{enumerate}

	Summarizing, (1) -- (2) suggests the validity of (\ref{cond_signal_ell_q_relaxed}) whereas (3) indicates that (\ref{cond_signal_ell_q_J}) is only likely to hold when $K \gtrsim \log |J|$.\\ 

    \subsection{Auxiliary results and proofs}\label{app_auxiliary_proof}
	\subsubsection{Proof of Lemma \ref{lem_lb_score}}
    By Proposition \ref{prop_score}, we have,  for any $i\ne j$ and $q\ge 2$,
    \begin{align*}
        [S_{q}(i,j)]^2  &\ge [S_{2}(i,j)]^2\\ 
        & = {1\over p-2}\min_{\|v\|_\i = 1}\left\|v^\T B_{\{i,j\}\sbt}\C B_{\nij\sbt}^\T \right\|_2^2\\
        &\ge{1\over p-2}\lambda_K\left(\C^{1/2}B_{\nij\sbt}^\T B_{\nij\sbt}\C^{1/2}\right)\lambda_K(\C)\min_{\|v\|_\i = 1}\left\|v^\T B_{\{i,j\}\sbt} \right\|_2^2
        \\
        &\overset{(i)}{=} c_r c_z \left(\|B_{i\sbt}\|_2^2 \wedge \|B_{j\sbt}\|_2^2\right)\left(
        1 - {|B_{i\sbt}^\T B_{j\sbt}|^2 \over \|B_{i\sbt}\|_2^2 \|B_{j\sbt}\|_2^2}
        \right)
        \\
        &\overset{(ii)}{\ge}   c_r c_z c_b \left|
		    \sin(\angle (A_{i\sbt}, A_{j\sbt})
		    \right|^2.
    \end{align*}
    In $(i)$, we use (\ref{def_cbzr}) and the fact $\lambda_K(LL^\T) = \lambda_K(L^\T L)$ together with 
    \[
        \min_{\|v\|_\i = 1}\left\|v^\T B_{\{i,j\}\sbt} \right\|_2^2 = \left(\|B_{i\sbt}\|_2^2 \wedge \|B_{j\sbt}\|_2^2\right)\left(
        1 - {|B_{i\sbt}^\T B_{j\sbt}|^2 \over \|B_{i\sbt}\|_2^2 \|B_{j\sbt}\|_2^2}
        \right)
    \]
    which can be deduced from Proposition \ref{prop_score_2}. In $(ii)$, we invoke (\ref{def_cbzr}) again and use the fact that $B_{i\sbt} = A_{i\sbt} / \sqrt{\Sigma_{ii}}$. \qed

    \subsubsection{Proof for Remark \ref{rem_one_pure}}\label{app_proof_rem_one_pure}
    Pick any $i\in I_k$ with $k\in [K]$ and $j$ with $A_{j\sbt}$ satisfying (\ref{quasi-scale}). 
    We show $S_2(i,j) \le 2\eps$. Recall $M_{ii} = B_{i\sbt}^\T \C B_{i\sbt}$. Observe from (\ref{score_ell_2}) that 
    \begin{align*}
        [S_{2}(i,j)]^2
        & = {1\over p-2}\min_{\|v\|_\i = 1}\left\|v^\T B_{\{i,j\}\sbt}\C B_{\nij\sbt}^\T \right\|_2^2\\\nonumber
        &\le{1\over p-2}\lambda_1\left(\C^{1/2}B_{\nij\sbt}^\T B_{\nij\sbt}\C^{1/2}\right)\min_{\|v\|_\i = 1}\left\|v^\T B_{\{i,j\}\sbt} \C^{1/2} \right\|_2^2
        \\\nonumber
        &\overset{(i)}{\le }  \left(M_{ii} \wedge  M_{jj}\right)\left(
        1 - {|B_{i\sbt}^\T \C B_{j\sbt}|^2 \over M_{ii}M_{jj}} 
        \right)
        \\
        &\overset{(ii)}{\le}   1 - {(\e_k^\T \C A_{j\sbt})^2 \over A_{j\sbt}^\T \C A_{j\sbt}}
    \end{align*}
    where $(i)$ can be deduced from Proposition \ref{prop_score_2} in conjunction with 
    \[
        \lambda_1\left(\C^{1/2}B_{\nij\sbt}^\T B_{\nij\sbt}\C^{1/2}\right) \le \sum_{\ell \ne i,j}B_{\ell\sbt}^\T \C B_{\ell \sbt} \le p-2.
    \]
    Step $(ii)$ uses $M_{ii}\le 1$ and $B_{i\sbt} = D_{\Sigma}^{-1/2}A_{i\sbt}$.  Since
    \begin{align*}
        {(\e_k^\T \C A_{j\sbt})^2 \over A_{j\sbt}^\T \C A_{j\sbt}} \ge  {(|A_{jk}| - \sum_{a\ne k}|A_{ja}||[\C]_{ak}|)^2 \over \|A_{j\sbt}\|_1^2 \|\C\|_\i} \ge 
        {(|A_{jk}| - \sum_{a\ne k}|A_{ja}|)^2 \over \|A_{j\sbt}\|_1^2} 
    \end{align*}
    by using $\|\C\|_\i = 1$, we conclude 
    \begin{align*}
    1 - {(\e_k^\T \C A_{j\sbt})^2 \over A_{j\sbt}^\T \C A_{j\sbt}} &\le \left(1 + {|A_{jk}| - \sum_{a\ne k}|A_{ja}| \over \|A_{j\sbt}\|_1} \right) \left(1 -{|A_{jk}| - \sum_{a\ne k}|A_{ja}| \over \|A_{j\sbt}\|_1} \right)\\
    & = {4|A_{jk}|\sum_{a\ne k}|A_{ja}| \over \|A_{j\sbt}\|_1^2} \le {4\sum_{a\ne k}|A_{ja}| \over \|A_{j\sbt}\|_1}.
    \end{align*}
    The result follows by invoking (\ref{quasi-scale}).

	\subsubsection{Lower bounds of $c_r$ defined in (\ref{def_cbzr}) when the rows of $B_{J\cdot}$ are i.i.d. sub-Gaussian random vectors. }\label{app_proof_lb_cr}
	
	The following lemma provides lower bounds for $c_r$ when the rows of $B_{J\cdot}$ are i.i.d. sub-Gaussian random vectors. It follows from an application of  \cite[Remark 5.40]{vershynin_2012}. 
	
	\begin{lemma}\label{lem_lb_cr}
	    Suppose the rows of $B_{J\sbt}$ are i.i.d. sub-Gaussian random vectors with sub-Gaussian constant $\gamma_B$ and second moment $\Sigma_B$. Assume $c \le \lambda_K(\Sigma_B) \le \lambda_1(\Sigma_B) \le c'$ and $\min\{c_b, c_z\} \ge c''$ for some constants $c,c',c''>0$. 
	    \item[(1)] If $|J| \ge CK$ for some 
	    sufficiently large $C>0$,  then 
	    \[
	        \PP\left\{(p-2)c_r \gtrsim \min_k|I_k| + |J|
	        \right\} \ge 1-2e^{-cK}
	    \]
	    for some positive constants $c=c(\gamma_B)$;
	        \item[(2)] If $|J| < CK$    and $\min_k|I_k| \asymp |I|/K$, then $(p-2)c_r \gtrsim p/K$ with probability equal to one.
	\end{lemma}
	\begin{proof}
	    We first bound from below $\lambda_K(B\C B^\T)$
	    in two cases: $|J| \ge CK$ and $|J| < CK$ for some sufficiently large constant $C>0$.  
	    When $|J| < CK$, we have
	    \begin{align*}
	       \lambda_K(B\C B^\T) &= \lambda_K(\C^{1/2}B^\T  B \C^{1/2}) = \lambda_K(\C^{1/2}B_{I\sbt}^\T  B_{I\sbt} \C^{1/2})\ge c_zc_b \min_{k}|I_k|.
	    \end{align*}
	    Since $|J| < CK$ and $|I| \ge 2K$, we have $|I|\asymp p$. Thus $\min_k|I_k| \asymp |I|/K$ implies $\min_k|I_k| \asymp p/K$. 
	    
	    When $C K \le |J|$, an application of \cite[Remark 5.40]{vershynin_2012} yields 
	    \[
	        \PP\left\{
	        \left\|
	        {1\over |J|} B_{J\sbt}^\T B_{J\sbt} - \Sigma_B
	        \right\|_{\rm op}^2 \le 
	        c(\gamma_B) \|\Sigma_B\|_{\rm op} \left(\sqrt{K\over |J|} + {K\over |J|}\right)
	        \right\}\ge 1-2e^{-c'(\gamma_B)K}.
	    \]
	    This together with Weyl's inequality implies, with the same probability, 
	    \[
	        {1\over |J|}\lambda_K(B_{J\sbt}^\T B_{J\sbt}) \ge \lambda_K(\Sigma_B) -\left[ c(\gamma_B) \|\Sigma_B\|_{\rm op} \left(\sqrt{K\over |J|} + {K\over |J|}\right)\right]^2 \gtrsim \lambda_K(\Sigma_B) \ge c
	    \]
	    by also using $CK \le |J|$ for a sufficiently large $C$ and the fact that $\Sigma_B$ has bounded eigenvalues. 
	    We then have 
	    \begin{align*}
	        \lambda_K(B\C B^\T) &\ge \lambda_K(\C^{1/2}B_{I\sbt}^\T  B_{I\sbt}\C^{1/2})+ \lambda_K(\C^{1/2}B_{J\sbt}^\T  B_{J\sbt}\C^{1/2})\gtrsim \min_k |I_k| + |J|.
	    \end{align*}
	    
	    Finally, observe that, for any $i\ne j$,
	    \begin{align*}
	        \lambda_K(B\C B^\T) &= \lambda_K(\C^{1/2}B^\T  B \C^{1/2})\\
	        &\le \lambda_K(\C^{1/2}B_{\nij\sbt}^\T  B_{\nij\sbt}\C^{1/2}) + B_{i\sbt}^\T \C B_{i\sbt} + B_{j\sbt}^\T \C B_{j\sbt}\\
	        &\le \lambda_K(B_{\nij\sbt} \C B_{\nij\sbt}^\T) + 2
	    \end{align*}
	    by using $B_{i\sbt}^\T \C B_{i\sbt} \le 1$. We thus conclude 
	    \[
	        c_r \ge {\lambda_K(B \C B^T) - 2\over p-2}.
	    \]
	    This together with the above two lower bounds of $\lambda_K(B \C B^T)$ finishes the proof. 
	\end{proof}

	\subsubsection{Proof of (\ref{bd_sintheta_nonpure})}\label{app_proof_angle_A}
	We give the proof of (\ref{bd_sintheta_nonpure}) when the entries of $A_{J\sbt}$ are i.i.d. sub-Gaussian with sub-Gaussian constant $\gamma_A$. 
	
	\begin{lemma}\label{lem_angle_A}
	    Suppose the entries of $A_{J\sbt}$ are i.i.d. sub-Gaussian random variables with sub-Gaussian constant $\gamma_A$ and second moment $\sigma_A^2$. Assume $\log |J| \le c_0K$ for sufficiently small constant $c_0>0$. Then
	    \[
	        \PP\left\{
	            \max_{i,j\in J,~ i\ne j}{|A_{i\sbt}^\T A_{j\sbt}|^2 \over A_{i\sbt}^\T A_{i\sbt}A_{j\sbt}^\T A_{j\sbt}} \le {c\log |J| \over K}
	        \right\} \ge 1- 8|J|^{-c'}
	    \]
	    for some constants $c,c'>0$.
	\end{lemma}
	\begin{proof}
	   Pick any $i,j \in J$ with $i\ne j$. An application of Corollary F.2 of \cite{bunea2020} with $\X = A_{J\sbt}$, $u=\e_i$, $v = \e_j$ and $\Sigma = \sigma_A^2\bI_K$ yields 
	   \[
	       \left|
	       {1\over K}A_{i\sbt}^\T A_{j\sbt} \right| \le c(\gamma_A)\sigma_A^2\left(\sqrt{t\over K} + {t\over K}\right)
	   \]
	   with probability $1-4e^{-t}$ for any $t\ge 0$. Applying Corollary F.2 of \cite{bunea2020} again yields 
	   \[
	       \left|
	       {1\over K}A_{i\sbt}^\T A_{i\sbt}  - \sigma_A^2\right| \le c(\gamma_A)\sigma_A^2\left(\sqrt{t\over K} + {t\over K}\right)
	   \]
	   with probability $1-4e^{-t}$. Choose $t = c\log|J|$ for some constant $c>0$, taking the union bounds over $i,j\in J$ and using $\log|J| \le c_0 K$, concludes that, with probability $1-8|J|^{-c}$, one has 
	   \[
	       \min_{i\in J} A_{i\sbt}^\T A_{i\sbt} \ge K\sigma_A^2 - c(\gamma_A)\sigma_A^2\left(\sqrt{cK\log|J|} + {c\log|J|}\right) \ge c'K\sigma_A^2,
	   \]
	   \[   
	       \max_{i,j\in J, i\ne j}\left|
	       A_{i\sbt}^\T A_{j\sbt} \right| \le c(\gamma_A)\sigma_A^2\left(\sqrt{c K\log|J|} + c\log|J|\right) \le c''\sigma_A^2\sqrt{K\log|J|}.
	   \]
	   These two displays yield the desired results.
	\end{proof}

	\subsection{Auxiliary lemmas}\label{app_aux_lem}
	The following lemma provides deviation inequalities for the entry-wise difference of the sample covariance matrix and its population level counterpart. 
	
	\begin{lemma}\label{lem_dev_Sigma}
		Under model (\ref{model}) and Assumption \ref{ass_distr}, with probability $1- 4p^{-1}$, one has
		\[
		|\wh \Sigma_{ij} - \Sigma_{ij}| \le  (\delta_n / 2) \sqrt{\Sigma_{ii}\Sigma_{jj}},\qquad \textrm{for all $1\le i, j \le p$.}
		\]
	\end{lemma}
	
	\begin{proof}
		The result follows by choosing $u = e_i$, $v = e_j$ and $t = 3\log p$ in Corollary F.2 of \cite{bunea2020} and taking the union bounds over $1\le i,j\le p$.
	\end{proof}
	
	We have similar deviation bounds for entries of $\wh R - R$. 
	
	\begin{lemma}\label{lem_dev_R}
		Under model (\ref{model}) and Assumption \ref{ass_distr}, with probability $1- 4p^{-1}$, one has
		\begin{equation}\label{bd_R_hat_R}
		\max_{1\le i,j\le p}|\wh R_{ij} - R_{ij}| \le \delta_n.
		\end{equation}
		Furthermore, with the same probability, the  following holds, for any $1\le q\le \i$,
		\begin{align*}
		\max_{1\le i,j\le p, i\ne j} \left\|
		\wh R_{i, \nij} - R_{i, \nij} 
		\right\|_q \le \delta_n(p-2)^{1/q}.
		\end{align*}
	\end{lemma}
	\begin{proof}
		The first result follows from Lemma \ref{lem_dev_Sigma} together with Lemma D.3 in \cite{bunea2020}. The second result follows immediately from noting that 
		\[
		\left\|
		\wh R_{i, \nij} - R_{i, \nij} 
		\right\|_q^q = \sum_{\ell\in [p]\nij} \left|
		\wh R_{i\ell} - R_{i\ell}
		\right|^q.
		\]
		
	\end{proof}

    \section{Supplementary simulation results}\label{app_sims}
    
    We provide additional simulation results in this section. The first section contains the performance of PVS and LOVE for selecting $K$ in various settings considered in Section \ref{sec_sims}. The second section compares the performance of $\wh \H$ from Algorithm \ref{alg_I} and the final $\wh \I$ from Algorithm \ref{alg_whole} for very small $K$. The last section compares the performance of PVS with $q = 2$ and $q = \i$.

    \subsection{Results of PVS and LOVE for selecting $K$}
    Figure \ref{fig_K} depicts the performance of PVS and LOVE for selecting $K$ in various settings considered in Section \ref{sec_sims}, that is, $n = 300$, $p = 500$, $K = 10$, $\eta = 1$, $\alpha = 2.5$ and $\rho_Z = 0.3$. Since both procedures of selecting $\mu$ in Section \ref{sec_cv} for PVS yield almost identical result, we only present the result of choosing the rank $1\le k\le \wh G$. PVS consistently selects $K$ in all cases whereas LOVE tends to under-estimate $K$ when the correlation of $Z$ is higher than $0.6$.

    \begin{figure}[ht]
        \centering
        \begin{tabular}{cc}
             \includegraphics[width = 0.4\textwidth]{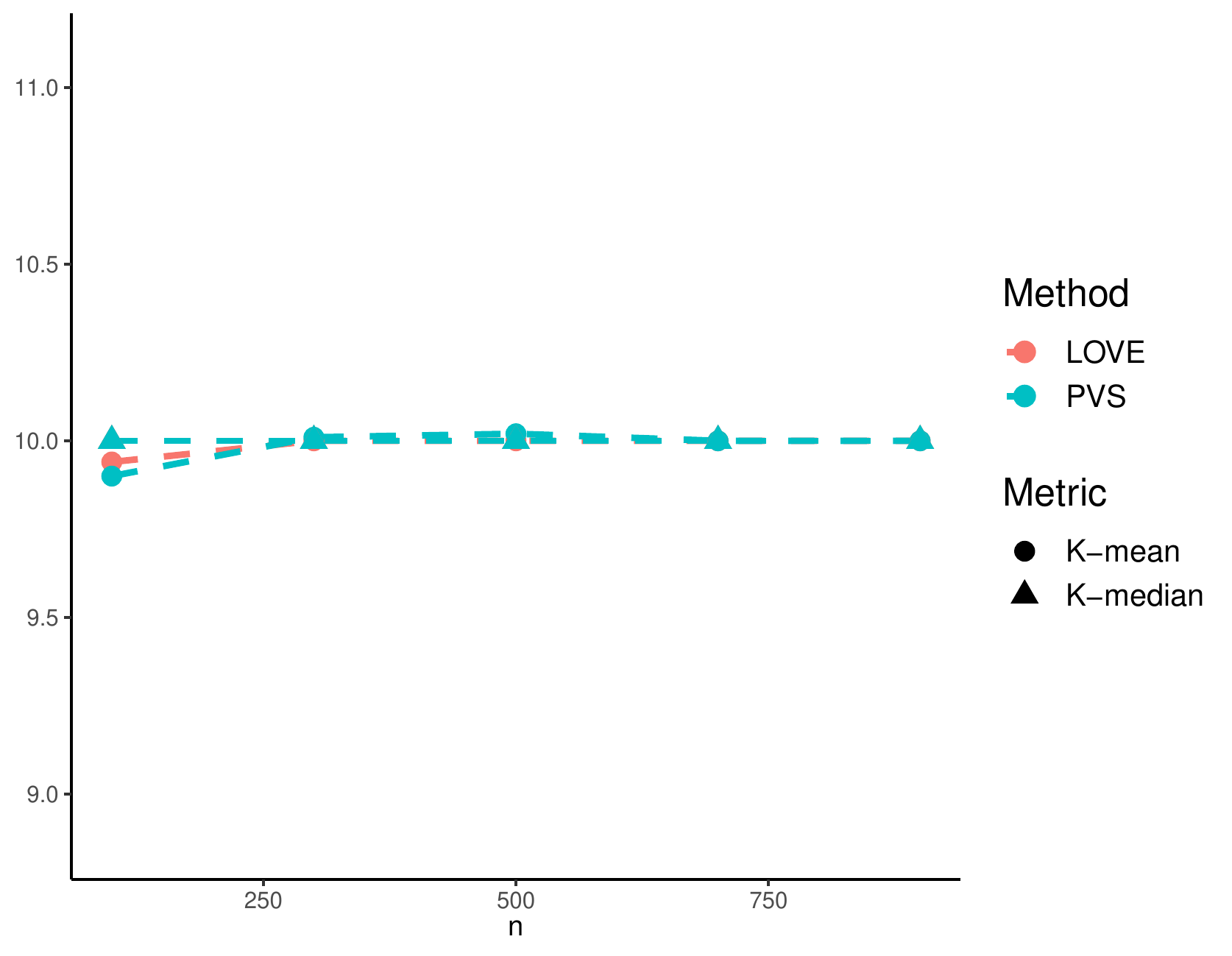} &  
             \includegraphics[width = 0.4\textwidth]{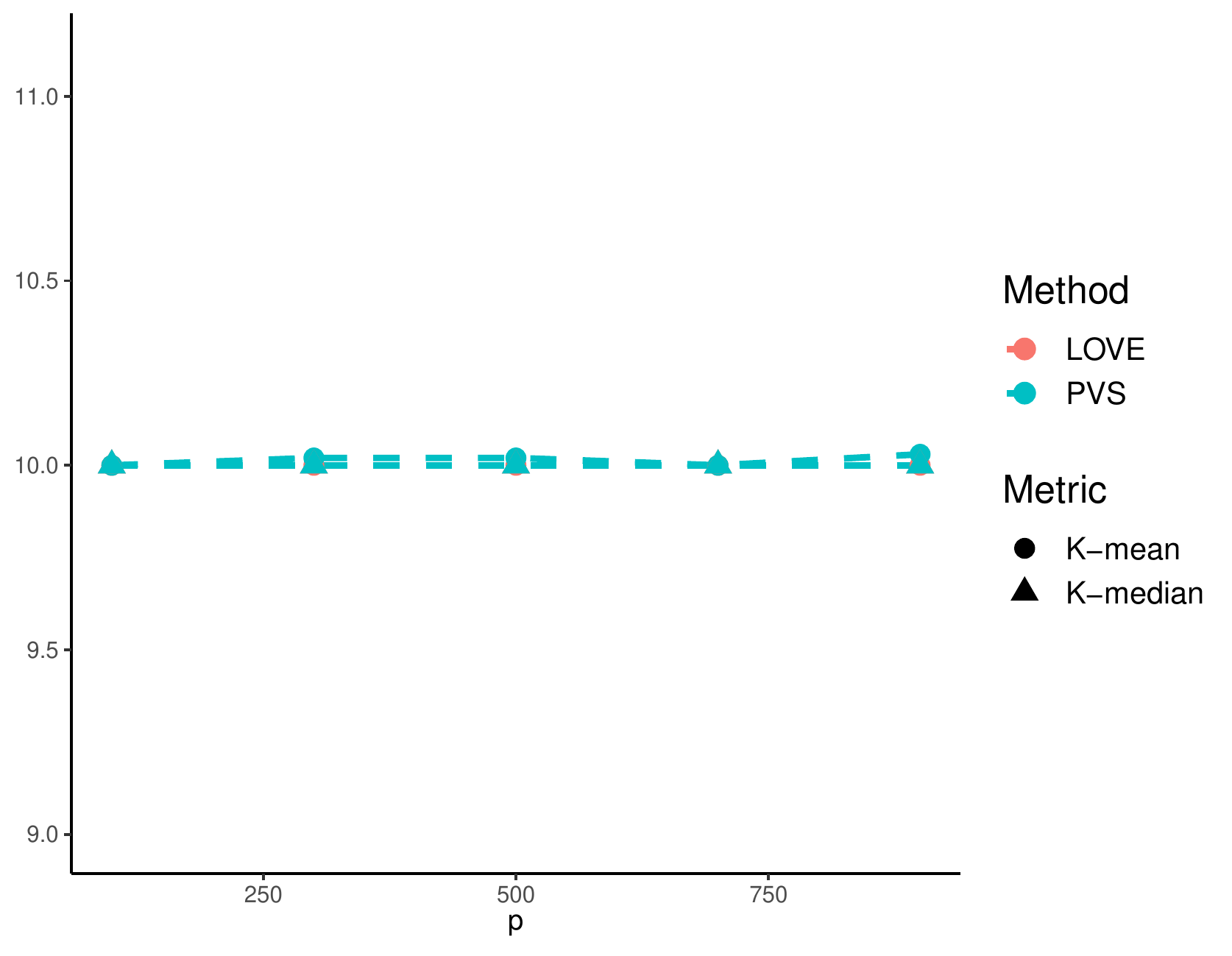} \\
             \includegraphics[width = 0.4\textwidth]{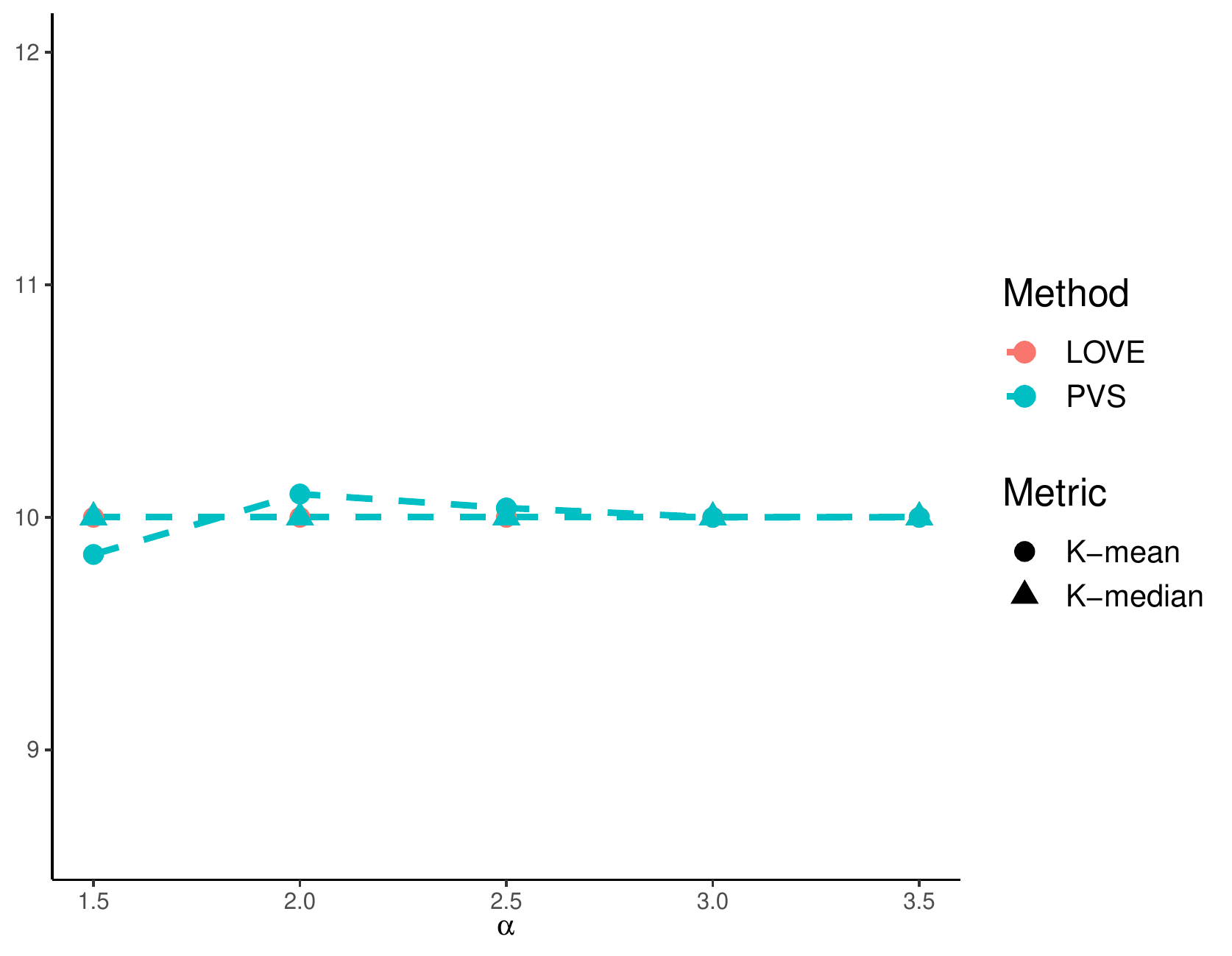} &
             \includegraphics[width = 0.4\textwidth]{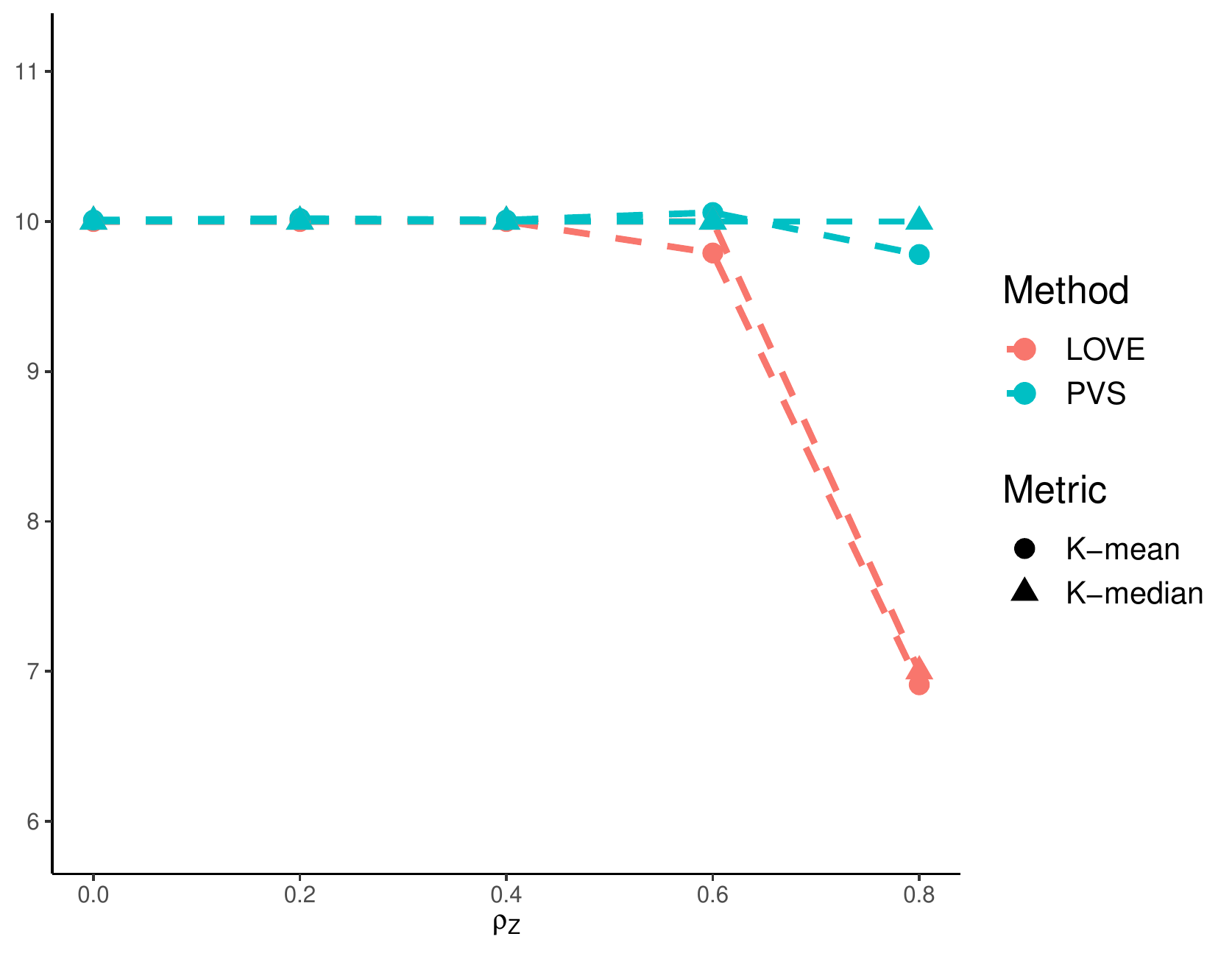}
        \end{tabular}
        \includegraphics[width = 0.4\textwidth]{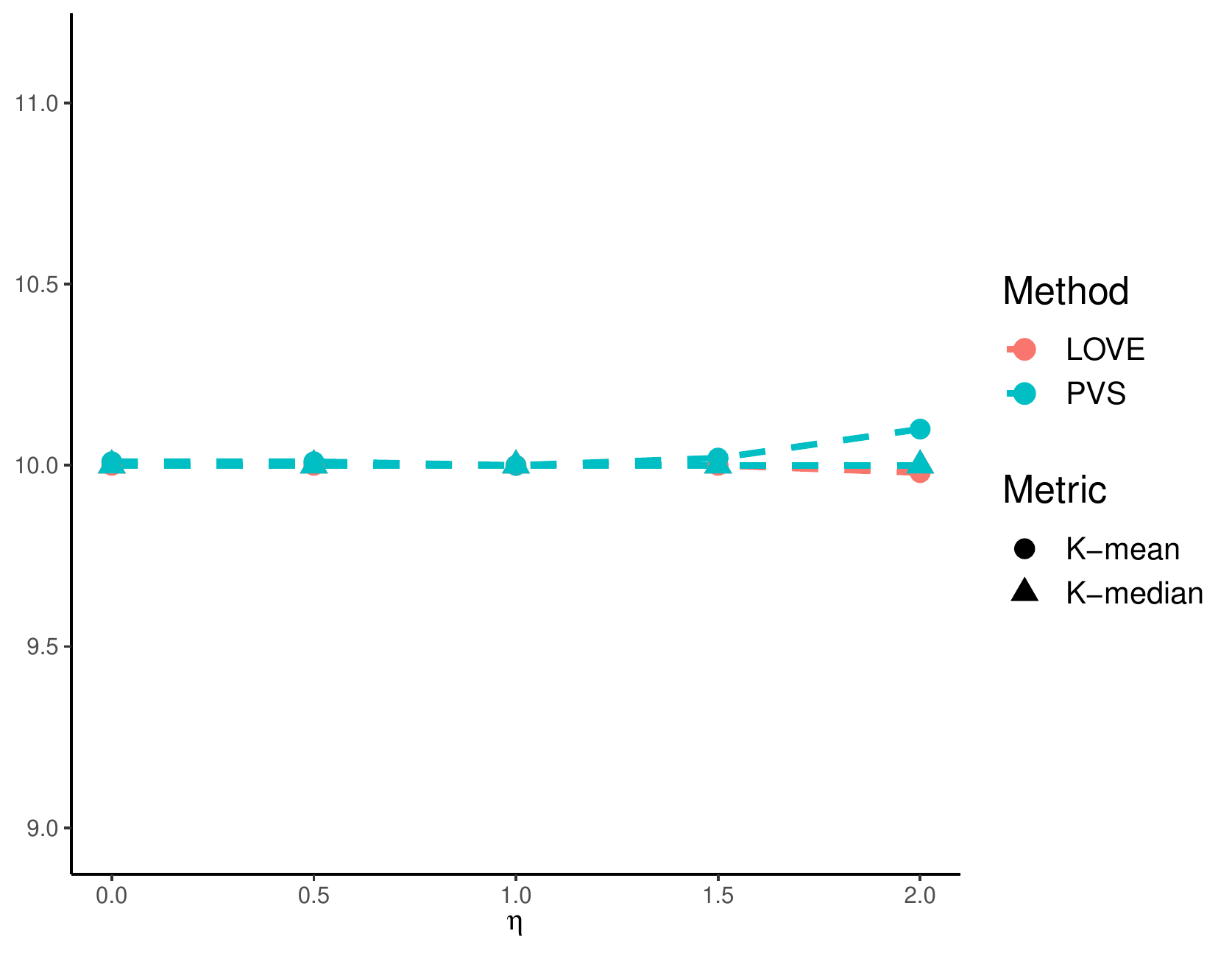}
        \caption{Selection of $K$ of PVS and LOVE for each setting. ``K-mean'' and ``K-median'' are, respectively, the mean and median of the selected $K$ over $100$ repetitions.}
        \label{fig_K}
    \end{figure}

    \subsection{Results of comparing $\wh \H$ and $\wh\I$ when $p = 500$ and $K = 5$}
    As seen in the previous section, both $\wh G$ from Algorithm \ref{alg_I} and $\wh K$ from Algorithm \ref{alg_whole} consistently selects $K$ when $p = 500$ and $K = 10$. In light of Corollary \ref{cor_I} and its following discussions, the consistency of $\wh G$ is as expected when the true $K$ is larger than $\log p$. In this section, we focus on the case where $K$ is smaller than $\log p$ and verify the performance of $\wh G$ and $\wh K$. We set $p = 500$, $K = 5$, $n = 300$, $\eta = 1$, $\rho_Z = 0.3$ and vary $\alpha \in \{1.5, 2, 2.5, 3, 3.5\}$. The performance of $\wh G$ (labelled as PVS-init) and $\wh K$ (labelled as PVS) 
    is shown in Figure \ref{fig_K_small}. Clearly, PVS consistently selects $K$ whereas PVS-init tends to over-estimate $K$ especially when the signal is small ($\alpha < 3$), in line with Corollary \ref{cor_I_prime} and Proposition \ref{prop_K}. On the other hand, the performance of $\wh \I$ (labelled as PVS) and $\wh \H$ (labelled as PVS-init) in Figure \ref{fig_I_K_small} further implies that the pruning step in Algorithm \ref{alg_I_post} successfully removes   the non-pure variables in $\wh \H$. To illustrate the advantage of $\wh I$ over $\wh H$, we also include the metric False Discovery Rate (FDR) defined as 
    \[
        \textrm{FDR} =  {|\wh I \cap J| \over |\wh I|}.
    \]
    Clearly, $\wh I$ from PVS has a lower FDR than $\wh H$ from PVS-init in all settings, in line with our Theorem \ref{thm_I_post}.
    
    \begin{figure}[htb!]
        \centering
        \includegraphics[width = 0.5\textwidth]{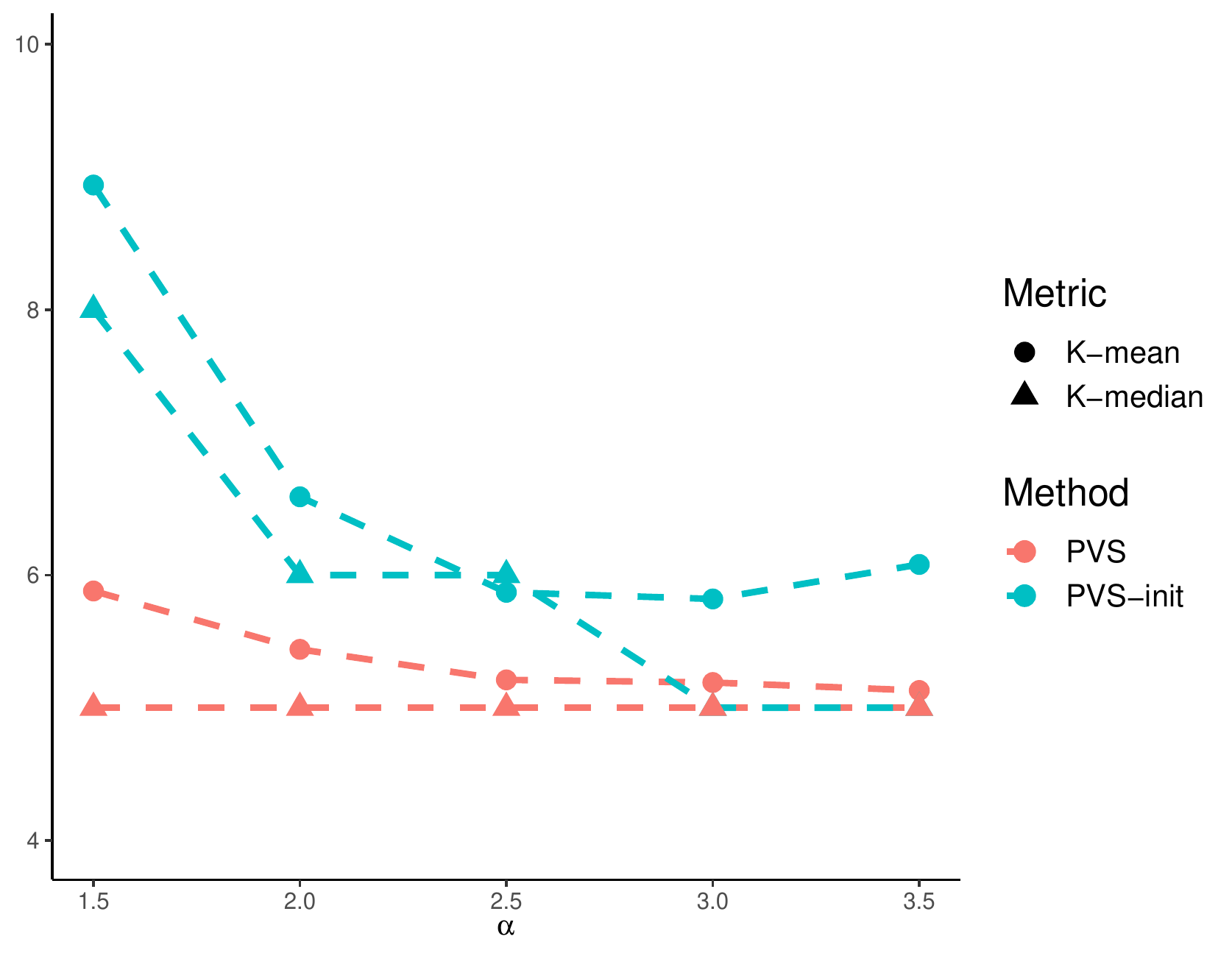}
        \caption{Performance of PVS-K-init and PVS-K}
        \label{fig_K_small}
    \end{figure}

    \begin{figure}[htb!]
        \centering
        \begin{tabular}{cc}
           \includegraphics[width = 0.5\textwidth]{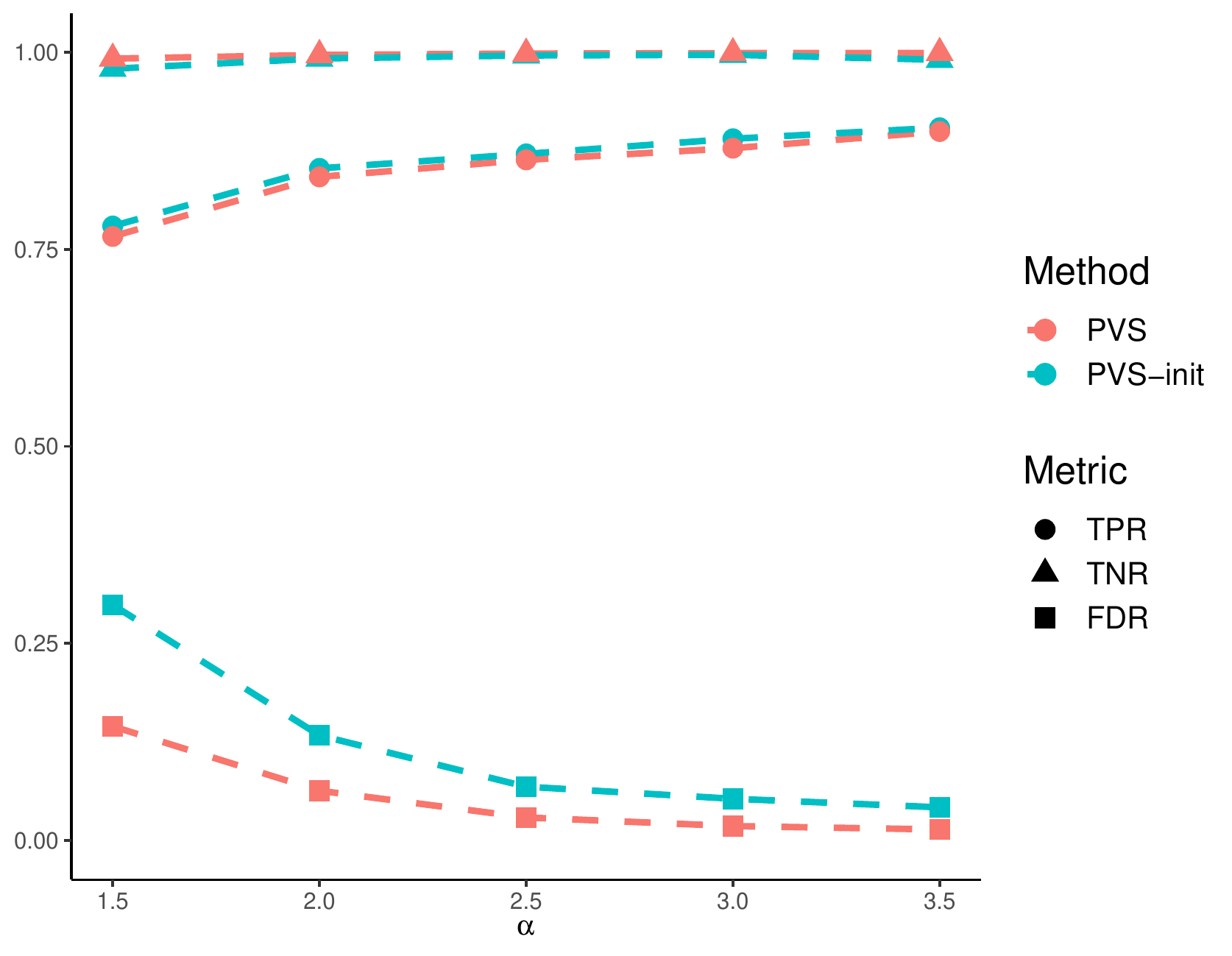} &
        \includegraphics[width = 0.5\textwidth]{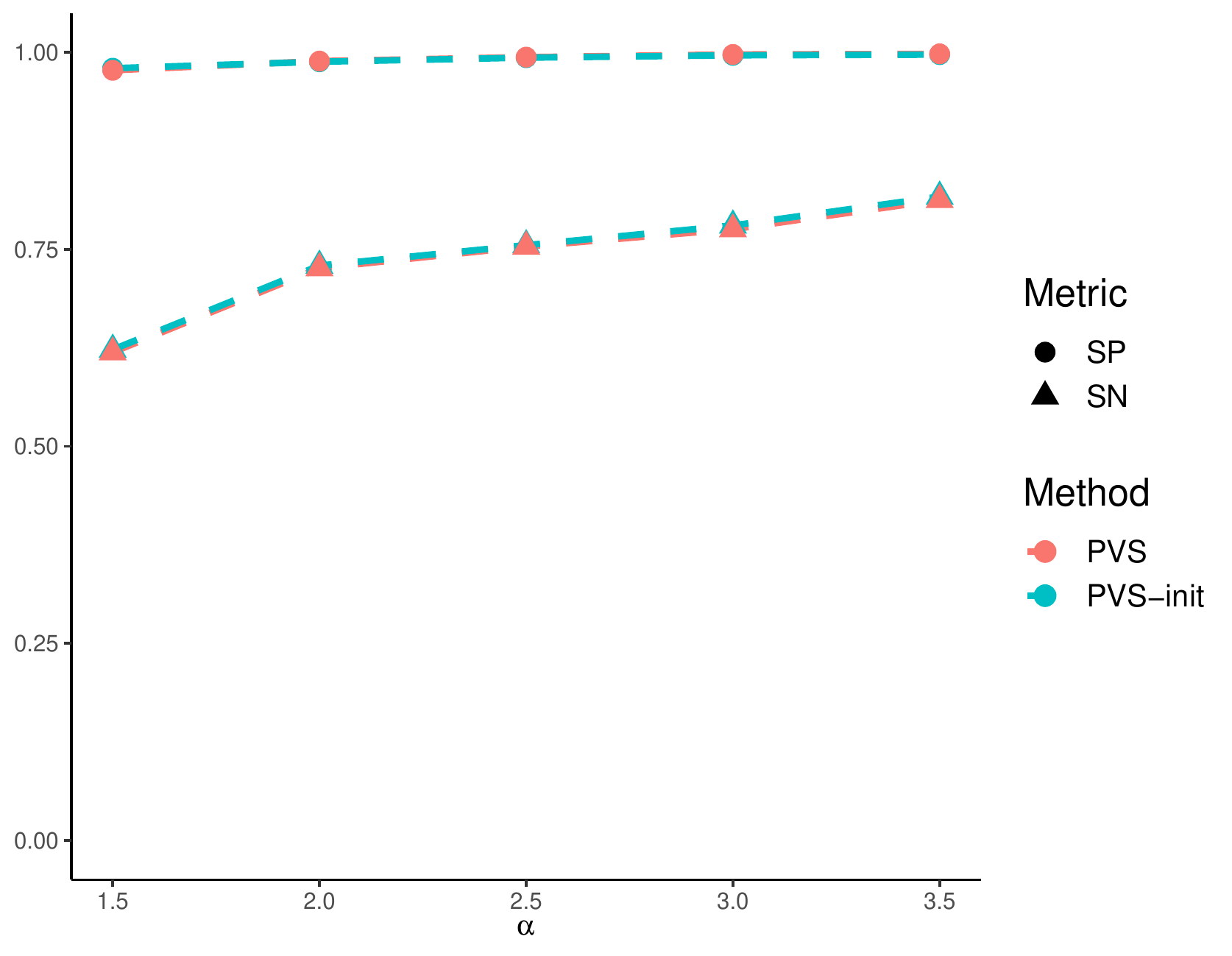}
        \end{tabular}
        \caption{Performance of $\wh \I$ and $\wh \H$}
        \label{fig_I_K_small}
    \end{figure}

    \subsection{Results of PVS by choosing $q = 2$ and $q = \i$ when $p = 100$}
    
    We provide the simulation results of comparing PVS with $q = 2$ (PVS-2) to PVS with $q = \i$ (PVS-Inf) in this section. We consider the general setting $p = 100$, $K = 5$, $\eta = 1$, $\alpha = 2.5$, $\rho_Z = 0.3$ and vary each parameter one at a time. We also include LOVE in the comparison. Figure \ref{fig_low_dim} shows the same findings between PVS-2 and LOVE as Section \ref{sec_sims}. On the other hand, PVS-Inf performs no better than PVS-2 but has much higher computational cost.

    \begin{figure}[htb!]
        \centering
        \begin{tabular}{cc}
             \includegraphics[width = 0.4\textwidth]{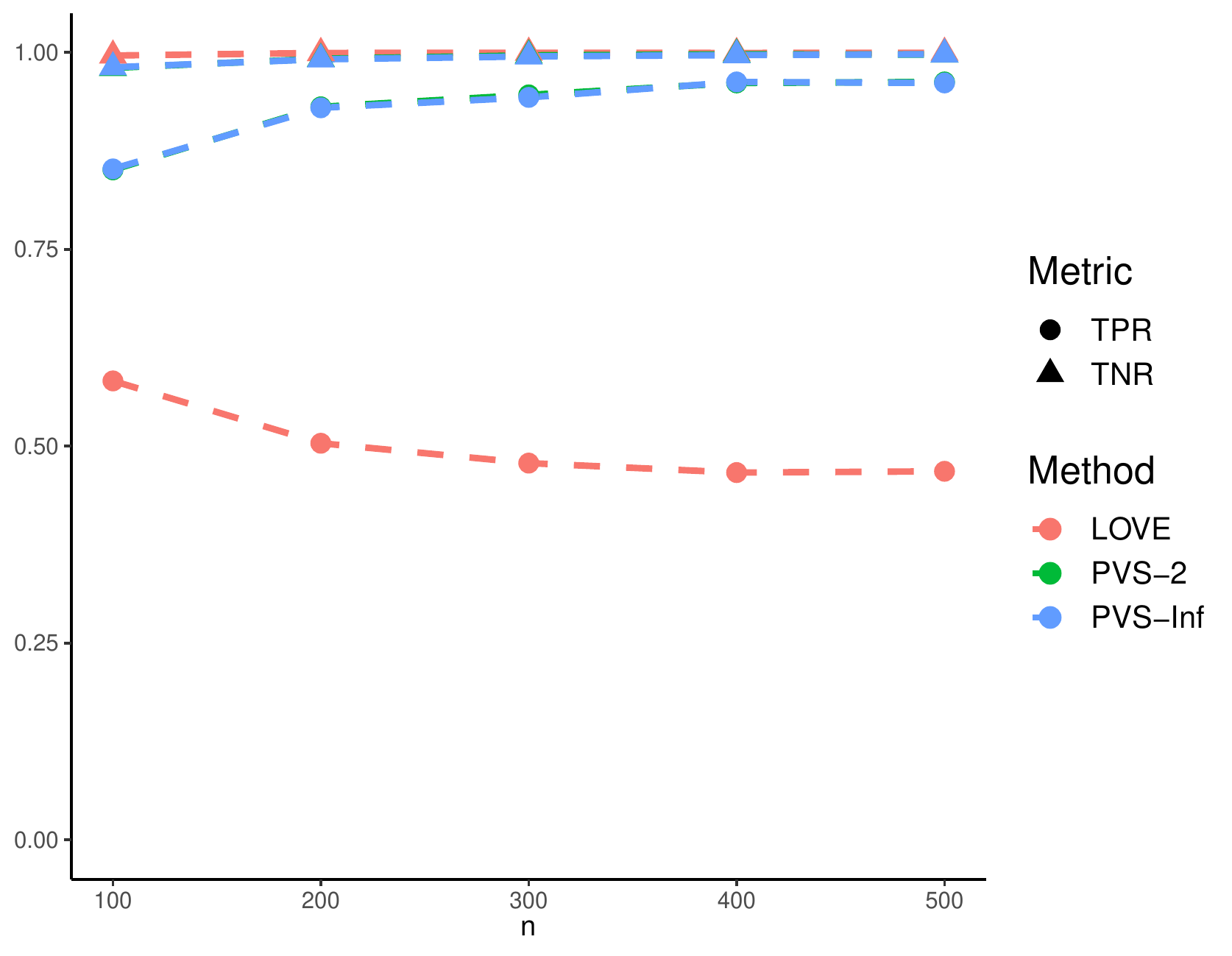} & \includegraphics[width = 0.4\textwidth]{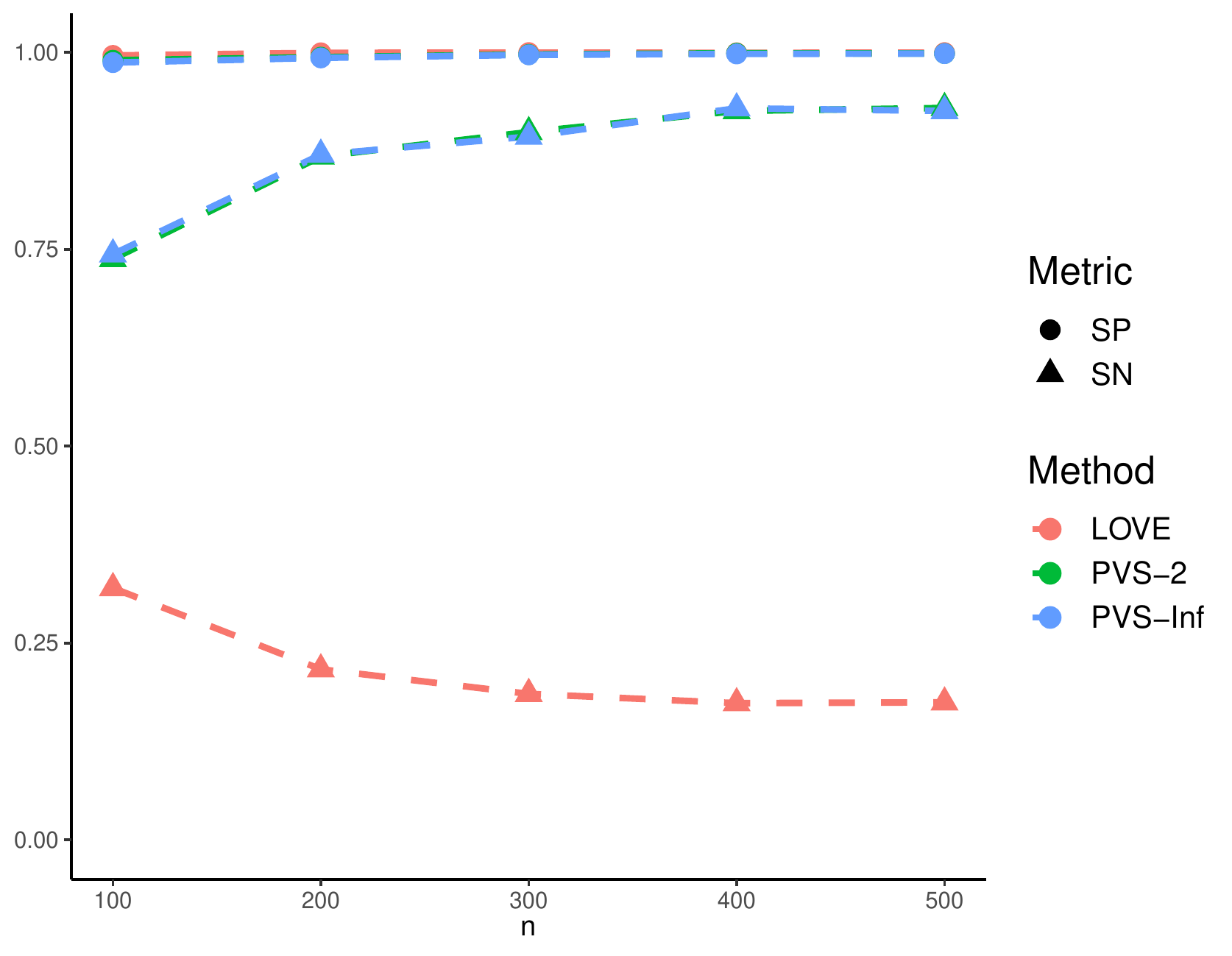}\\ 
             \includegraphics[width = 0.4\textwidth]{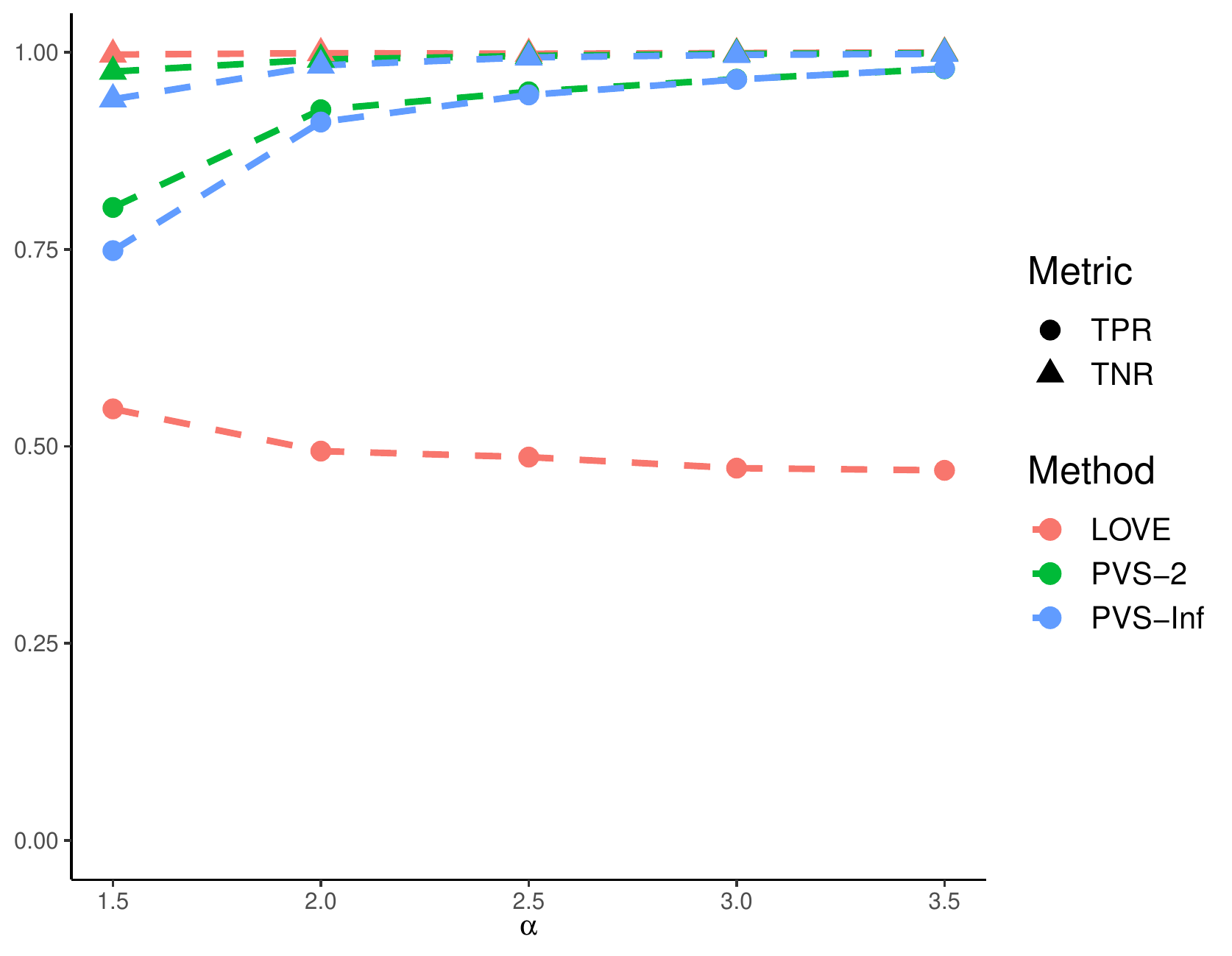} 
             & \includegraphics[width = 0.4\textwidth]{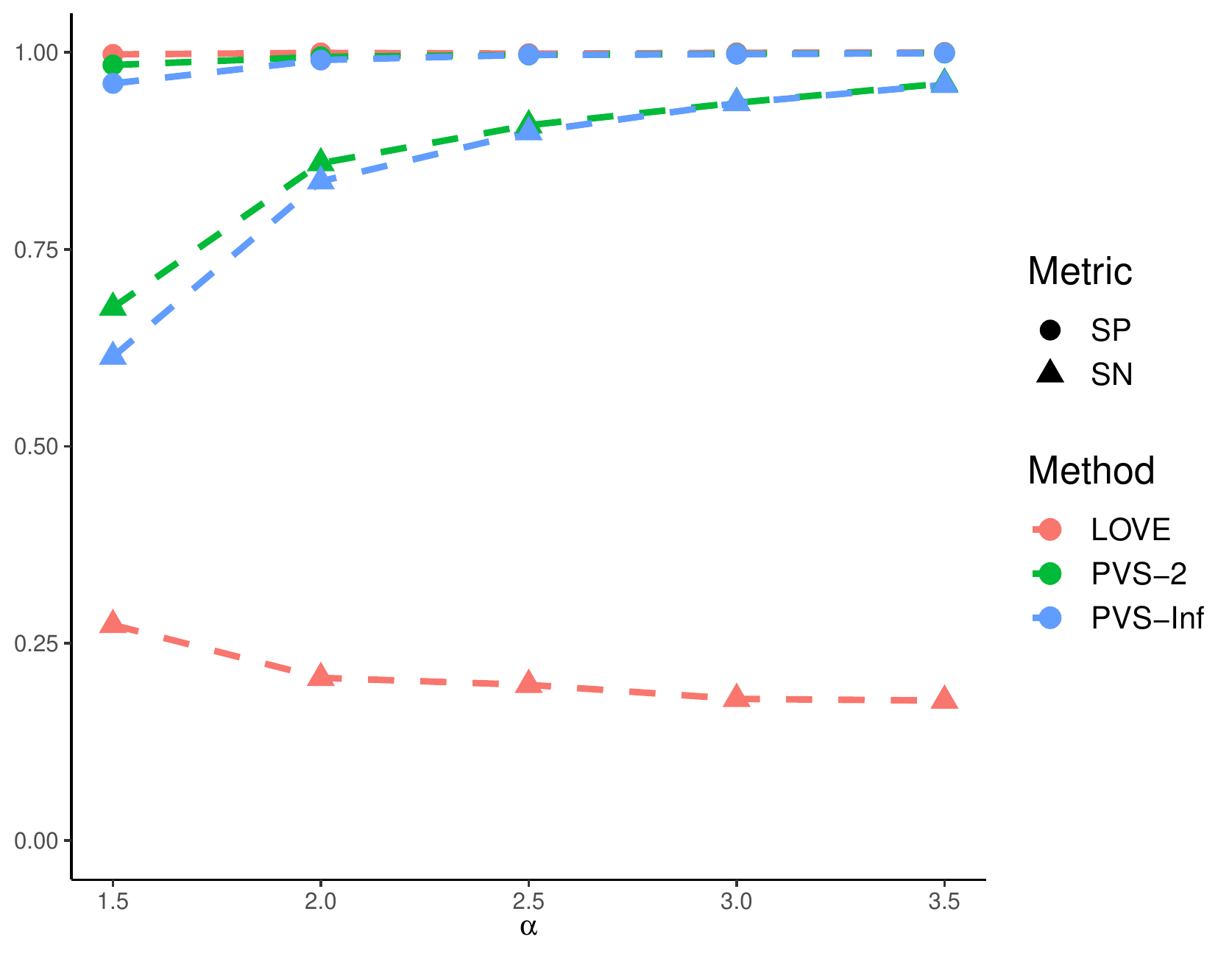}\\
             \includegraphics[width = 0.4\textwidth]{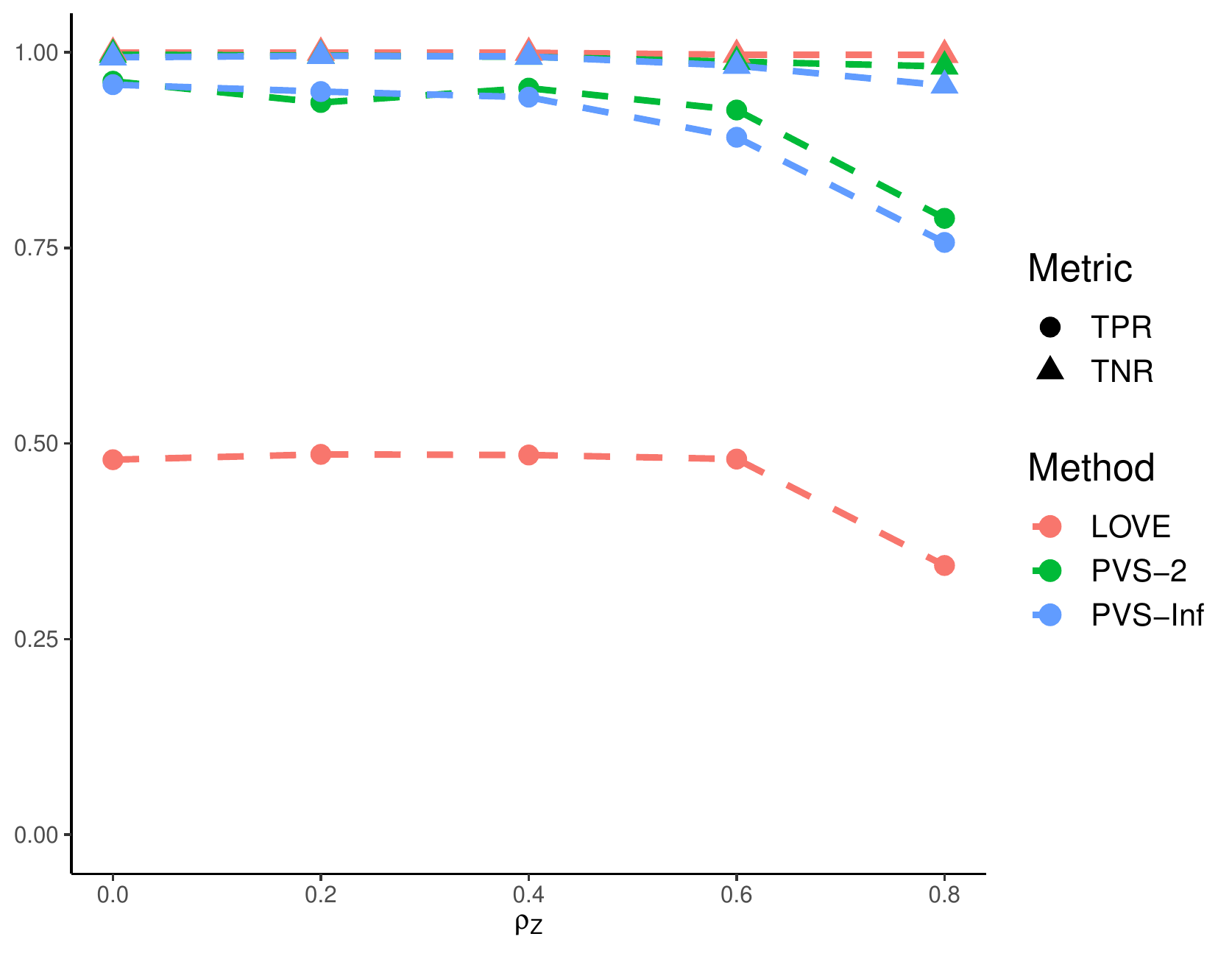} & 
             \includegraphics[width = 0.4\textwidth]{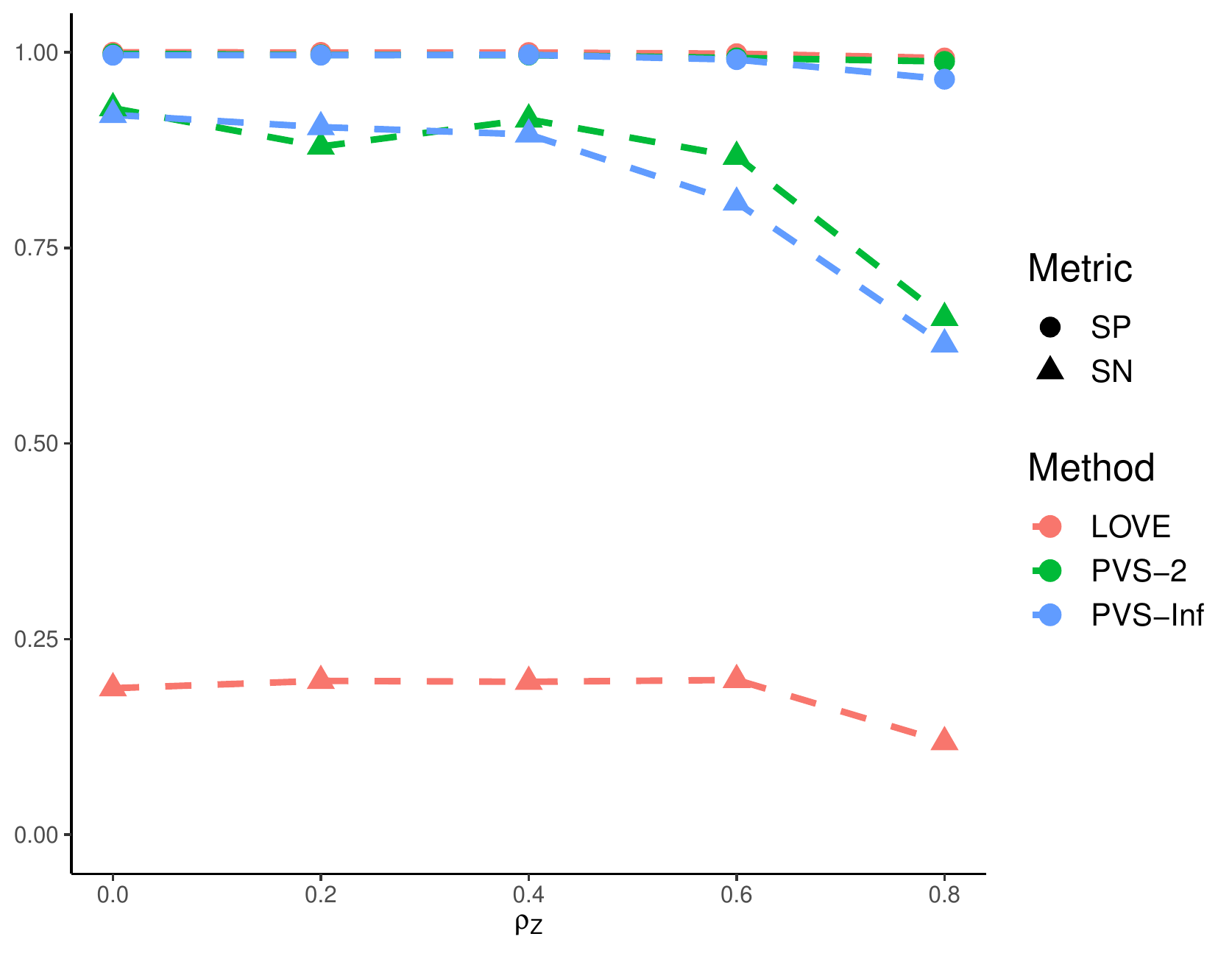}\\
             \includegraphics[width = 0.4\textwidth]{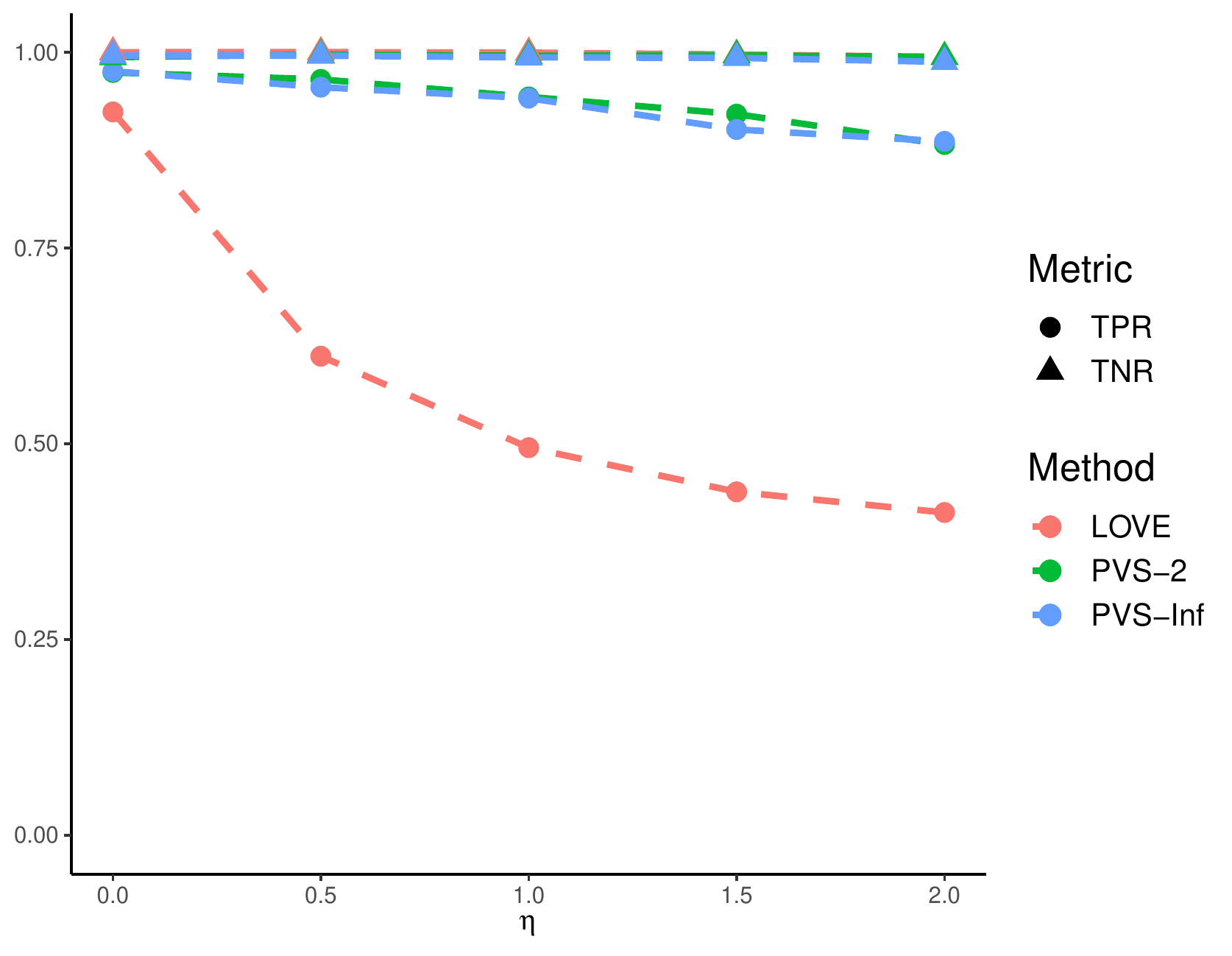} & 
             \includegraphics[width = 0.4\textwidth]{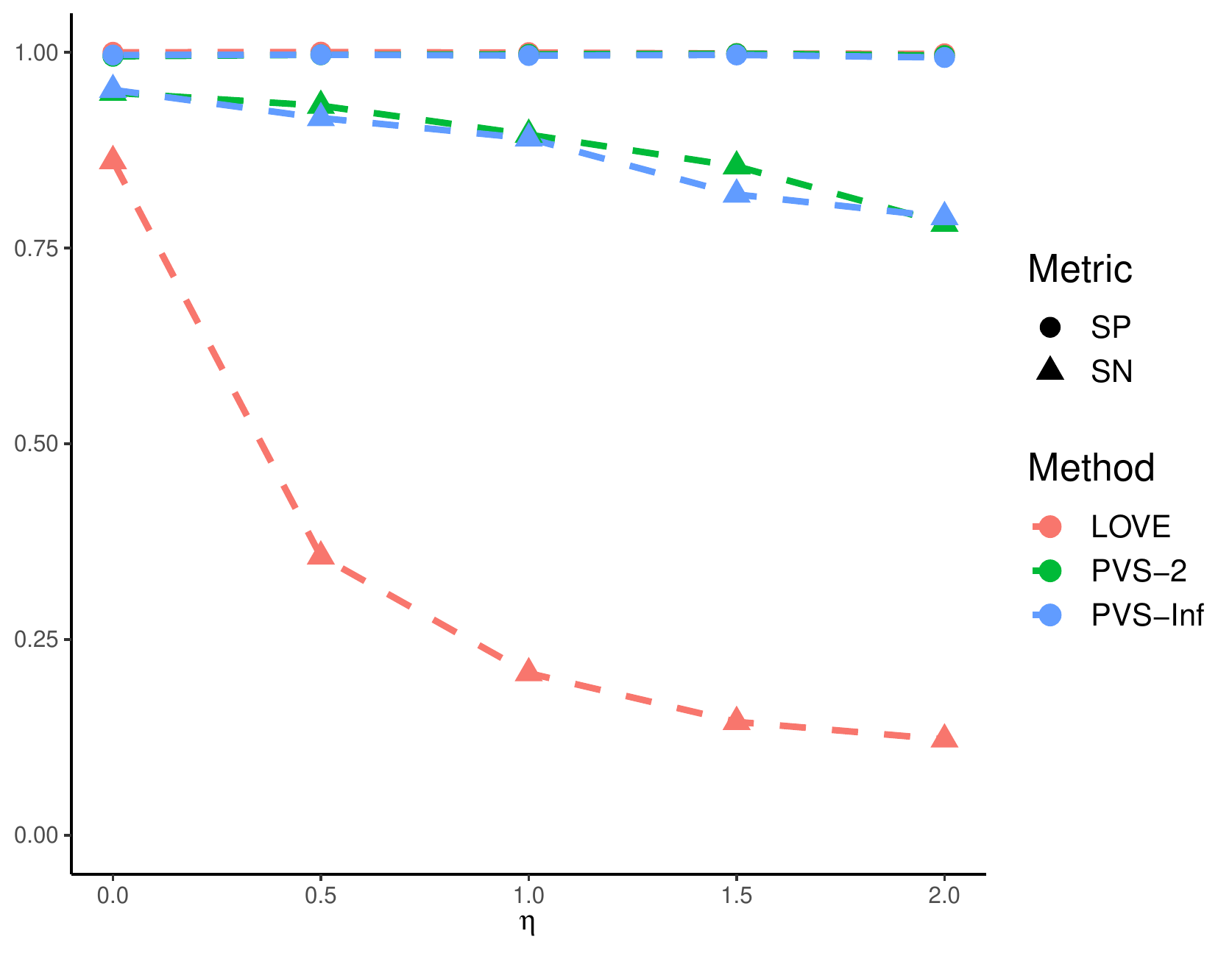}
        \end{tabular}
        \caption{TPR, TNR, SP and SN of PVS-2, PVS-Inf and LOVE}
        \label{fig_low_dim}
    \end{figure}

	\bibliographystyle{ims}
	\bibliography{ref}

\end{document}